\newif\ifabstract
\abstracttrue
\newif\iffull
\ifabstract \fullfalse \else \fulltrue \fi
\documentclass[11pt]{article}
\usepackage{amsfonts}
\usepackage{amssymb}
\usepackage{amstext}
\usepackage{amsmath}
\usepackage{xspace}
\usepackage{theorem}
\usepackage{graphicx}
\usepackage{url}
\usepackage{graphics}
\usepackage{colordvi}
\usepackage{amsfonts}
\usepackage{caption}
\usepackage{subcaption}
\usepackage[usenames,dvipsnames]{xcolor}
\usepackage{verbatim}
\usepackage{color}

\advance \oddsidemargin by -0.8in
\newcommand{\myparskip}{3pt}
\parskip \myparskip

\newenvironment{proofof}[1]{\noindent{\bf Proof of #1.}}%
        {\hspace*{\fill}$\Box$\par\vspace{4mm}}

\newcommand{\cdstar}{c^{**}}
\newcommand{\tc}{\tilde c}

\newcommand{\misr}{{\sf MISR}\xspace}

\newcommand{\ceil}[1]{\ensuremath{\left\lceil#1\right\rceil}}
\newcommand{\floor}[1]{\ensuremath{\left\lfloor#1\right\rfloor}}

\newcommand{\band}{\wedge}

\newcommand{\event}{{\cal{E}}}

\newcommand{\NP}{\mbox{\sf NP}}

\newcommand{\ZPP}{\mbox{\sf ZPP}}

\newcommand{\opt}{\mathsf{OPT}}

\newcommand{\set}[1]{\left\{ #1 \right\}}

\newcommand{\tset}{{\mathcal T}}
\newcommand{\iset}{{\mathcal{I}}}
\newcommand{\pset}{{\mathcal{P}}}

\newcommand{\lset}{{\mathcal{L}}}
\newcommand{\bset}{{\mathcal{B}}}
\newcommand{\aset}{{\mathcal{A}}}
\newcommand{\cset}{{\mathcal{C}}}
\newcommand{\tcset}{\tilde{\mathcal{C}}}

\newcommand{\fset}{{\mathcal{F}}}
\newcommand{\ff}[1]{{\mathbb{F}}^{(#1)}}
\newcommand{\ffi}{{\mathbb{F}}^{(i)}}

\newcommand{\xset}{{\mathcal{X}}}
\newcommand{\wset}{{\mathcal{W}}}

\newcommand{\rset}{{\mathcal{R}}}
\newcommand{\trset}{\tilde{\mathcal{R}}}

\newcommand{\dset}{{\mathcal{D}}}

\newcommand{\hset}{{\mathcal{H}}}
\newcommand{\vset}{{\mathcal{V}}}

\newcommand{\sset}{{\mathcal{S}}}

\newcommand{\be}{\begin{enumerate}}
\newcommand{\ee}{\end{enumerate}}
\newcommand{\bd}{\begin{description}}
\newcommand{\ed}{\end{description}}
\newcommand{\bi}{\begin{itemize}}
\newcommand{\ei}{\end{itemize}}

\newtheorem{theorem}{Theorem}[section]

\newtheorem{lemma}[theorem]{Lemma}
\newtheorem{observation}[theorem]{Observation}
\newtheorem{corollary}[theorem]{Corollary}
\newtheorem{claim}[theorem]{Claim}
\newtheorem{proposition}[theorem]{Proposition}

\newtheorem{definition}{Definition}[section]
\newenvironment{proof}{\par \smallskip{\bf Proof:}}{\hfill\stopproof}
\def\stopproof{\square}
\def\square{\vbox{\hrule height.2pt\hbox{\vrule width.2pt height5pt \kern5pt
\vrule width.2pt} \hrule height.2pt}}

\renewcommand{\phi}{\varphi}
\newcommand{\eps}{\epsilon}

\newcommand{\half}{\ensuremath{\frac{1}{2}}}

\newcommand{\poly}{\operatorname{poly}}

\newcommand{\expect}[2][]{\text{\bf E}_{#1}\left [#2\right]}
\newcommand{\prob}[2][]{\text{\bf Pr}_{#1}\left [#2\right]}

\newenvironment{properties}[2][0]
{
\begin{enumerate} \setcounter{enumi}{#1}}{\end{enumerate}}

\def\MISR{\textsf{MISR}\xspace}
\def\PTAS{\textsf{PTAS}\xspace}
\def\QPTAS{\textsf{QPTAS}\xspace}

\def\etal{et al.\xspace}
\def\poly{\mathrm{poly}}

\begin{document}

\title{On Approximating Maximum Independent Set of Rectangles}
\author{Julia Chuzhoy\thanks{Toyota Technological Institute at Chicago. Email: {\tt cjulia@ttic.edu}. Supported in part by NSF grant CCF-1318242.}\and Alina Ene\thanks{Department of Computer Science and DIMAP, University of Warwick. Email: {\tt A.Ene@dcs.warwick.ac.uk}.}}

\begin{titlepage}

\maketitle
\thispagestyle{empty}

\begin{abstract}
	We study the Maximum Independent Set of Rectangles (\MISR)  problem: given a set of $n$ axis-parallel rectangles, find a largest-cardinality subset of the rectangles, such that no two of them overlap. \MISR is a basic geometric optimization problem with many applications, that has been studied extensively. Until recently, the best approximation algorithm for it achieved an $O(\log \log n)$-approximation factor. In a recent breakthrough, Adamaszek and Wiese provided a \emph{quasi-polynomial} time approximation scheme: a $(1-\eps)$-approximation algorithm with running time $n^{O(\poly(\log n)/\eps)}$. Despite this result, obtaining a \PTAS or even a polynomial-time constant-factor approximation remains a challenging open problem. In this paper we make progress towards this goal by providing an algorithm for \MISR that achieves a $(1 - \eps)$-approximation in time $n^{O(\poly(\log\log{n} / \eps))}$. We introduce several new technical ideas, that we hope will lead to further progress on this and related problems.
\end{abstract}
\end{titlepage}

\section{Introduction}

In the Maximum Independent Set of Rectangles (\MISR) problem, the input is a set $\rset$ of $n$ axis-parallel rectangles, and the goal is to find a maximum-cardinality subset of the rectangles, such that no two of them overlap.
\MISR is a fundamental geometric optimization problem with several applications to map labeling~\cite{AKS98,DF92}, resource allocation~\cite{LNO02}, and data mining~\cite{KMP98, FMMT01,LSW97}. It is also a special case of the classical \textsf{Maximum Independent Set} problem, where the input is an $n$-vertex graph $G$, and the goal is to find a maximum-cardinality subset $S$ of its vertices, so that  no edge of $G$ has both endpoints in $S$. \textsf{Maximum Independent Set} is one of the most fundamental and extensively studied problems in combinatorial optimization. Unfortunately, it is known to be very difficult to approximate: the problem does not have an $n^{1 - \eps}$-approximation algorithm for any constant $\eps$ unless $\NP = \ZPP$~\cite{Hastad01}, and the best current positive result gives an $O(n / \log^2{n})$-approximation algorithm~\cite{BH92}. It is therefore natural to focus on important classes of special cases of the problem, where better approximation guarantees may be achievable. This direction has proved to be especially fruitful for instances stemming from geometric objects in the plane. Results in this area range from Polynomial-Time Approximation Schemes (\PTAS) for ``fat objects", such as disks and squares~\cite{ErlebachJS05}, to an $n^{\eps}$-approximation for arbitrary geometric shapes \cite{FoxP11}. Unfortunately, the techniques used in algorithms for fat geometric objects seem to break down for other geometric shapes. Rectangles are among the simplest shapes that are not fat, which puts \MISR close to the boundary of the class of geometric problems for which \PTAS is achievable with current techniques.

\MISR is a basic geometric variant of Independent Set, and rectangles seem to be among the simplest shapes that capture several of the key difficulties associated with objects that are not fat. It is then not surprising that \MISR has attracted a considerable amount of interest from various research communities. Since the problem is known to be
\NP-hard~\cite{FPT81, IA83}, the main focus has been on designing approximation algorithms. Several groups of researches have
independently suggested $O(\log n)$-approximation algorithms for \MISR~\cite{AKS98,KMP98,Nielsen00}, and Berman \etal~\cite{BermanDMR01} showed that there is a $\ceil{\log_k{n}}$ approximation for any fixed $k$. More recently an $O(\log \log{n})$-approximation was shown~\cite{ChalermsookC09}, that remains the best current approximation algorithm that runs in polynomial time. The result of \cite{ChalermsookC09} also gives a matching upper bound on the integrality gap of the natural LP relaxation for \MISR. The best current lower bound on the integrality gap of this LP relaxation is a small constant, and understanding this gap is a long-standing open question with a beautiful connection to rectangle coloring; see \cite{Chalermsook11} and references therein.
In a recent breakthrough, Adamaszek and Wiese~\cite{AdamaszekW13} designed a Quasi-Polynomial Time Approximation Scheme (\QPTAS) for \MISR: namely, a $(1 - \eps)$-approximation algorithm with running time $n^{O(\poly(\log{n}/\eps))}$, using completely different techniques. Their result can be seen as a significant evidence that \MISR may admit a \PTAS. However, obtaining a \PTAS, or even an efficient constant-factor approximation remains elusive for now.

In this paper, we make progress towards this goal, by providing an algorithm for \MISR that achieves a $(1 - \eps)$-approximation and runs in time $n^{O\left ((\log\log{n} / \eps)^4\right)}$. We introduce several new technical ideas that we hope will lead to further progress on this and related problems.

The \MISR problem seems to be central to understanding several other geometric problems. The work of~\cite{AdamaszekW13} has been very influential, and has lead to several new results, including, for example, a \QPTAS for Maximum Independent Set of Polygons~\cite{AW2,Har-Peled14}, and \QPTAS for several geometric Set Cover problems~\cite{geom-SC}.

\paragraph{Other related work}
Several important special cases of \MISR have been studied extensively. In particular, there is a {\PTAS} for squares --- and more generally, rectangles with bounded aspect ratio \cite{ErlebachJS05} --- and large rectangles whose width or height is within a constant factor of the size of the bounding box that encloses the entire input \cite{AdamaszekW13}. %Har-Peled \cite{Har-Peled14} extended the divide and conquer approach of \cite{AdamaszekW13,AdamaszekW14} to handle arbitrary polygons. Perhaps surprisingly, Mustafa \etal \cite{MustafaRR14} showed that the techniques of \cite{AdamaszekW13} can also be applied to the geometric \textsf{Set Cover} problem.
We note that a more general weighted version of the \MISR problem has also been considered, where all input rectangles are associated with non-negative weights, and the goal is to find a maximum-weight subset of non-overlapping rectangles. As mentioned earlier, there are several algorithms for \MISR that achieve an $O(\log n)$-approximation, and these results hold in the weighted setting as well. The long-standing $O(\log n)$-approximation was improved in the work of Chan and Har-Peled that achieved an $O(\log n / \log \log n)$-approximation for the weighted problem \cite{ChanH12}. This result remains the best polynomial-time approximation algorithm for the weighted problem, as the $O(\log\log n)$-approximation algorithm of~\cite{ChalermsookC09} only applies to the unweighted version of \MISR. The work of Adamaszek and Wiese~\cite{AdamaszekW13} extends to the weighted version and provides a \QPTAS for it as well. There seem to be several technical difficulties in extending our results to the weighted version of the problem, and we leave this as an open problem.

\paragraph{Our Results and Techniques.}
Our main result is summarized in the following theorem.

\begin{theorem}\label{thm: main}
	There is an algorithm for the \misr problem that, given any set $\rset$ of $n$ axis-parallel rectangles and a parameter $0<\eps<1$, computes a $(1 - \eps)$\footnote{So far we have followed the convention that approximation factors of algorithms are greater than 1, but for our QPTAS it is more convenient for us to switch to approximation factors of the form $(1-\eps)$.}-approximate solution to instance $\rset$, in time $n^{O((\log\log{n}/\eps)^4)}$.
\end{theorem}

In order to put our techniques in context, we first give a high-level overview of the approach of Adamaszek and Wiese \cite{AdamaszekW13}. The description here is somewhat over-simplified, and is different from the description of~\cite{AdamaszekW13}, though the algorithm is essentially the same. Their approach is based on dynamic programming, and uses the divide-and-conquer paradigm. Starting with the initial set of rectangles, the algorithm recursively partitions the input into smaller sub-instances. A key insight is the use of closed polygonal curves to partition the instances: given such a curve, one can discard the rectangles that intersect the curve; the remaining rectangles can be naturally partitioned into two sub-instances, one containing the rectangles lying in the interior of the curve and the other containing rectangles lying outside the curve\footnote{The sub-instances could have holes. In order to simplify the exposition, instead of working with instances with holes, we introduce what we call ``fake rectangles" and use them to partition the instance.}. Adamaszek and Wiese  show that for every set $\rset^*$ of non-overlapping rectangles and an integral parameter $L$, there is a closed polygonal curve $C$, whose edges are parallel to the axes, so that $C$ has at most $L$ corners; the number of rectangles of $\rset^*$ intersecting $C$ is at most $O(|\rset^*|/L)$; and at most a $3/4$-fraction of the rectangles of $\rset^*$ lie on either side of the curve $C$. We call such a curve $C$ a \emph{balanced $L$-corner partitioning curve} for set $\rset^*$. 

Given any subset $\rset'\subseteq\rset$ of rectangles, we denote by $\opt(\rset')$ the optimal solution to instance $\rset'$, and we denote $\opt=\opt(\rset)$. Throughout this exposition, all polygons and polygonal curves have all their edges parallel to the axes. We sometimes refer to the number of corners of a polygon as its boundary complexity.

The approach of \cite{AdamaszekW13} can now be described as follows. Let $L=\Theta(\log n/\eps)$ and $L^*=\Theta(L\cdot \log n)$. The algorithm uses dynamic programming. Every entry of the dynamic programming table $T$ corresponds to a polygon $P$ that has at most $L^*$ corners. The entry $T[P]$ will contain an approximate solution to instance $\rset(P)$, that consists of all rectangles $R\in \rset$ with $R\subseteq P$. We say that $P$ defines a basic instance if $|\opt(\rset(P))|\leq \log n$. We can check whether $P$ defines a basic instance, and if so, find an optimal solution for it in time $n^{O(\log n)}$ via exhaustive search.
In order to compute the entry $T[P]$ where $\rset(P)$ is a non-basic instance,  we go over all pairs $P',P''\subsetneq P$ of polygons with $P'\cap P''=\emptyset$, such that $P'$ and $P''$ have at most $L^*$ corners each, and we let $T[P]$ contain the best solution $T[P']\cup T[P'']$ among all such possible pairs of polygons. 
In order to analyze the approximation factor achieved by this algorithm, we build a partitioning tree, that will simulate an idealized run of the dynamic program. Every vertex $v$ of the partitioning tree is associated with some polygon $P(v)$ that has at most $L^*$ corners, and stores some solution to instance $\rset(v)=\rset(P(v))$, consisting of all rectangles $R\in \rset$ with  $R\subseteq P(v)$. For the root vertex of the tree, the corresponding polygon $P$ is the bounding box of our instance. Given any leaf vertex $v$ of the current tree, such that the instance $\rset(v)$ is non-basic, we add two children $v',v''$ to $v$, whose corresponding polygons $P'$ and $P''$ are obtained by partitioning $P$ with the balanced $L$-corner partitioning curve $C$ for the set $\opt(\rset(v))$ of rectangles. We terminate the algorithm when for every leaf vertex $v$, $\rset(v)$ is a basic instance. It is easy to verify that the height of the resulting tree is $O(\log n)$. The polygon associated with the root vertex of the tree has $4$ corners, and for every $1\leq i\leq \log n$, the polygons associated with the vertices lying at distance exactly $i$ from the root have at most $4+iL$ corners. Therefore, every polygon associated with the vertices of the tree has at most $L^*$ corners, and corresponds to some entry of the dynamic programming table. Once the tree is constructed, we compute  solutions to sub-instances associated with its vertices, as follows. For every leaf $v$ of the tree, the solution associated with $v$ is the optimal solution to instance $\rset(v)$; for an inner vertex $v$ of the tree with children $v'$ and $v''$, the solution associated with $v$ is the union of the solutions associated with $v'$ and $v''$. Let $\rset'$ be the solution to the \MISR problem associated with the root vertex of the tree. Then it is easy to see that the solution computed by the dynamic programming algorithm has value at least $|\rset'|$. Moreover, from our choice of parameters, $|\rset'|\geq |\opt(\rset)|(1-\eps)$. This is since for every inner vertex $v$ of the tree, with children $v'$ and $v''$, the loss incurred by the partitioning procedure at $v$, $\lambda(v)=|\opt(\rset(v))|-|\opt(\rset(v'))|-|\opt(\rset(v''))|\leq  |\opt(\rset(v))|/L\leq \eps|\opt(\rset(v))|/\log n$. It is then easy to verify that the total loss of all vertices that lie within distance exactly $i$ from the root, for any fixed $0\leq i\leq \log n$, is at most $\eps |\opt(\rset)|/\log n$, and the total loss of all vertices is at most $\eps|\opt(\rset)|$. It is also immediate to verify that the value of the solution stored at the root vertex of the tree is at least $|\opt(\rset)|-\sum_v\lambda(v)$, and so we obtain a $(1-\eps)$-approximation.

In order to bound the running time of the algorithm, it is not hard to show by a  standard transformation to the problem input, that it is enough to consider polygons $P$ whose corners have integral coordinates between $1$ and $2n$. The number of entries of the dynamic programming table, and the running time of the algorithm, are then bounded by $n^{O(L^*)}=n^{O(\log^2n/\eps)}$.
As a warmup, we show that this running time can be improved to $n^{O(\log n/\eps)}$. The idea is that, instead of computing a balanced $L$-corner partition of the set $\opt(\rset(P))$ of rectangles, we can compute a different partition that reduces the boundary complexities of the two resulting polygons. If $P$ has boundary complexity greater than $L$, then we can compute a polygonal curve $C$,  that partitions $P$ into polygons $P'$ and $P''$, such that the number of corners of each of the two polygons $P'$ and $P''$ is smaller than the number of corners of $P$ by a constant factor, and $|\opt(\rset(P))|-|\opt(\rset(P'))|-|\opt(\rset(P''))| \leq f(L) \cdot |\opt(\rset(P))|$, where $f(L)=O(1/L)$. In our partitioning tree we can then alternate between computing balanced $L$-corner curves, and computing partitions that reduce the boundary complexities of the polygons, so that the number of corners of the polygons does not accumulate as we go down the tree. This allows us to set $L^*=L=\Theta(\log n/\eps)$, and obtain a running time of $n^{O(\log n/\eps)}$.

The bottleneck in the running time of the above algorithm is the number of entries in the dynamic programming table, which is $n^{O(L^*)}$, where $L^*$ is the number of corners that we allow for our polygons, and the term $n$ appears since there are $\Theta(n^2)$ choices for each such corner.
In order to improve the running time, it is natural to try one of the following two approaches: either (i) decrease the parameter $L^*$, or (ii) restrict the number of options for choosing each corner. The latter approach can be, for example, implemented by discretization: we can construct an $(N\times N)$-grid $G$, where $N$ is small enough. We say that a polygon $P$ is \emph{aligned with $G$} if all corners of $P$ are also vertices of $G$. We can then restrict the polygons we consider to the ones that are aligned with $G$. Unfortunately, neither of these approaches works directly.
For the first approach, since the depth of the partitioning tree is $\Theta(\log n)$, we can only afford to lose an $O(\eps/\log n)$-fraction of rectangles from the optimal solution in every iteration, that is, on average, for an inner vertex $v$ of the tree $\lambda(v)\leq O(\eps/\log n)\cdot |\opt(\rset(v))|$ must hold. It is not hard to show that this constraint forces us to allow the partitioning curve $C$ to have as many as $\Omega(\log n/\eps)$ corners, and so in general $L^*=\Omega(\log n/\eps)$ must hold. For the second approach, over the course of our algorithm, we will need to handle sub-instances whose optimal solution values are small relatively to $|\opt|$, and their corresponding polygons have small areas. If we construct an $(N\times N)$-grid $G$, with $N<<n$, then polygons that are aligned with $G$ cannot capture all such instances.

In order to better illustrate our approach, we start by showing a $(1-\eps)$-approximation algorithm with running time $n^{O(\sqrt{\log n}/\eps^3)}$. This algorithm already needs to overcome the barriers described above, and will motivate our final algorithm. 
Consider the divide-and-conquer view of the algorithm, like the one we described in the construction of the partitioning tree. We can partition this algorithm into $\Theta(\sqrt{\log n})$ phases, where the values of the optimal solutions $|\opt(\rset(v))|$ of instances $\rset(v)$ considered in every phase go down by a factor of approximately $2^{\sqrt{\log n}}$. In other words, if we consider the partitioning tree, and we call all vertices at distance exactly $i$ from the root of the tree \emph{level-$i$ vertices}, then every phase of the algorithm roughly contains $\Theta(\sqrt{\log n})$ consecutive levels of the tree. Therefore, the number of such phases is only $O(\sqrt {\log n})$, and so we can afford to lose an $\Theta(\eps/\sqrt{\log n})$-fraction of the rectangles from the optimal solution in every phase. At the end of every phase, for every polygon $P$ defining one of the resulting instances $\rset(P)$ of the problem, we can then afford to repeatedly partition $P$ into sub-polygons, reducing their boundary complexity to $O(\sqrt{\log n}/\eps)$. This allows us to use polygons with only $L_1=\Theta(\sqrt{\log n}/\eps)$ corners as the ``interface'' between the different phases. Within each phase, we still need to allow the polygons we consider to have $L_2=\Theta(\log n/\eps)$ corners. However, now we can exploit the second approach: since the values of the optimal solutions of all instances considered within a single phase are relatively close to each other - within a factor of $2^{\Theta(\sqrt{\log n}})$, we can employ discretization, by constructing a grid with $2^{O(\sqrt{\log n})}$ vertical and horizontal lines, and requiring that polygons considered in the phase are aligned with this grid.

To summarize, we use a two-level recursive construction. The set of level-1 polygons (that intuitively serve as the interface between the phases), contains all polygons whose corners have integral coordinates between $1$ and $2n$, and they have at most $L_1=\Theta(\sqrt{\log n}/\eps)$ corners. The number of such polygons is $n^{O(L_1)}=n^{O(\sqrt{\log n})/\eps}$.
Given a level-1 polygon $P$, we construct a collection $\cset(P)$ of level-2 polygons $P'\subseteq P$. We start by constructing a grid $G_P$ that discretizes $P$, and has $2^{O(\sqrt{\log n})}$ vertical and horizontal lines. The grid has the property that for every vertical strip $S$ of the grid, the value of the optimal solution of the instance $\rset(P\cap S)$ is bounded by $|\opt(\rset(P))|/2^{\Theta(\sqrt{\log n})}$, and the same holds for the horizontal strips. Set $\cset(P)$ contains all polygons $P'\subseteq P$ that have at most $L_2=\Theta(\log n/\eps)$ corners, so that $P'$ is aligned with $G_P$. The number of such polygons is bounded by $\left(2^{O(\sqrt{\log n})}\right)^{O(L_2)}=2^{O(\log^{3/2}n/\eps)}\leq n^{O(\sqrt{\log n}/\eps)}$. The final set $\cset$ of polygons corresponding to the entries of the dynamic programming table includes all level-1 polygons, and for every level-1 polygon $P$, all level-$2$ polygons in the set $\cset(P)$. The algorithm for computing the entries of the dynamic programming table remains unchanged. This reduces the running time to $n^{O(\sqrt{\log n}/\eps^3)}$.

In order to improve the running time to $n^{\poly(\log\log n/\eps)}$, we extend this approach to $O(\log \log n)$ recursive levels.
As before, we partition the execution of the algorithm into phases, where a phase ends when the values of the optimal solutions of all instances involved in it decrease by a factor of roughly $\sqrt n$. Therefore, the algorithm has at most $2$ phases. At the end of each phase, we employ a ``clean-up'' procedure, that reduces the number of corners of every polygon to $L_1=\Theta((\log\log n)^3/\eps)$, with at most $2|\opt|\cdot f(L_1)$ total loss in the number of rectangles from the optimal solution, where $f(L)=O(1/L)$. These polygons, that we call level-1 polygons, serve as the interface between the different phases. The set $\cset_1$ of level-1 polygons then contains all polygons whose corners have integral coordinates between $1$ and $2n$, that have at most $L_1$ corners. The number of such polygons is bounded by $n^{O(L_1)}$. Consider now an execution of a phase, and let $P$ be our initial level-1 polygon. Since the values of the optimal solutions of instances considered in this phase are at least $|\opt(\rset(P))|/\sqrt n$, we can construct an $(O(\sqrt n)\times O(\sqrt n))$-grid $G_P$, that will serve as our discretization grid. We further partition the execution of the current phase (that we refer to as \emph{level-1 phase}) into two level-2 phases, where the value of the optimal solution inside each level-2 phase goes down by a factor of roughly $n^{1/4}$, and we let $L_2=2L_1$. The total number of level-2 phases, across the execution of the whole algorithm, is then at most $4$, and at the end of each such phase, we again apply a clean-up procedure, that decreases the number of corners in each polygon to $L_2$. The loss incurred in every level-2 phase due to the cleanup procedure can be bounded by $|\opt|\cdot f(L_2)=|\opt|\cdot f(L_1)/2$, and the total loss across all level-2 phases is at most $2|\opt| f(L_1)$. For every level-1 polygon $P$, we define a set $\cset_2(P)$ of level-2 polygons, that contains all polygons $P'\subseteq P$ with at most $L_2$ corners, so that $P'$ is aligned with $G_P$. We continue the same procedure for $\Theta(\log\log n)$ recursive levels. For each level $i$, we let $L_i=2L_{i-1}=2^{i-1}L_1$, and we let $\rho_i=n^{1/2^i}$. In each level-$i$ phase, the values of the optimal solutions to the instances defined by the corresponding polygons should decrease by a factor of roughly $\rho_i$, so there are approximately $2^i$ level-$i$ phases overall. At the end of each phase, we apply the clean-up procedure, in order to decrease the number of corners of each polygon to $L_i$. The loss at the end of each level-$i$ phase in the number of rectangles from the optimal solution is then at most  $f(L_i)\cdot |\opt|=f(L_1)\cdot |\opt|/2^{i-1}$, and since the number of level-$i$ phases is  $2^i$, the total loss due to the cleanup procedure in level-$i$ phases is bounded by $2f(L_1)|\opt|$. Summing up over all levels, the total loss due to the cleanup procedure is bounded by $2|\opt|f(L_1)\log \log n\leq \eps |\opt|/\log\log n$. Additional loss is incurred due to the balanced $L_{\log\log n}$-corner partitions of level-$(\log\log n)$ instances, but this loss is analyzed as before, since $L_{\log\log n}=\Omega(\log n/\eps)$. In order to define level-$i$ polygons, for every level-$(i-1)$ polygon $P$, we compute a grid $G_P$ that has roughly $O(\rho_i)$ vertical and horizontal lines, so that for every vertical or horizontal strip $S$ of the grid $G_P$, the value of the optimal solution of instance $\rset(S\cap P)$ is roughy $|\opt(\rset(P))|/\rho_i$. We then let $\cset_i(P)$ contain all polygons $P'\subseteq P$ that have at most $L_i$ corners and are aligned with $G_P$. The final set of level-$i$ polygons is the union of all sets $\cset_i(P)$ for all level-$(i-1)$ polygons $P$.
Let $\cset_i$ denote the set of all level-$i$ polygons, and let $\cset$ be the set of all polygons of all levels. Then it is immediate to verify that
$|\cset_i|\leq |\cset_{i-1}|\cdot \rho_i^{O(L_i)}\leq n^{O(L_1)}$, and since we employ $O(\log \log n)$ levels, overall $|\cset|\leq n^{O(L_1\log\log n)}=n^{O((\log\log n)^4/\eps)}$.

\paragraph{Future directions.} Unlike the algorithm of \cite{AdamaszekW13}, our algorithm does not extend to the weighted setting where each rectangle has a weight and the goal is to find a maximum weight set of independent rectangles. The main technical obstacle is the discretization procedure where we construct a grid and we restrict the partitions into sub-instances to be aligned with the grid. In the weighted setting, there may be some heavy rectangles of the optimal solution that are not aligned with the grid. The optimal solution uses only a small number of such rectangles, but there may be many of them present in the input instance, and we do not know beforehand which of these rectangles are in the optimal solution and thus which rectangles to remove to obtain the alignment property. We leave the extension to the weighted setting, as well as more general shapes such as rectilinear polygons, as directions for future work.

\paragraph{Organization.} We start with preliminaries in Section~\ref{sec: prelims}, and summarize the partitioning theorems that we use in Section~\ref{sec: partitions}. We then give a general outline of the dynamic programming--based algorithm and its analysis that we employ throughout the paper in Section~\ref{sec: alg outline}. Section~\ref{sec: running time exp log n} contains a recap of the algorithm of Adamaszek and Wiese~\cite{AdamaszekW13} in our framework; we also improve its running time to $n^{O(\log n)/\eps^3}$. In Section~\ref{sec: running time exp sqrt log n} we introduce the technical machinery we need in order to improve the running time of the algorithm, and show a \QPTAS with running time $n^{O(\sqrt{\log n})/\eps^3}$, using the two-level approach. This approach is then extended in Section~\ref{sec: n to polyloglog opt running time} to $O(\log\log n)$ recursive levels, completing the proof of Theorem~\ref{thm: main}.

\label{---------------------------------sec: prelims-----------------------------------------}
%-------------------------------------------------------
\section{Preliminaries}\label{sec: prelims}
%-------------------------------------------------------
%-------------------------------------------------------
In the Maximum Independent Set of Rectangles (\misr) problem, the input is a set $\mathcal{R} = \{R_1, R_2, ..., R_n\}$ of $n$ axis-parallel rectangles in the 2-dimensional plane. Each rectangle $R_i$ is specified by the coordinates of its lower left corner $(x_i^{(1)}, y_i^{(1)})$ and its upper right corner $(x_i^{(2)}, y_i^{(2)})$. We view the input rectangles as open subsets of points of the plane, so $R_i = \{(x,y) \mid x_i^{(1)} < x < x_i^{(2)} \;\textrm{and}\; y_i^{(1)} < y < y_i^{(2)}\}$, and we assume that all rectangles have non-zero area. We say that two rectangles $R_i$ and $R_j$ intersect if $R_i \cap R_j \neq \emptyset$, and we say that they are disjoint otherwise.
The goal in the MISR problem is to find a maximum-cardinality subset $\rset^*\subseteq \rset$ of rectangles, such that all rectangles in $\rset^*$ are mutually disjoint. %We denote by $\Gamma(R)$ the boundary of the rectangle $R$.

\paragraph{Canonical Instances.}
We say that a set $\rset$ of rectangles is {\it non-degenerate}, if and only if for every pair $R,R'\in \rset$ of distinct rectangles, for every corner $p=(x,y)$ of $R$ and every corner $p'=(x',y')$ of $R'$, $x\neq x'$ and $y\neq y'$. 
We say that an input $\rset$ to the \misr problem is \emph{canonical}, if and only if $\rset$ is a non-degenerate set of rectangles, whose corners have integral coordinates between $1$ and $2n$. The following claim allows us to transform any input instance of the \misr problem into an equivalent canonical instance. The proof uses standard arguments and appears in the Appendix.

\begin{claim}\label{claim: canonical instances}
There is an efficient algorithm, that, given an instance $\rset$ of the \misr problem on $n$ rectangles, whose optimal solution value is denoted by $w^*$, computes a canonical instance $\rset'$ of \misr, whose optimal solution value is at least $w^*$. Moreover, given a solution $\sset'$ to $\rset'$, we can efficiently compute a solution $\sset$ to $\rset$ of the same value.
\end{claim}

 From now on, we assume without loss of generality that our input instance $\rset$ is a canonical one. %Abusing the notation, we denote by $\opt$ both an optimal solution to $\rset$ and its value.
 Let $B$ be the closed rectangle whose lower left corner is $(0,0)$ and upper right corner is $(2n+1,2n+1)$. We call $B$ the {\it bounding box} of $\rset$. Notice that the boundaries of the rectangles in $\rset$ are disjoint from  the  boundary of $B$. 

%----------------------------------------------------------------
%----------------------------------------------------------------
%----------------------------------------------------------------
%----------------------------------------------------------------
%----------------------------------------------------------------
%----------------------------------------------------------------

\paragraph{Sub-Instances.}
Over the course of our algorithm, we will define sub-instances of the input instance $\rset$. Each such sub-instance is given by some (not necessarily connected) region $S\subseteq B$, and it consists of all rectangles $R\in \rset$ with $R\subseteq S$. In~\cite{AdamaszekW13}, each such sub-instance was given by a polygon $S$, whose boundary edges are parallel to the axes, and the boundary contains at most $\poly\log n$ edges. We define our sub-instances slightly differently. Each sub-instance is defined by a family $\fset$ of at most $O(\log n)$ axis-parallel closed rectangles that are contained in $B$ (but they do not necessarily belong to $\rset$), where every pair of rectangles in $\fset$ are mutually internally disjoint (that is, they are disjoint except for possibly sharing points on their boundaries). We call such rectangles $F\in \fset$ {\it fake rectangles}. Let $S(\fset)=B\setminus \left (\bigcup_{F\in \fset} F\right )$. The sub-instance associated with $\fset$, that we denote by $\rset(\fset)$, consists of all rectangles $R\in \rset$ with $R\subseteq S(\fset)$. In other words, we view the fake rectangles as ``holes'' in the area defined by the bounding box $B$, and we only consider input rectangles that do not intersect these holes. For example, if $S$ is a simple polygon whose boundary is axis-parallel, and contains $L$ corners, then, as we show later, we can ``pad'' the area $B\setminus S$ with a set $\fset$ of $O(L)$ internaly disjoint fake rectangles, and obtain the instance $\rset(\fset)$, containing all rectangles $R\subseteq S$. However, our definition of sub-instances is slightly more general than that of~\cite{AdamaszekW13}, as, for example, $S(\fset)$ is not required to be connected. As we will see later, this definition is more convenient when performing boundary simplification operations. We notice that while the set $\rset(\fset)$ of rectangles is non-degenerate, we do not ensure that $\fset$ is non-degenerate.

Given a sub-instance $\rset(\fset)$, we denote the optimal solution for this sub-instance by $\opt_{\fset}$. The {\it boundary complexity} of the sub-instance $\rset(\fset)$ is defined to be the number of the fake rectangles, $|\fset|$. We denote by $\opt=\opt_{\emptyset}$ the optimal solution to the original problem.% Recall that $W^*\leq \wmax\cdot n\leq n^3$, where $\wmax$ is the maximum weight of any rectangle in $\rset$. %We remind the reader here that we also use the notation $\opt(\rset)$ to denote the optimal solution for a given input set $\rset$ of rectangles. We hope that it will always be clear from the context which of the two meanings we will give to the notation $\opt(\cdot)$.

%We will sometimes identify the sub-instance $\rset(\fset)$ with the set $\fset$ of the fake rectangles defining it. Therefore, when we say ``instance $\fset$'', we mean that $\fset$ is a set of mutually internally-disjoint closed rectangles, and we refer to the sub-instance $\rset(\fset)$. We say that two instances $\fset,\fset'$ are \emph{disjoint} iff $S(\fset)\cap S(\fset')=\emptyset$. We say that $\fset'$ is a sub-instance of $\fset$ iff $S(\fset')\subseteq S(\fset)$. We say that it is a \emph{strict sub-instance} of $\fset$ iff $S(\fset')\subsetneq S(\fset)$. 
%Throughout this writeup, we denote by $\aset(\fset)$ the value of the solution returned by the $O(\log\log \opt)$-approximation algorithm from Corollary~\ref{cor: an log log opt approx} on the instance $\rset(\fset)$.

We say that a set $\fset$ of rectangles is a \emph{valid set of fake rectangles} if and only if $\fset$ consists of closed rectangles $F\subseteq B$ that are mutually internally disjoint.

%----------------------------------------------------------------
%----------------------------------------------------------------
%----------------------------------------------------------------
%----------------------------------------------------------------
%----------------------------------------------------------------
%----------------------------------------------------------------
\paragraph{An $O(\log\log |\opt|)$-Approximation Algorithm.}
We will need to use an approximation algorithm for the problem, in order to estimate the values of optimal solutions of various sub-instances. In~\cite{ChalermsookC09}, an $O(\log \log n)$-approximation algorithm was shown for \MISR. Using the following theorem, whose proof is deferred to the Appendix, we can improve the approximation factor to $O(\log\log |\opt|)$. %We include the proof in Appendix for completeness.

\begin{theorem}\label{thm: from n to opt}
There is an efficient algorithm, that, given an instance $\rset$ of \misr, whose optimal solution value is denoted by $w^*$, computes another instance $\rset'$ of \misr with $|\rset'|\leq O\left((w^*)^4\right)$, such that the value of the optimal solution to $\rset'$ is $\Omega(w^*)$. Moreover, given any solution $S'$ to instance $\rset'$, we can efficiently find a solution $S$ to $\rset$, with $|S|\geq |S'|$.
\end{theorem}

\begin{corollary}\label{cor: an log log opt approx}
There is a universal constant $c_A$, and a $(c_A\cdot \log\log |\opt|)$-approximation algorithm for \MISR.
\end{corollary}

\begin{proof}
Given an instance $\rset$ of \MISR, whose optimal solution value is denoted by $w^*=|\opt|$, we compute a new instance $\rset'$ with $|\rset'|=O\left (|\opt|^4\right )$, using Theorem~\ref{thm: from n to opt}. Let $\opt'$ denote an optimal solution for $\rset'$. We then run the $O(\log\log n)$-approximation algorithm of~\cite{ChalermsookC09} on $\rset'$. Let $S'$ be the solution produced by that algorithm. Then $|S'|\geq \Omega\left (|\opt'| / \log\log |\rset'|\right) = \Omega \left(|\opt| / \log\log |\opt| \right)$. Using Theorem~\ref{thm: from n to opt}, we then find a solution $S$ to $\rset$ of value at least $\Omega \left(|\opt| / \log\log |\opt|\right )$.
\end{proof}

Throughout the paper, we denote by $\aset$ the algorithm from Corollary~\ref{cor: an log log opt approx}. Given any valid set $\fset$ of fake rectangles, we denote by $\aset(\fset)$ the value of the solution returned by algorithm $\aset$ on instance $\rset(\fset)$.

%-------------------------------------------------------------------
%-------------------------------------------------------------------
%-------------------------------------------------------------------
%-------------------------------------------------------------------
%-------------------------------------------------------------------
%-------------------------------------------------------------------
%-------------------------------------------------------------------
%-------------------------------------------------------------------
%-------------------------------------------------------------------
%-------------------------------------------------------------------
%-------------------------------------------------------------------
%-------------------------------------------------------------------
\paragraph{Decomposition Pairs and Triples.}
Our algorithm employs the Divide-and-Conquer paradigm, similarly to the algorithm of~\cite{AdamaszekW13}. Intuitively, given a sub-instance $\rset(\fset)$ of the problem, associated with the polygon $S(\fset)$, we would like to partition it into two (or sometimes three) sub-instances, associated with polygons $S_1$ and $S_2$, respectively. We require that $S_1\cap S_2=\emptyset$ and $S_1, S_2\subsetneq S(\fset)$. Since we define the polygons in terms of the fake rectangles, we employ the following definition.

\begin{definition}
Let $\fset$ any valid set of fake rectangles. We say that $(\fset_1,\fset_2)$ is a \emph{valid decomposition pair} for $\fset$ if for each $i\in \set{1,2}$, $\fset_i$ is a valid set of fake rectangles with $S(\fset_i)\subsetneq S(\fset)$, and $S(\fset_1)\cap S(\fset_2)=\emptyset$. Similarly, we say that $(\fset_1,\fset_2,\fset_3)$ is a \emph{valid decomposition triple} for $\fset$ if for each $i\in \set{1,2,3}$, $\fset_i$ is a valid set of fake rectangles with $S(\fset_i)\subsetneq S(\fset)$, and for all $1\leq i\neq j\leq 3$, $S(\fset_i)\cap S(\fset_j)=\emptyset$. 
\end{definition}

\paragraph{Alignment.} Let $Z$ be any set of points in the plane. Let $X$ be the set of all $x$-coordinates of the points in $Z$, and $Y$ the set of all $y$-coordinates of the points in $Z$. We say that a point $p=(x,y)$ is \emph{aligned} with $Z$, if and only if $x\in X$ and $y\in Y$ (notice that this does not necessarily mean that $p\in Z$; however, if, for example, $X$ and $Y$ only contain integers, then the coordinates of $p$ must be integral). More generally, we say that a polygon $P$ is aligned with $Z$ if and only if every corner of the boundary of $P$ is aligned with $Z$. It is easy to see that the alignment property is transitive: if $P_1,P_2,P_3$ are polygons, where $P_2$ is aligned with the corners of $P_1$, and $P_3$ is aligned with the corners of $P_2$, then $P_3$ is aligned with the corners of $P_1$.

%-------------------------------------------------------
%-------------------------------------------------------
%-------------------------------------------------------
%-------------------------------------------------------
\paragraph{Rectangle Padding.}
%-------------------------------------------------------
%-------------------------------------------------------
%-------------------------------------------------------
%-------------------------------------------------------
We need the following two simple lemmas that allow us to pad polygons with rectangles. The proofs are deferred to the Appendix.

\begin{lemma}\label{lemma: simple tiling}
Let $P$ be any simple closed axis-parallel polygon whose boundary has $L$ corners. Then there is a valid set $\fset$ of fake rectangles, with $|\fset|\leq L-3$, such that $\bigcup_{F\in \fset}F=P$. Moreover, if $Z$ is the set of all points serving as the corners of $P$, then every rectangle in $\fset$ is aligned with $Z$.
\end{lemma}

\begin{lemma}\label{lemma: non-simple-tiling}
Let $P$ be a simple open axis-parallel polygon whose boundary has $L$ corners, and let $B$ be a closed rectangle (the bounding box), such that $P\subseteq B$. Then there is a valid set $\fset$ of fake rectangles, with $|\fset|\leq L+2$, such that $\bigcup_{F\in \fset}F=B\setminus P$.  Moreover, if $Z$ denotes the set of all points serving as the corners of $P$ and $B$, then every rectangle in $\fset$ is aligned with $Z$.
\end{lemma}

%-------------------------------------------------
%-------------------------------------------------
%-------------------------------------------------
%-------------------------------------------------
\label{----------------------------------------sec: partitions---------------------------------------}
\section{Balanced Partitions of \misr Instances}\label{sec: partitions}

%-------------------------------------------------
%-------------------------------------------------
%-------------------------------------------------
%-------------------------------------------------

Since we employ the divide-and-conquer paradigm, we need algorithms that partition a given instance $\iset$ of the \MISR problem into a small number of sub-instances. We would like this partition to be roughly balanced, so that the value of the optimal solution in each sub-instance is at most an $\alpha$-fraction of the optimal solution value for $\iset$, for some constant $0<\alpha<1$. We also need to ensure that we do not lose too many rectangles by this partition. A very useful  tool in obtaining such balanced partitions is $r$-good partitions, that we define below.

Suppose we are given any sub-instance of our problem, defined by a valid set $\fset$ of fake rectangles. Let $\opt'=\opt_{\fset}$ be any optimal solution to the instance $\rset(\fset)$. We would like to find a partition $\pset$ of the bounding box $B$ into rectangles, that we will refer to as ``cells''. 
Each cell $P\in \pset$ is viewed as a closed rectangle, and we say that a rectangle $R\in \opt'$ intersects $P$ if and only if $P\cap R\neq \emptyset$. 
Given a cell $P$, we let $N_P$ be the number of all rectangles of $\opt'$ intersecting $P$.

\begin{definition}
A partition $\pset$ of $B$ into rectangular cells is called $r$-good with respect to $\fset$ and $\opt'$, if and only if:
\begin{itemize}
\item each fake rectangle $F\in \fset$ is a cell of $\pset$;

\item for all other cells $P$, $N_P\leq 20|\opt'|/r$; and
\item $\pset$ contains at most $c^*r$ cells, for some universal constant $c^*>1$.
\end{itemize}
\end{definition}

We note that $\pset\setminus \fset$ defines a partition of $S(\fset)$. It may be more intuitive to view the $r$-good partition as a partition of $S(\fset)$ and not of $B$. However, it is more convenient for us to define $\pset$ as a partition of $B$, as we will see later. Notice that we only require that the number of the rectangles of $\opt'$ intersecting each cell $P\in \pset$ is at most $20|\opt'|/r$, and we ignore all other rectangles of $\rset$ (including those that do not participate in $\rset(\fset)$).

The following theorem shows that there exists an $r$-good partition of $B$ into rectangular cells, for a suitably chosen parameter $r$. Similar theorems have been proved before~\cite{CS95,CF90,Har-Peled14}. We include the proof of the theorem in the Appendix for completeness.

\begin{theorem}\label{thm: r-good partition}
Let $\fset$ be a valid set of fake rectangles with integral coordinates, where $|\fset|=m$ and let  $\opt'$ be any optimal solution to the instance $\rset(\fset)$. Then for every $\max\set{m,3}\leq r\leq |\opt'|/2$, there is an $r$-good (with respect to $\fset$ and $\opt'$) partition $\pset$ of $B$ into rectangular cells, such that the corners of each cell have integral coordinates.\end{theorem}

%From Theorem~\ref{thm: r-good partition}, there is a universal constant $c^*$, such that for any set $\fset$ of $m$ fake rectangles with integral coordinates, for any integer $r\geq \max\set{10,m}$, and any optimal solution $\opt'$ to instance $\rset(\fset)$, there is a partition $\pset$ of $B$ into at most $r$ rectangular cells with integral coordinates, such that each cell intersects at most $\frac{c^*\cdot |\opt(\fset)|}{r}$ rectangles of $\opt'$, and each fake rectangle $F\in \fset$ is a cell of $\pset$. We will use the constant $c^*$ throughout our algorithm.

%-------------------------------------------------------
%-------------------------------------------------------
\paragraph{Balanced Partitions.}
%-------------------------------------------------------
%-------------------------------------------------------
%-------------------------------------------------------
%-------------------------------------------------------
Let $G=(V,E)$ be any embedded planar graph, with weights $w(v)\geq 0$ on the vertices of $G$, such that $\sum_{v\in V}w(v)=W$. 
The \emph{size} of a face $F$ in the embedding of $G$ is the number of vertices on the boundary of $F$, counting multiple visits when traversing the boundary. 
We say that a simple cycle $C$ of $G$ is a \emph{weighted separator} if the total weight of the vertices in the interior of $C$, and the total weight of the vertices in the exterior of $C$ is at most $2W/3$. We use the following theorem of Miller~\cite{Miller}.

\begin{theorem}\label{thm: balanced separators in planar graphs} Let $G=(V,E)$ be an embedded $n$-vertex $2$-connected planar graph, with an assignment $w(v)$ of weights to its vertices. Then there exists a simple cycle weighted separator containing at most $2\sqrt{2\floor{\frac s 2}n}$ vertices, where $s$ is the maximum face size.
\end{theorem}

The following two theorems are used to partition a given instance into sub-instances. Similar techniques were used in previous work~\cite{AdamaszekW13,Har-Peled14}, so we defer the proofs of these theorems to the Appendix. The first theorem allows us to reduce the boundary complexity of a given sub-instance, by partitioning it into two smaller sub-instances with smaller boundary complexities, while the second theorem allows us to partition a given sub-instance into several sub-instances whose optimal solution values are significantly smaller than the optimal solution value of the original sub-instance. Both partition procedures guarantee that the optimal solution value only goes down by a small amount.

%------------------------------------------------------------------
%------------------------------------------------------------------
%------------------------------------------------------------------
%------------------------------------------------------------------

\begin{theorem}\label{thm: balanced partition reducing boundary complexity}
There is a universal constant $c_1>1$, such that the following holds.
Let $\fset$ be any valid set of fake rectangles, with  $|\fset|=L>3$,  and let  $\opt'$ be any optimal solution to $\rset(\fset)$. Let $L\leq r\leq |\opt'|/2Œ$ be any parameter, and let $\pset$ be any $r$-good (with respect to $\fset$ and $\opt'$) partition of $B$. Finally, let $Z$ be the set of the corners of all rectangles in $\pset$. Then there are two sets $\fset_1,\fset_2$ of fake rectangles, such that:

\begin{itemize}
%\item $\fset_1$ consists of disjoint closed rectangles, and so does $\fset_2$;
%\item $S(\fset_1)\cap S(\fset_2)=\emptyset$, and $S(\fset_1)\cup S(\fset_2)\subseteq S(\fset)$;
\item $(\fset_1,\fset_2)$ is a valid decomposition pair for $\fset$;
\item $|\fset_1|,|\fset_2|\leq \frac 2 3L+c_1\sqrt r$;
\item $|\opt_{\fset_1}|+|\opt_{\fset_2}|\geq |\opt_{\fset}|\cdot \left (1-\frac{c_1}{\sqrt r}\right )$; and
\item all rectangles in $\fset_1\cup \fset_2$ are aligned with $Z$.
\end{itemize}
\end{theorem}
%------------------------------------------------------------------
%------------------------------------------------------------------
%-----------------------------------------------------------------

%------------------------------------------------------------------
%------------------------------------------------------------------
%------------------------------------------------------------------

%------------------------------------------------------------------
\begin{theorem}\label{thm: balanced partition general} There is a universal constant $c_2>1$, such that the following holds.
Let $\fset$ be a valid set of fake rectangles, with $|\fset|=L\geq 0$, and let $\opt'$ be any optimal solution to $\rset(\fset)$. Let $\max\set{L,2^{24}c^*}\leq r\leq |\opt'|/2$ be a parameter, where $c^*$ is the constant from the definition of $r$-good partitions, and let $\pset$ be any $r$-good partition of $B$ with respect to $\fset$ and $\opt'$. Finally, let $Z$ be the set of the corners of all rectangles in $\pset$. Then there are two sets $\fset_1,\fset_2$ of fake rectangles, such that:

\begin{itemize}
%\item $\fset_1$ consists of disjoint closed rectangles, and so does $\fset_2$;
\item $(\fset_1,\fset_2)$ is a valid decomposition pair for $\fset$;
%\item $S(\fset_1)\cap S(\fset_2)=\emptyset$, and $S(\fset_1)\cup S(\fset_2)\subseteq S(\fset)$;
\item $|\fset_1|,|\fset_2|\leq L+c_2\sqrt r$;
\item $|\opt_{\fset_1}|,|\opt_{\fset_2}|\leq 3|\opt_{\fset}|/4$;
\item $|\opt_{\fset_1}|+|\opt_{\fset_2}|\geq |\opt_{\fset}|\cdot \left (1-\frac{c_2}{\sqrt r}\right )$; and
\item all rectangles in $\fset_1\cup \fset_2$ are aligned with $Z$.
\end{itemize}
\end{theorem}
%------------------------------------------------------------------
%------------------------------------------------------------------

%-------------------------------------------------------
%-------------------------------------------------------
%-------------------------------------------------------
%-------------------------------------------------------
%-------------------------------------------------------
%-------------------------------------------------------
%-------------------------------------------------------
%-------------------------------------------------------
%-------------------------------------------------------
%-------------------------------------------------------
The following corollary plays an important role in all our algorithms. Roughly speaking, it allows us to partition any sub-instance of the problem into three smaller sub-instances, whose optimal solution values go down by a constant factor, but whose boundary complexities remain appropriately bounded.

\begin{corollary}\label{corollary: main partition}
There is a universal constant $c_3>10$, such that the following holds. For any parameter $L^*>c_3$, for any valid set $\fset$ of fake rectangles with $|\fset|=L\leq L^*$, such that $|\opt_{\fset}|\geq 64 (L^*)^2$, there are three sets $\fset_1,\fset_2, \fset_3$ of fake rectangles, such that:

\begin{itemize}
%\item For each $1\leq i\leq 3$, set $\fset_i$ consists of disjoint closed rectangles;
%\item For each $1\leq i\neq j\leq 3$, $S(\fset_i)\cap S(\fset_j)=\emptyset$, and $S(\fset_1)\cup S(\fset_2)\cup S(\fset_3)\subseteq S(\fset)$;
\item $(\fset_1,\fset_2,\fset_3)$ is a valid decomposition triple for $\fset$;
\item $|\fset_1|,|\fset_2|,|\fset_3|\leq L^*$;
\item For each $1\leq i\leq 3$, $|\opt_{\fset_i}|\leq 3|\opt_{\fset}|/4$;
\item $\sum_{i=1}^3|\opt_{\fset_i}|\geq |\opt_{\fset}|\cdot \left (1-\frac{c_3}{L^*}\right )$; and
\item if the rectangles in $\fset$ have integral coordinates, then so do the rectangles in $\fset_1\cup\fset_2\cup \fset_3$.
\end{itemize}
\end{corollary}

Notice that in all of the above partitioning theorems, the optimal values of the resulting sub-instances reduce by the factor of at least $3/4$. It is therefore more natural for us to work with logarithms to the base of $4/3$. In the rest of this paper, unless stated otherwise, all logarithms are to the base of $4/3$.

%-----------------------------------------------------------------------
%-----------------------------------------------------------------------
%-----------------------------------------------------------------------
%-----------------------------------------------------------------------
%-----------------------------------------------------------------------
%-----------------------------------------------------------------------
%-----------------------------------------------------------------------
%-----------------------------------------------------------------------

\label{----------------------------------------sec: algorithm outline------------------------------------------------}
\section{Algorithm Outline}\label{sec: alg outline}

All our algorithms follow the same general outline, that we describe here. We define a family $\cset$ of \emph{important sets of fake rectangles}, that contains $\emptyset$, so that every element of $\cset$ is a valid set of fake rectangles. As an example, $\cset$ may contain all valid sets $\fset$ of fake rectangles with $|\fset|\leq L^*$ for some bound $L^*$, such that all corners of all rectangles in $\fset$ have integral coordinates. Additionally, we define a set $\cset'\subseteq \cset$ of \emph{basic important sets of fake rectangles}. Intuitively, for each $\fset\in\cset'$, the corresponding instance $\rset(\fset)$ is ``simple'' in some sense: for example, its optimal solution value may be suitably small. We assume that we are given an algorithm $\aset'$, that, given a set $\fset\in \cset$ of fake rectangles, tests whether $\fset\in \cset'$, and an algorithm $\aset''$, that computes a $(1-\eps/2)$-approximate solution to each instance $\rset(\fset)$ with $\fset\in \cset'$. We discuss the running times of these algorithms later.
 We will ensure that $\set{B}\in \cset'$ (the set of fake rectangles containing the bounding box only). This guarantees that for every set $\fset\in \cset\setminus \cset'$ there is always a valid decomposition pair $(\fset_1,\fset_2)$  with $\fset_1,\fset_2\in \cset$ - for example, where $\fset_1=\fset_2=\set{B}$.

Once the families $\cset,\cset'$, and the algorithms $\aset',\aset''$ are fixed, our algorithm is also fixed, and proceeds via simple dynamic programming, as follows. The dynamic programming table $T$ contains, for every important set $\fset\in \cset$ of fake rectangles, an entry $T[\fset]$, that will contain an approximate solution to the corresponding instance $\rset(\fset)$. 
In order to initialize $T$, for every important set $\fset\in \cset$ of fake rectangles, we test whether $\fset\in \cset'$ using algorithm $\aset'$, and if so, we apply algorithm $\aset''$ to compute a valid $(1-\eps/2)$-approximate solution to instance $\rset(\fset)$, which is then stored at $T[\fset]$.
Once we finish the initialization step, we start to fill out the entries $T[\fset]$ for $\fset\in \cset\setminus\cset'$ from smaller to larger values of the area of $S(\fset)$.% In other words, $T[\fset]$ is processed only after all sets $\fset'\in \cset$ with $S(\fset')\subsetneq S(\fset)$ have been processed. 
%
%A set $\fset\in \cset$ is called a \emph{basic set}, if the optimal solution for instance $\rset(\fset)$ contains at most one rectangle. Observe that we can check whether $\fset$ is a basic set, and if so, find an optimal solution for instance $\rset(\fset)$ efficiently. In order to initialize the dynamic programming table, for each basic set $\fset$, we let $T(\fset)$ contain the optimal solution to instance $\rset(\fset)$. 

Consider now some important set $\fset\in \cset\setminus\cset'$ of fake rectangles, and assume that for all $\fset'\in \cset$ with $S(\fset')\subsetneq S(\fset)$, we have already processed the entry $T[\fset']$. Entry $T[\fset]$ is computed as follows. First, for every triple $\fset_1,\fset_2,\fset_3\in \cset$ of important sets of fake rectangles, such that $(\fset_1,\fset_2,\fset_3)$ is a valid decomposition triple for $\fset$, we consider the solution $\xset=T[\fset_1]\cup T[\fset_2]\cup T[\fset_3]$. We do the same for every pair $\fset_1,\fset_2\in \cset$ of important sets of fake rectangles, such that $(\fset_1,\fset_2)$ is a valid decomposition pair for $\fset$. Among all such solutions $\xset$, let $\xset^*$ be the one of maximum value. We then store the solution $\xset^*$ in $T[\fset]$. 
Note that since we ensure that every set $\fset\in \cset\setminus\cset'$ has a valid decomposition pair $(\fset_1,\fset_2)$ with $\fset_1,\fset_2\in \cset$, this step is well defined.
This finishes the description of the algorithm. The final solution is stored in the entry $T[\emptyset]$. Notice that the choice of $\cset,\cset'$, and the algorithms $\aset',\aset''$ completely determines our algorithm. The running time depends on $|\cset|$, and on the running times of the algorithms $\aset'$ and $\aset''$.

%Let $f'(n)$, $f''(n)$ be the running times of algorithms $\aset'$ and $\aset''$, respectively, when applied to sets $\fset$ of fake rectangles, such that $|\fset|,|\rset(\fset)|\leq n$. Then the total running time of our algorithm is bounded by $O(|\cset|^4\cdot \poly(n)+|\cset|(f'(n)+f''(n)))$.

It is immediate to see that every entry $T[\fset]$ of the dynamic programming table contains a feasible solution to instance $\rset(\fset)$. We only need to show that the value of the solution stored in $T[\emptyset]$ is close to $|\opt|$.
This is done by constructing a \emph{partitioning tree}, that we define below.

\begin{definition}
Suppose we are given a canonical instance $\rset$ of the \misr problem, a family $\cset$ of important sets of fake rectangles, and a subset $\cset'\subseteq \cset$ of basic important sets of fake rectangles. Assume also that we are given a set $\fset\in \cset$ of fake rectangles. A partitioning tree $\tset(\fset)$ for $\fset$ is a rooted tree, whose every vertex $v\in V(\tset)$ is labeled with an important set $\fset(v)\in \cset$ of fake rectangles, such that the following hold:

\begin{itemize}
%\item For every leaf $v$ of the tree, the corresponding set $\fset(v)\in \cset'$, and $\mu(v)$ is the value returned by algorithm $\aset''$ on instance $\rset(\fset(v))$;

\item if $v$ is the root of the tree, then $\fset(v)=\fset$; and
\item if $v$ is an inner vertex of the tree, and $\set{v_1,\ldots,v_r}$ are its children, then $r\in\set{2,3}$, and  $\set{\fset(v_i)}_{i=1}^r$ is either a valid decomposition pair or a valid decomposition triple for $\fset(v)$.
%\item for all $1\leq i\neq j\leq 3$, $S(\fset_i)\cap S(\fset_j)=\emptyset$;
%\item $S(\fset_1)\cup S(\fset_2)\cup S(\fset_3)\subseteq S(\fset)$; and
\end{itemize}

We say that $\tset(\fset)$ is a \emph{complete partitioning tree} for $\fset$, if additionally for every leaf vertex $v$ of $\tset(\fset)$, $\fset(v)\in \cset'$.

Let $\iset(\tset(\fset))$ and $\lset(\tset(\fset))$ denote the sets of all inner vertices and all leaf vertices of the tree $\tset(\fset)$, respectively. For every inner vertex $v\in \iset(\tset(\fset))$, whose children are denoted by $v_1,\ldots,v_r$, we define the loss at $v$ to be $\lambda(v)=|\opt_{\fset(v)}|-\sum_{i=1}^r|\opt_{\fset(v_i)}|$. The loss of the tree $\tset(\fset)$, denoted by $\Lambda(\tset(\fset))$ is $|\opt_{\fset}|-\sum_{v\in \lset(\tset(\fset))}|\opt_{\fset(v)}|=\sum_{v\in \iset(\tset(\fset))}\lambda(v)$.
\end{definition}

The following useful observation is immediate from the definition of the partitioning tree.

\begin{observation}\label{obs: partitioning tree - disjointness of non-descendants}
Let $\tset(\fset)$ be a partitioning tree for some set $\fset\in \cset$. If $U\subseteq V(\tset(\fset))$ is any subset of vertices of the tree, such that no vertex of $U$ is a descendant of the other in the tree, then $\set{S(\fset(v))}_{v\in U}$ are all mutually disjoint, and

\[\sum_{v\in U}|\opt_{\fset(v)}|\leq |\opt_{\fset}|.\]
\end{observation}

\begin{definition}
Suppose we are given a canonical instance $\rset$ of the \misr problem, a family $\cset$ of important sets of fake rectangles, and a subset $\cset'\subseteq \cset$ of basic important sets of fake rectangles. A full partitioning tree for $\rset$ is a complete partitioning tree $\tset(\fset)$ for $\fset=\emptyset$.
\end{definition}

%\item $\mu(v)=\sum_{i=1}^3\mu(v_i)$.

Given a full partitioning tree $\tset$, we will denote by $\iset(\tset)$ and $\lset(\tset)$ the sets of all the inner vertices and all the leaf vertices of $\tset$, respectively. For every vertex $v$ of $\tset$, we associate a value $\mu(v)$ with $v$, as follows. If $v$ is a leaf of $\tset$, then $\mu(v)$ is the value of the $(1-\eps/2)$-approximate solution to instance $\rset(\fset(v))$ returned by the algorithm $\aset''$. If $v$ is an inner vertex of $\tset$, then $\mu(v)$ is the sum of values $\mu(v')$ for all children $v'$ of $v$.
It is immediate to see that for every vertex $v$ of the full partitioning tree $\tset$, $\mu(v)$ is the sum of the values $\mu(v')$ for all descendants $v'$ of $v$ that are leaves. We denote by $\mu(\tset)$ the value $\mu(v)$ of the root vertex $v$ of $\tset$. The following observation connects the value of the solution computed by the dynamic programming algorithm to $\mu(\tset)$.

\begin{observation}\label{obs: DP to tree}
For every vertex $v$ of the tree $\tset$, the entry $T[\fset(v)]$ of the dynamic programming table contains a solution to instance $\rset(\fset(v))$, whose value is at least $\mu(v)$.
\end{observation}
\begin{proof}
The proof is by induction on the depth of the vertex $v$. The assertion is clearly true for the leaves of the tree. Consider now some inner vertex $v$ of the tree, and the corresponding important set $\fset(v)\in \cset$ of fake rectangles. Let $v_1,\ldots,v_r$ be the children of $v$ (where $r\in \set{2,3}$), and  let $\fset_1,\ldots,\fset_r$ be the sets of fake rectangles associated with them. Then $\set{\fset_i}_{i=1}^r$ is either a valid decomposition pair or a valid decomposition triple for $\fset(v)$, and for all $1\leq i\leq r$, $\fset_i\in \cset$. Therefore, the dynamic programming algorithm considers the solution $\xset$, obtained by taking the union of the solutions stored in $\set{T[\fset_i]}_{i=1}^r$. From the induction hypothesis, for each $1\leq i\leq r$, the value of the solution stored in $T[\fset_i]$ is at least $\mu(v_i)$, and so the value of the final solution stored at $T[\fset]$ is at least $\mu(v)$.
\end{proof}

The following simple observation will be useful in analyzing the approximation factors achieved by our algorithms.

\begin{observation}\label{obs: final analysis}
Suppose we are given a canonical instance $\rset$ of the \misr problem, a family $\cset$ of important sets of fake rectangles, and a subset $\cset'\subseteq \cset$ of basic important sets of fake rectangles. Assume further that there exists a full partitioning tree $\tset$ for $\rset$, whose loss $\Lambda(\tset)\leq \eps |\opt|/2$. Then the dynamic programming-based algorithm described above computes a $(1-\eps)$-approximate solution to $\rset$.
\end{observation}

\begin{proof}
Let $v_r$ be the root vertex of the tree $\tset$, so $\fset_{v_r}=\emptyset$. From the above discussion, entry $T[v_r]$ of the dynamic programming table stores a valid solution to instance $\rset=\rset(\fset(v_r))$ of value at least $\mu(v_r)$. It now remains to show that $\mu(v_r)\geq (1-\eps)|\opt|$. Indeed:

\[\mu(v_r)=\sum_{v\in \lset(\tset)}\mu(v)\geq \sum_{v\in \lset(\tset)}(1-\eps/2)|\opt_{\fset(v)}|=(1-\eps/2)\left(|\opt_{v_r}|-\Lambda(\tset)\right)\geq (1-\eps/2)^2|\opt|\geq (1-\eps)|\opt|.\]
\end{proof}

Notice that in order to analyze our algorithm, we now only need to show the existence of a suitable partitioning tree. We do not need to provide an efficient algorithm to construct such a tree, and in particular, we can assume that we know the optimal solution $\opt$ to our instance $\rset$ when constructing the tree.

As a warm-up, we show a QPTAS whose running time is $n^{O(\log |\opt|/\epsilon^3)}$. This algorithm is very similar to the algorithm of \cite{AdamaszekW13}, except that we use a somewhat more sophisticated partitioning scheme (namely, Corollary~\ref{corollary: main partition}). In order to develop the intuition for our final algorithm, we then improve the running time to $n^{O(\sqrt{\log |\opt|}/\epsilon^3)}$, and show an algorithm with running time $n^{O((\log\log |\opt|/\epsilon)^4)}$ at the end. Notice that since we have assumed that our instance is canonical, $|\opt|\leq O(n^4)$, and so the running time of our final algorithm is $n^{O((\log\log n/\epsilon)^4)}$.

\label{------------------------------------------sec: running time exp log n-----------------------------------------}
%-------------------------------------------------------
%-------------------------------------------------------
%-------------------------------------------------------
%-------------------------------------------------------
%-------------------------------------------------------
%-------------------------------------------------------
\section{A QPTAS with running time $n^{O(\log |\opt| /\epsilon^3)}$}\label{sec: running time exp log n}
%-------------------------------------------------------
%-------------------------------------------------------
%-------------------------------------------------------
%-------------------------------------------------------
%-------------------------------------------------------
%-------------------------------------------------------

We need the following parameters for our algorithm. Let $L^*=\frac{2c_3\log |\opt|}{\epsilon}$, where $c_3$ is the constant from Corollary~\ref{corollary: main partition}, and let $\tau=64\left({L^*}\right )^2=\Theta((\log |\opt|/\eps)^2)$. 
Notice that any valid set $\fset$ of fake rectangles with $|\opt_{\fset}|> \tau$ and $|\fset|\leq L^*$ satisfies the conditions of Corollary~\ref{corollary: main partition}. 

We let the family $\cset$ of important sets of fake rectangles contain all valid sets $\fset$ of fake rectangles with $|\fset|\leq L^*$, such that all corners of all rectangles in $\fset$ have integral coordinates. Therefore, $|\cset|=n^{O(L^*)}=n^{O(\log |\opt|/\eps)}$.
The family $\cset'\subseteq \cset$ of basic sets of fake rectangles contains all sets $\fset$ with $|\opt_{\fset}|\leq \tau$. We can verify whether $\fset\in \cset'$, and if so, we can find the optimal solution to instance $\rset(\fset)$ in time $n^{O(\tau)}=n^{O(\log^2|\opt|/\eps^2)}$. This defines the algorithms $\aset'$ and $\aset''$, and we can now employ the dynamic programming-based algorithm, described in Section~\ref{sec: alg outline}. In order to analyze the running time of the algorithm, observe that the initialization step takes time $O(|\cset|)\cdot n^{O(\log^2|\opt|/\eps^2)}=n^{O(\log^2|\opt|/\eps^2)}$, and the remaining part of the algorithm takes time $O(|\cset|^4\poly(n))=n^{O(\log |\opt|/\eps)}$, so overall the running time is $n^{O(\log^2|\opt|/\eps^2)}$. We later show how to improve the running time to $n^{O(\log |\opt|/\eps^3)}$ by replacing algorithms $\aset'$ and $\aset''$ with more efficient algorithms.

It now remains to show that the value of the solution computed by the algorithm is at least $(1-\eps)|\opt|$. We do so by the constructing a full partitioning tree $\tset$ for $\rset$, as described in Section~\ref{sec: alg outline}. We start with the tree $\tset$ containing a single vertex $v$, with $\fset(v)=\emptyset$. While there is a leaf vertex $v\in \tset$ with $\fset(v)\in \cset\setminus\cset'$ we add three children $v_1,v_2$ and $v_3$ to vertex $v$. Applying Corollary~\ref{corollary: main partition} to $\fset(v)$, we obtain a decomposition triple $(\fset_1,\fset_2,\fset_3)$ for $\fset(v)$, and we associate each of the three new vertices $v_1,v_2,v_3$ with the sets $\fset_1,\fset_2$ and $\fset_3$, respectively. Notice that for $i\in\set{1,2,3}$, $|\fset_i|\leq L^*$, and if $\fset(v)\in \cset$, then all corners of all rectangles in $\fset_i$ have integral coordinates, so $\fset_1,\fset_2,\fset_3\in \cset$. Notice also that from Corollary~\ref{corollary: main partition}, the loss of vertex $v$ is $\lambda(v)=|\opt_{\fset(v)}|-\sum_{i=1}^3|\opt_{\fset_i}|\leq \frac{c_3|\opt_{\fset(v)}|}{L^*}$.

The algorithm terminates when for every leaf vertex $v$ of $\tset$, $\fset(v)\in \cset'$.

It is now enough to prove that the loss of the tree $\tset$ is at most $\eps|\opt|/2$. We partition the inner vertices of the tree $\tset$ into subsets $U_1,U_2,\ldots$, where a vertex $v\in \iset(\tset)$ belongs to $U_i$ if and only if the number of vertices of $\tset$ lying on the unique path connecting $v$ to the root of the tree is exactly $i$. Since the values $|\opt_{\fset(v)}|$ decrease by the factor of at least $3/4$ as we go down the tree, the number of non-empty subsets $U_i$ is bounded by $\log |\opt|$. Consider now some $1\leq i\leq \log |\opt|$. Notice that for every pair $v,v'\in U_i$ of vertices, neither is a descendant of the other, and so from Observation~\ref{obs: partitioning tree - disjointness of non-descendants}, $\sum_{v\in U_i}|\opt_{\fset(v)}|\leq |\opt|$. Therefore, the total loss of all vertices in $U_i$ is:

\[\sum_{v\in U_i}\lambda(v)\leq \sum_{v\in U_i}\frac{c_3|\opt_{\fset(v)}|}{L^*}\leq \frac{c_3|\opt|}{L^*}\]

Overall, the total loss of the tree $\tset$ is:

\[\sum_{i=1}^{\log|\opt|}\sum_{v\in U_i}\lambda(v)\leq \frac{c_3|\opt|\cdot \log |\opt|}{L^*}\leq \frac{\eps|\opt|}2,\]

since $L^*=\frac{2c_3\log|\opt|}{\epsilon}$.

So far, we have shown an algorithm that computes a $(1-\eps)$-approximate solution in time $n^{O(\log^2|\opt|/\eps^2})$. We can also use it to compute, for any instance $\rset(\fset)$ defined by any valid set $\fset$ of fake rectangles, a $(1-\eps/2)$-approximate solution for $\rset(\fset)$, in time $n^{O(\log^2|\opt_{\fset}|/\eps^2)}$. We denote this algorithm by $\aset^*$.

We now describe a slightly modified version of the algorithm, whose running time is $n^{O(\log |\opt|/\eps^3)}$. 
We assume that $\eps>1/\log^2|\opt|$, since otherwise algorithm $\aset^*$ provides an $(1-\eps)$-approximation in time $n^{O(\log |\opt|/\eps^3)}$.
The family $\cset$ of important sets of fake rectangles remains the same as before, but the family $\cset'\subseteq \cset$ of basic sets of fake rectangles is defined slightly differently: it contains all sets $\fset\in \cset$ of fake rectangles, such that algorithm $\aset$ from Corollary~\ref{cor: an log log opt approx} returns a solution of value at most $\tau$ to instance $\rset(\fset)$. We then let $\aset'$ be the algorithm $\aset$. Notice that if $\fset\in \cset'$, then $|\opt_{\fset}|\leq O(\aset(\rset(\fset)) \log\log |\opt|)=O(\log^2|\opt|\log\log |\opt|/\eps^2)=O(\poly\log |\opt|)$. We can now use algorithm $\aset^*$ to compute a $(1-\eps/2)$-approximate solution for every instance $\rset(\fset)$ with $\fset\in \cset'$, in time $n^{O(\log^2|\opt_{\fset}|/\eps^2)}=n^{O((\log\log |\opt|)^2/\eps^2)}$. The rest of the algorithm remains unchanged, except that we now use the algorithm $\aset^*$ instead of $\aset''$. It is immediate to verify that the running time of the algorithm is $n^{O(\log |\opt|/\eps^3)}$. For every set $\fset\in \cset\setminus\cset'$ of fake rectangles, we now have $|\opt_{\fset}|\geq \aset(\fset)\geq \tau$, and so $\fset$ is a valid input to Corollary~\ref{corollary: main partition}. We can use the same construction of the partitioning tree as before, to show that the value of the solution computed by the algorithm is at least $(1-\eps)|\opt|$.

\label{------------------------------------------sec: running time exp sqrt log n-----------------------------------------}
%-------------------------------------------------------
%-------------------------------------------------------
%-------------------------------------------------------
%-------------------------------------------------------
%-------------------------------------------------------
%-------------------------------------------------------
%-------------------------------------------------------
\section{A QPTAS with Running Time $n^{O(\sqrt {\log |\opt|}/\eps^3)}$}\label{sec: running time exp sqrt log n}
%-------------------------------------------------------
%-------------------------------------------------------
%-------------------------------------------------------
%-------------------------------------------------------
%-------------------------------------------------------
%-------------------------------------------------------
%-------------------------------------------------------
In Section~\ref{sec: running time exp log n}, we have presented a $(1-\epsilon)$-approximation algorithm for MISR with running time $n^{O(\log|\opt|/\eps^3)}$. The running time directly depends on $|\cset|=n^{O(\log |\opt|/\eps)}$, and this value is the bottleneck in the running time of the algorithm. Since the corners of the rectangles in $\fset$ have integral coordinates between $0$ and $2n+1$, we have $\Theta(n^2)$ choices for every corner of each such rectangle, and since we allow the sets $\fset\in \cset$ to contain up to $L^*=\Theta(\log |\opt|/\eps)$ such rectangles, we obtain $n^{O(\log|\opt|/\eps)}$ choices overall.

Let us informally define a set $\iset$ of points in the plane to be a set of \emph{points of interest} if $\iset$ contains all points that may serve as corners of fake rectangles in sets $\fset\in \cset$. In the algorithm from Section~\ref{sec: running time exp log n}, set $\iset$ contains all points $(x,y)$, where $x$ and $y$ are integers between $0$ and $2n+1$, and so $|\iset|=\Theta(n^2)$.
In order to improve the running time of the algorithm, it is natural to try one of the following two approaches: (i) reduce the number of points of interest; or (ii) reduce the parameter $L^*$ (the maximum allowed cardinality of sets $\fset\in \cset$).

Unfortunately, it is not hard to see that neither of these approaches works directly. Since we eventually need to consider sub-instances $\rset(\fset)$ where $|\opt_{\fset}|$ is very small, we need to allow many points of interest - almost as many as $\Theta(n^2)$. This rules out the first approach.

In order to see that the second approach does not work directly, consider the partitioning tree $\tset$ that we have defined for the analysis of the algorithm, and the partition $(U_1,U_2,\ldots,U_z)$ of the inner vertices of $\tset$ into levels, according to their distance from the root vertex. Recall that the total loss of all vertices at a given level $U_i$ was $O(\eps |\opt|/\log |\opt|)$, and the number of levels is $z=O(\log |\opt|)$, thus giving us a total loss of $O(\eps |\opt|)$ overall. In order to obtain a $(1-\eps)$-approximation, the loss at every level must be bounded by $O(\eps |\opt|/\log |\opt|)$, and so on average, when we apply Corollary~\ref{corollary: main partition} to some set $\fset(v)\in \cset$ of fake rectangles, corresponding to some vertex $v\in V(\tset)$, we cannot afford to lose more than an $O(\eps |\opt_{\fset}|/\log |\opt|)$-fraction of the rectangles of $\opt_{\fset}$. It is not hard to show that this forces us to set $L^*=\Omega(\log |\opt|/\eps)$ (in other words, we cannot obtain a substantially better tradeoff between $L^*$ and the number of the rectangles lost, than that in Corollary~\ref{corollary: main partition}).

We get around this problem as follows. Consider the process of constructing the partitioning tree $\tset$, and let us assume that in every iteration, we choose the leaf $v$ in the current tree with the largest value $|\opt_{\fset(v)}|$ to process, breaking ties arbitrarily. We divide the execution of the algorithm into $O(\sqrt{ \log |\opt|})$ phases, where the $j$th phase finishes when for every leaf $v$ in the current tree, $|\opt_{\fset(v)}|\leq |\opt|/2^{j\sqrt {\log |\opt|}}$. At the end of each phase, while there is a leaf $v$ in $\tset$, whose corresponding set $\fset(v)$ of fake rectangles has boundary complexity $|\fset(v)|>\Omega({\sqrt {\log|\opt|}/\eps})$, we repeatedly apply Theorem~\ref{thm: balanced partition reducing boundary complexity} in order to find a valid decomposition pair $(\fset_1,\fset_2)$ for $\fset(v)$, with $|\fset_1|,|\fset_2|<3|\fset(v)|/4$. This allows us to lower the boundary complexities of the instances that we consider to $O(\sqrt{\log |\opt|}/\eps)$. In the process of doing so, we will lose roughly an $O(\eps/\sqrt{\log |\opt|})$-fraction of the rectangles from the optimal solution. However, since the number of phases is bounded by $O(\sqrt{\log |\opt|})$, we can afford this loss.
Therefore, sets $\fset$ of fake rectangles that we obtain at the end of each phase will have a small boundary complexity - only $O(\sqrt{\log |\opt|}/\eps)$. We call such sets $\fset$ \emph{level-1 sets}.

Inside each phase, we still need to allow the boundary complexities of the instances we consider to be as high as $\Theta(\log |\opt|/\eps)$. However, we can now exploit the fact that for all instances processed in a phase, the values of their optimal solutions are  close to each other - to within a factor of $2^{\sqrt {\log |\opt|}}$. This allows us to define a smaller set of points of interest for the sub-instances considered in every phase, through discretization. The resulting sets $\fset$ of fake rectangles are called \emph{level-2 sets}.
In the next section, we provide the technical machinery that allows us to perform this discretization, as well as the analogues of Theorem~\ref{thm: balanced partition reducing boundary complexity} and Corollary~\ref{corollary: main partition} in this discretized setting. We then describe our algorithm and its analysis. The technical tools developed in this section are also used in our final $n^{O(((\log\log n)/\eps)^4)}$-time algorithm.

\subsection{Grid-Aligned $r$-good Partitions}
%Given an instance $\fset$ of value $\opt'$ of the problem, 
A grid $G$ of size $(z\times z)$ is defined by a collection $\vset=\set{V_0,\ldots,V_z}$ of vertical lines, and a collection $\hset=\set{H_0,\ldots,H_z}$ of horizontal lines, where $V_0$ and $V_z$ coincide with the left and the right boundaries of the bounding box $B$ respectively, and $H_0$, $H_z$ coincide with the bottom and the top boundaries of $B$ respectively. Each vertical line $V_i$ is specified by its $x$-coordinate $x_i$, and it starts at the bottom boundary of $B$ and ends at the top boundary of $B$. Similarly, each horizontal line $H_i$ is specified by its $y$-coordinate $y_i$, and it starts at the left boundary of $B$ and ends at the right boundary of $B$.
We assume that the vertical lines are indexed by their left-to-right order, that is, for each $0\leq i<z$, $V_i$ lies to the left of $V_{i+1}$, and similarly, all horizontal lines are indexed by their bottom-to-top order. Every consecutive pair $V_i,V_{i+1}$ of vertical lines defines a vertical strip $S^V_i$ of the bounding box, and every consecutive pair $H_j,H_{j+1}$ of horizontal lines defines a horizontal strip $S^H_j$. The set of vertices of $G$ is the set of all intersection points of its vertical and horizontal lines.

Suppose we are given a grid $G$, and a valid set $\fset$ of fake rectangles. We say that $\fset$ is \emph{aligned} with $G$, if every corner of every rectangle of $\fset$ belongs to the set $Z$ of the vertices of the grid $G$.

\begin{definition}
Given a valid set $\fset$ of fake rectangles and a parameter $\rho\geq 1$, we say that a grid $G=(\vset,\hset)$ is $\rho$-accurate for $\fset$, if and only if:

\begin{itemize}
\item for every vertical strip $S_i^V$ of the grid, the value of the optimal solution of the sub-instance defined by all rectangles contained in $S_i^V\cap S(\fset)$ is at most $\ceil{|\opt_{\fset}|/\rho}$, and the same holds for every horizontal strip; and

\item $\fset$ is aligned with the grid $G$.
\end{itemize}
\end{definition}

Notice that we allow $\rho>|\opt_{\fset}|$.

The following two observations, that we repeatedly use later, follow easily from the definition of $\rho$-accurate grids.
\begin{observation}\label{obs: rho and rho' accurate grids}
If $G$ is a $\rho$-accurate grid for some valid set $\fset$ of fake rectangles, then for every $\rho'\leq \rho$, $G$ is a $\rho'$-accurate grid for $\fset$.
\end{observation}

\begin{observation}\label{obs: rho accurate for big subinstances}
Let $G$ be a $\rho$-accurate grid for some valid set $\fset$ of fake rectangles, and let $\fset'$ be any valid set of fake rectangles with $S(\fset')\subseteq S(\fset)$, such that $\fset'$ is aligned with $G$. If $|\opt_{\fset'}|\geq \alpha|\opt_{\fset}|$, then $G$ is an $\alpha\cdot \rho$-accurate grid for $\fset'$.
\end{observation}

We will also need the following claim. The proof is deferred to the Appendix.

\begin{claim}\label{claim: find a rho-accurate grid}
There is an efficient algorithm, that, given a valid set $\fset$ of fake rectangles and a parameter $\rho\geq 1$, computes a $\rho$-accurate grid for $\fset$ of size $(z\times z)$, where $z\leq 4c_A\rho\log\log |\opt_{\fset}|+2|\fset|$, and $c_A$ is the constant from Corollary~\ref{cor: an log log opt approx}. Moreover, if all corners of all rectangles in $\rset\cup \fset$ have integral coordinates, then so do all vertices of the grid $G$.
\end{claim}

Over the course of our algorithm, we will construct $\rho$-accurate grids $G$ with respect to some valid sets $\fset$ of fake rectangles, for some values $\rho$ that we specify later. We would like then to find valid decomposition pairs and triples for $\fset$, that are aligned with the grid $G$. In order to be able to do so, we need an analogue of Corollary~\ref{corollary: main partition}, that allows us to find a valid decomposition triple for a $G$-aligned valid set $\fset$ of fake rectangles. Recall that Corollary~\ref{corollary: main partition} uses Theorems~\ref{thm: balanced partition reducing boundary complexity} and \ref{thm: balanced partition general} in the two partitions it performs, and these two theorems in turn start with some $r$-good partition of the input instance $\rset(\fset)$. The sub-instances produced by these theorems are then guaranteed to be aligned with the $r$-good partition. We need the final sub-instances to be aligned with the grid $G$. Therefore, we generalize the notion of $r$-good partitions to grid-aligned $r$-good partitions. We then show an analogue of Theorem~\ref{thm: r-good partition}, proving that such $r$-good grid-aligned partitions can be constructed with the right choice of  parameters. Finally, we prove an analogue of Corollary~\ref{corollary: main partition} that produces grid-aligned sub-instances.
We start with a definition of a grid-aligned $r$-good partition.

\begin{definition} Let $\fset$ be any valid set of fake rectangles, $\opt'$ any optimal solution to instance $\rset(\fset)$, and let $G$ a be grid.
A partition $\pset$ of $B$ into rectangular cells is called a $G$-aligned $r$-good partition with respect to $\fset$ and $\opt'$, iff:
\begin{itemize}
\item Every cell in the partition intersects at most $20|\opt'|/r$ rectangles of $\opt'$; 
\item $\pset$ contains at most $c^{**}r$ cells, where $c^{**}\geq 1$ is some universal constant; 
\item Each fake rectangle $F\in \fset$ is a cell of $\pset$; and
\item Every cell in the partition is aligned with the grid $G$.
\end{itemize}
\end{definition}

Notice that a $G$-aligned $r$-good partition is in particular also an $r$-good partition (where we use a constant $\cdstar$ instead of $c^*$, but this is immaterial because both are some universal constants). In particular, Theorems~\ref{thm: balanced partition reducing boundary complexity} and \ref{thm: balanced partition general} still remain valid if we apply them to $G$-aligned $r$-good partitions, except that we need to replace $c^*$ by $\cdstar$. The proof of the following Theorem is deferred to the Appendix.

\begin{theorem}\label{thm: rho-aligned r-good partition}
Let $\fset$ be any valid set of fake rectangles with $|\fset|=m$ and $\opt'$ any optimal solution to $\rset(\fset)$. Let $G$ be a $\rho$-accurate grid with respect to $\fset$, for some $\rho>\max\set{m,3}$, and let $r$ be a parameter, with $\max\set{m,3}\leq r\leq \min\set{\rho,|\opt'|/16}$. 
Then there is a $G$-aligned $r$-good partition $\pset$ of $B$, with respect to $\fset$ and $\opt'$.
\end{theorem}

Finally, we prove an analogue of Corollary~\ref{corollary: main partition} to obtain partitions into sub-instances that are aligned with $G$.
The proof is almost identical to the proof of Corollary~\ref{corollary: main partition}, except that we use Theorem~\ref{thm: rho-aligned r-good partition} instead of Theorem~\ref{thm: r-good partition} to find grid-aligned $r$-good partitions. For completeness, the proof appears in the Appendix.

\begin{corollary}\label{corollary: main partition with grid}
There is a universal constant $\tc>10$, such that the following holds. For any parameter $L^*>\tc$, for any valid set $\fset$ of fake rectangles, with  $|\fset|=L\leq L^*$  and $|\opt_{\fset}|\geq 512(L^*)^2$, given any $\rho$-accurate grid $G$ for $\fset$, where $\rho\geq 32(L^*)^2$, there is a valid decomposition triple $(\fset_1,\fset_2,\fset_3)$ for $\fset$, such that:

\begin{itemize}
%\item for each $1\leq i\neq j\leq 3$, $S(\fset_1)\cap S(\fset_2)=\emptyset$;
%\item $S(\fset_1)\cup S(\fset_2)\cup S(\fset_3)\subseteq S(\fset)$;
\item for all $1\leq i\leq 3$, $|\fset_i|\leq 3L^*/4$;
\item for all $1\leq i\leq 3$, $|\opt_{\fset_i}|\leq 3|\opt_{\fset}|/4$;
\item $\sum_{i=1}^3|\opt_{\fset_i}| \geq |\opt_{\fset}|\cdot \left (1-\frac{\tc}{L^*}\right )$; and
\item The rectangles in $\fset_1\cup\fset_2\cup \fset_3$ are aligned with the grid $G$.
\end{itemize}
\end{corollary}

The following corollary follows immediately from the above discussion, and it will also be useful for us later.

\begin{corollary}\label{corollary: partition with grid to decrease boundary}
 For any valid set $\fset$ of fake rectangles with $|\fset|=L>\tc$ and $|\opt_{\fset}|\geq 512L^2$, given any $\rho$-accurate grid $G$ for $\fset$, where $\rho \geq 32L^2$, there is a valid decomposition pair $(\fset_1,\fset_2)$ for $\fset$, such that:

\begin{itemize}
%\item for each $1\leq i\neq j\leq 3$, $S(\fset_1)\cap S(\fset_2)=\emptyset$;
%\item $S(\fset_1)\cup S(\fset_2)\cup S(\fset_3)\subseteq S(\fset)$;
\item $|\fset_1|,|\fset_2|\leq 3L/4$;
\item $|\opt_{\fset_1}|+|\opt_{\fset_2}|\geq |\opt_{\fset}|\cdot \left (1-\frac{\tc}{L}\right )$; and
\item The rectangles in $\fset_1\cup\fset_2$ are aligned with the grid $G$,
\end{itemize}

where $\tc$ is the constant from Corollary~\ref{corollary: main partition with grid}.
\end{corollary}

The proof follows from the first step in the proof of Corollary~\ref{corollary: main partition with grid}, that produces two sub-instances $\fset_1',\fset_2'$ of $\fset$ with the desired properties, after setting $L^*=L$.

%-------------------------------------------------------
%-------------------------------------------------------
%-------------------------------------------------------
%-------------------------------------------------------
\subsection{Cleanup Trees}
%In this section, all logarithms are to the base of $4/3$.
Recall that in order to analyze our dynamic programming algorithm, we employ partitioning trees. Recall also that we divide the execution of our algorithm into $O(\sqrt{ \log {|\opt|}})$ phases, where phase $j$ ends when for every leaf $v$ of the current tree, $|\opt_{\fset(v)}|<\frac{|\opt|}{2^{j\sqrt{\log |\opt|}}}$. Once a phase ends, we would like to reduce the boundary complexity of each instance corresponding to the leafs of the current tree, and we will employ  cleanup trees in order to do so.
We  now define cleanup trees. 

\begin{definition}
Suppose we are given integral parameters $\tc <L_2<L_1$,  where $\tilde c$ is the constant from Corollary~\ref{corollary: partition with grid to decrease boundary}, and denote $\delta=\ceil{\log_{4/3}(L_1/L_2)}$. Assume further that we are given
 a valid set $\fset$ of fake rectangles with $|\fset|\leq L_1$, and a $\rho$-accurate grid $G$ for $\fset$, for some parameter $\rho$. An $(L_1,L_2)$-cleanup tree for $\fset$ and $G$ is a rooted binary tree $\tset$, such that every vertex $v$ of $\tset$ is associated with some valid set $\fset(v)$ of fake rectangles, and the following holds:

\begin{itemize}
\item $\fset(v)$ is aligned with $G$ and $|\fset(v)|\leq L_1$;

\item if $v$ is the root vertex of $\tset$, then $\fset(v)=\fset$;

\item if $v$ is an inner vertex of $\tset$, then it has exactly two children, that we denote by $v_1$ and $v_2$, and $(\fset(v_1),\fset(v_2))$ is a valid decomposition pair for $\fset(v)$; and

\item if $v$ is a leaf of $\tset$, then either $|\fset(v)|\leq L_2$ with $|\opt_{\fset(v)}|\geq |\opt_{\fset}|/L_1^{\delta+2}$, or $S(\fset(v))=\emptyset$.
\end{itemize}

Given a cleanup tree $\tset$, let $\iset(\tset)$ and $\lset(\tset)$ denote the sets of its inner vertices and leaves, respectively. For an inner vertex $v\in \iset(\tset)$, whose children are denoted by $v_1$ and $v_2$, the loss at $v$ is $\lambda(v)=|\opt_{\fset(v)}|-|\opt_{\fset(v_1)}|-|\opt_{\fset(v_1)}|$. The total loss of the tree $\tset$, $\Lambda(\tset)=|\opt_{\fset}|-\sum_{v\in \lset(\tset)}|\opt_{\fset(v)}|=\sum_{v\in \iset(\tset)}\lambda(v)$.
\end{definition}

\begin{theorem}\label{thm: cleanup tree}
Suppose we are given integral parameters $\tc <L_2<L_1$,  where $\tilde c$ is the constant from Corollary~\ref{corollary: partition with grid to decrease boundary}, and denote $\delta=\ceil{\log_{4/3}(L_1/L_2)}$. Assume further that we are given
 a valid set $\fset$ of fake rectangles with $|\fset|\leq L_1$  and $|\opt_{\fset}|\geq 512L_1^{\delta+2}$, and a $\rho$-accurate grid $G$ for $\fset$, for some parameter $\rho>32L_1^{\delta+2}$. Then there is an $(L_1,L_2)$-cleanup tree $\tset$ for $\fset$ and $G$, whose loss $\Lambda(\tset)\leq \frac{12\tilde c|\opt_{\fset}|}{L_2}$.
\end{theorem}
%
%Let $L_1,L_2$ be integral parameters with $\tilde c<L_2<L_1$, where $\tilde c$ is the constant from Corollary~\ref{corollary: partition with grid to decrease boundary}, and let $\fset$ be any valid set of fake rectangles with $|\fset|\leq L_1$ and $|\opt_{\fset}|\geq 512L_1^{\delta+2}$. Denote $\delta=\ceil{\log(L_1/L_2)}$, and let $\rho$ be a parameter,  with $\rho>32L_1^{\delta+2}$. Let $G$ be a $\rho$-accurate grid for $\fset$. Then there is an $(L_1,L_2)$-cleanup tree for $\fset$ and $G$, such that, if we denote by $\lset$ the set of all leaves of $\tset$, then for every leaf $v\in \lset$, , and $|\opt_{\fset}|-\sum_{v\in \lset}|\opt_{\fset(v)}|\leq \frac{12\tilde c'|\opt_{\fset}|}{L_2}$.
%\end{theorem}

\begin{proof}
%We assume that $\tc''>\tc$, where $\tc$ is the constant from Corollary~\ref{corollary: partition with grid to decrease boundary}.
For all integers $0\leq i\leq \log L_1-1$, we say that a vertex $v$ of tree $\tset$ is a level-$i$ vertex if $\frac{L_1}{(4/3)^{i+1}}< |\fset(v)|\leq \frac{L_1}{(4/3)^i}$. We say that it is an interesting vertex if $S(\fset(v))\neq \emptyset$. Throughout the construction of the tree $\tset$ we maintain the following invariant: if a vertex $v$ of $\tset$ is an interesting level-$i$ vertex, for some $0\leq i\leq \log L_1-1$, then:

\[|\opt_{\fset(v)}|\geq \frac{|\opt_{\fset}|}{L_1^i}.\]

We also ensure that for every vertex $v$ of $\tset$, $|\fset(v)|\leq L_1$, and $\fset(v)$ is aligned with $G$.

We start our construction with the root $v_0$ of $\tset$, and we set $\fset(v_0)=\fset$. Clearly, our invariant holds for $\tset$. While there is an interesting leaf vertex $v$ in the tree, with $|\fset(v)|>L_2$, select any such vertex, and assume that $v$ belongs to some level $i$, for $0\leq i\leq \delta$. Recall that from our invariant, $|\opt_{\fset(v)}|\geq  \frac{|\opt_{\fset}|}{L_1^i}\geq  \frac{|\opt_{\fset}|}{L_1^{\delta}}$. From Observation~\ref{obs: rho accurate for big subinstances}, $G$ remains a $\rho'=\frac{\rho}{L_1^{\delta}}$-accurate grid for $\fset(v)$.  Since we have assumed that $\rho>32L_1^{\delta+2}$, we get that $\rho'\geq 32L_1^2\geq 32|\fset(v)|^2$, and since we have assumed that $|\opt_{\fset}|\geq 512 L_1^{\delta+2}$, we get that $|\opt_{\fset(v)}|\geq 512 L_1^2$.

Therefore, we can apply Corollary~\ref{corollary: partition with grid to decrease boundary} to $\fset$, to obtain a valid decomposition pair $(\fset_1,\fset_2)$ for $\fset$, such that both $\fset_1$ and $\fset_2$ are aligned with $G$, and $|\fset_1|,|\fset_2|\leq 3|\fset|/4$. Assume without loss of generality that $|\opt_{\fset_1}|\leq |\opt_{\fset_2}|$. We add two children, $v_1$ and $v_2$ to the tree $\tset$. 
If $|\opt_{\fset_1}|\geq \frac{|\opt_{\fset(v)}|}{L_1}$, then we set $\fset(v_1)=\fset_1$ and $\fset(v_2)=\fset_2$. Notice that both $v_1$ and $v_2$ now belong to level $(i+1)$, and:

\[|\opt_{\fset_1}|,|\opt_{\fset_2}|\geq \frac{|\opt_{\fset(v)}|}{L_1}\geq \frac{|\opt_{\fset}|}{L_1^{i+1}},\]
so our invariant continues to hold.

Otherwise, we let $\fset(v_1)=\set{B}$, so $S(\fset(v_1))=\emptyset$, and we let $\fset(v_2)=\fset_2$. As before, $v_2$ is a level-$(i+1)$ vertex, and the invariant continues to hold.

Recall that $|\opt_{\fset(v)}|-\left(|\opt_{\fset_1}|+|\opt_{\fset_2}|\right)\leq \frac{\tc\cdot |\opt_{\fset(v)}|}{|\fset(v)|}\leq \frac{\tc\cdot |\opt_{\fset(v)}|(4/3)^{i+1}}{L_1}$, since we assumed that $v$ lies at level $i$. Additionally, we may discard up to $\frac{|\opt_{\fset(v)}|}{L_1}$ rectangles if $|\opt_{\fset_1}|<\frac{|\opt_{\fset(v)}|}{L_1}$. Therefore, in total:

\[\lambda(v)=|\opt_{\fset(v)}|-\left(|\opt_{\fset(v_1)}|+|\opt_{\fset(v_2)}|\right )\leq \frac{2\tc\cdot |\opt_{\fset(v)}|(4/3)^{i+1}}{L_1}.\]

The algorithm terminates when for every leaf vertex $v$, either $|\fset(v)|\leq L_2$, or $S(\fset(v))=\emptyset$. It is immediate to verify that the algorithm constructs a valid cleanup tree, and for every leaf vertex $v\in \lset$, with $S(\fset(v))\neq \emptyset$, we get that $|\fset(v)|\leq L_2$, and  $|\opt_{\fset(v)}|\geq |\opt_{\fset}|/L_1^{\delta+2}$. We now only need to bound the loss of the tree. Let $\iset=\iset(\tset)$ denote the set of all inner vertices of $\tset$. 
For all $0\leq i\leq \delta$, let $U_i$ denote the set of all inner vertices of $\tset$ that belong to level $i$. Then $\iset=\bigcup_{i=0}^{\delta}U_i$. If $v,v'\in U_i$, then neither can be a descendant of the other in $\tset$, and so from Observation~\ref{obs: partitioning tree - disjointness of non-descendants}, $\sum_{v\in U_i}|\opt_{\fset(v)}|\leq |\opt_{\fset}|$. Since the loss at every level-$i$ vertex $v$ is bounded by $\frac{2\tc\cdot |\opt_{\fset(v)}|(4/3)^{i+1}}{L_1}$, we get that:

\[\sum_{v\in U_i}\lambda(v)\leq \sum_{v\in U_i}\frac{2\tc\cdot |\opt_{\fset(v)}|(4/3)^{i+1}}{L_1}\leq \frac{2\tc\cdot |\opt_{\fset}|(4/3)^{i+1}}{L_1},\]

and so overall:

\[\Lambda(\tset)=\sum_{i=0}^{\delta}\sum_{v\in U_i}\lambda(v)\leq 2\tc\cdot |\opt_{\fset}|\sum_{i=0}^{\delta}\frac{(4/3)^{i+1}}{L_1}\leq \frac{12\tc |\opt_{\fset}|}{L_2}.\]
\end{proof}

%-------------------------------------------------------
%-------------------------------------------------------
%-------------------------------------------------------
%-------------------------------------------------------
\subsection{The Algorithm}\label{subsec: alg}
%-------------------------------------------------------
%-------------------------------------------------------
%-------------------------------------------------------
Recall that all logarithms in this section are to the base of $4/3$. 
For convenience, we denote $|\opt|$ by $N$.
We assume that $\eps>1/\log^4N$, since otherwise the $(1-\eps)$-approximation algorithm $\aset^*$ with running time $n^{O(\log^2N/\eps^2)}$ from Section~\ref{sec: running time exp log n} gives an $(1-\eps)$-approximation in time $n^{O(\sqrt {\log N}/\eps^3)}$.
 We use three parameters: $L^*_1=\frac{100\tc \sqrt{\log N}}{\eps}$, $L_2^*=\frac{100 \tc \log N}{\eps}$, and $\rho=\left(\frac 4 3\right )^{2\sqrt{ \log N}}$, where $\tc$ is the parameter from Corollary~\ref{corollary: main partition with grid}. Let $\delta=\ceil{\log(L_2^*/L_1^*)}=\Theta(\log \log N)$, and let $\eta=512(L_2^*)^{\delta+3}\cdot (4/3)^{\sqrt{\log N}}=(\log N)^{O(\log\log N)}\cdot (4/3)^{\sqrt{\log N}}=2^{\Theta(\sqrt{\log N})}$. We will assume that $N$ is large enough, so, for example, $\left(\frac 4 3\right )^{\sqrt{ \log N}}>32(L_2^*)^{\delta+3}=2^{\Theta((\log\log N)^2)}$, as otherwise $N$ is bounded by some constant and the problem can be solved efficiently via exhaustive search. 

We define the family $\cset$  of important sets of fake rectangles in two steps. Set $\cset$ will consist of two subsets $\cset_1$ and $\cset_2$, that are defined at step 1 and step 2, respectively. %Family $\cset'\subseteq \cset$ of basic important sets of fake rectangles will consist of two subsets, $\cset_1'\subseteq \cset_1$ and $\cset_2'\subseteq \cset_2$.

\paragraph{Step 1: Family $\cset_1$.}
%Let $L^*_1=\frac{x\sqrt{\log |\opt|}}{\eps}$.
Family $\cset_1$ contains all valid sets $\fset$ of fake rectangles, such that $|\fset|\leq L^*_1$, and all rectangles in $\fset$ have integral coordinates between $0$ and $2n+1$. It is immediate to verify that $|\cset_1|=n^{O(\sqrt{\log N}/\eps)}$. Notice that $\set{\emptyset}\in\cset_1$.

\paragraph{Step 2: Family $\cset_2$.}
%Let $L_2^*=\frac{y\sqrt{\log |\opt|}}{\eps}$.
Consider now any important set $\fset\in \cset_1$ of fake rectangles. We define a collection $\cset_2(\fset)$ of sets of fake rectangles, and we will eventually set $\cset_2=\bigcup_{\fset\in \cset_1}\cset_2(\fset)$. 

In order to define the family $\cset_2(\fset)$ of fake rectangles, we apply Claim~\ref{claim: find a rho-accurate grid} to construct a $\rho$-accurate grid $G$ for $\fset$, so that all vertices of $G$ have integral coordinates. The size of the grid is $(z\times z)$, where $z=O(\rho \log\log N+|\fset|)\leq O(\rho\log\log N+\sqrt{ \log N}/\eps)\leq 2^{O(\sqrt{\log N})}$. We then let $\cset_2(\fset)$ contain all valid sets $\fset'$ of fake rectangles, with $S(\fset')\subseteq S(\fset)$ and $|\fset'|\leq L_2^*$, such that $\fset'$ is aligned with $G$. %Additionally, we add the set $\set{B}$ to $\cset_2(\fset)$.
 Notice that $\cset_2(\fset)\cap \cset_1\neq \emptyset$, as for example, both families contain the sets $\fset$ and $\set{B}$.

Clearly, $|\cset_2(\fset)|\leq z^{O(L_2^*)}=2^{O(\log^{3/2}N/\eps)}=n^{O(\sqrt{\log N}/\eps)}$, and we can compute the family $\cset_2(\fset)$ in time $n^{O(\sqrt{\log N}/\eps)}$.% We define $\cset_2'(\fset)\subseteq \cset_2(\fset)$ to contain all sets $\fset'$ of fake rectangles, such that $|\fset'|>L_1^*$, and $\aset(\fset')<\frac{\aset(\fset)}{(4/3)^{\sqrt{\log N}}}$. For all such sets $\fset'\in \cset_2'(\fset)$, the solution associated with $\fset'$ is $\emptyset$.

Finally, we set $\cset_2=\bigcup_{\fset\in \cset_1}\cset_2(\fset)$, and $\cset=\cset_1\cup \cset_2$. Then $|\cset|\leq |\cset_1|+ \sum_{\fset\in \cset_1}|\cset_2(\fset)|\leq |\cset_1|\cdot n^{O(\sqrt{\log N}/\eps)}\leq n^{O(\sqrt{\log N}/\eps)}$.

We now define the family $\cset'\subseteq \cset$ of basic sets of fake rectangles, and the corresponding algorithms $\aset'$ and $\aset''$. 
Recall that $\aset$ is the $(c_A\log\log |\opt|)$-approximation algorithm for \MISR from Corollary~\ref{cor: an log log opt approx}, and for any valid set $\fset$ of fake rectangles, we denote by $\aset(\fset)$ the value of the solution produced by algorithm $\aset$ on input $\rset(\fset)$.
Family  $\cset'$ contains all sets $\fset\in \cset$ with $\aset(\fset)<\eta$, and we use the algorithm $\aset$ in order to identify the sets $\fset\in \cset'$. Notice that if $\fset\in \cset'$, then $|\opt_{\fset}|\leq O(\eta \log\log N)\leq 2^{O(\sqrt{ \log N})}$. We can compute an $(1-\eps/2)$-approximate solution to each such instance $\rset(\fset)$ in time $n^{O(\sqrt{\log N}/\eps^3)}$, using the $(1-\eps)$-approximation algorithm from Section~\ref{sec: running time exp log n}, whose running time is $n^{O(\log |\opt_{\fset}|/\eps^3)}=n^{O(\sqrt{\log N}/\eps^3)}$. We employ this algorithm as $\aset''$.

 We can now use the dynamic programming-based algorithm from Section~\ref{sec: alg outline}. 
The initialization step takes time at most $|\cset|\cdot n^{O(\sqrt{\log N}/\eps^3)}=n^{O(\sqrt{\log N}/\eps^3)}$, and the rest of the algorithm runs in time $O(|\cset|^4)=n^{O(\sqrt{\log N}/\eps)}$, so the total running time is $n^{O(\sqrt{\log N}/\eps^3)}$. It now remains to show that the algorithm computes a solution of value at least $(1-\eps)|\opt|$. As before, we do so using partitioning trees.

%-------------------------------------------------------
%-------------------------------------------------------
%-------------------------------------------------------
%-------------------------------------------------------
\subsection{Analysis}\label{subsec: analysis}
%-------------------------------------------------------
%-------------------------------------------------------
%-------------------------------------------------------
In this section, we analyze the algorithm, by constructing the partitioning tree $\tset$. Our tree $\tset$ will be composed of a number of smaller trees, that we compute using the following theorem.

\begin{theorem}\label{thm: inner tree}
For every set $\fset\in \cset_1\setminus\cset'$ of fake rectangles, there is a partitioning tree $\tset(\fset)$, such that for every leaf vertex $v\in \lset(\tset(\fset))$, either (i) $S(\fset)= \emptyset$; or (ii) $\fset(v)\in \cset_1$ and $|\opt_{\fset(v)}|\leq\frac{|\opt_{\fset}|}{(4/3)^{\sqrt{\log N}}}$. The loss of the tree $\Lambda(\tset)\leq \frac{24\tilde c|\opt_{\fset}|}{L_1^*}$.
\end{theorem}

\begin{proof}
Let $G$ be the $\rho$-accurate grid that we have computed for $\fset$.
Our initial tree $\tset(\fset)$ consists of a single vertex $v$, with $\fset(v)=\fset$.
The construction of the tree $\tset(\fset)$ consists of two stages. The first stage is executed as long as there is any leaf vertex $v$ in $\tset(\fset)$ with $|\opt_{\fset(v)}|\geq\frac{|\opt_{\fset}|}{(4/3)^{\sqrt{\log N}}}$. During this stage, we will ensure that throughout its execution, for every vertex $v$ of the tree, if $S(\fset(v))\neq \emptyset$, then $|\opt_{\fset(v)}|\geq \frac{|\opt_{\fset}|}{L_2^*\cdot(4/3)^{ \sqrt{\log N}}}$.
Notice that, since $\fset\not\in \cset'$, $|\opt(\fset)|\geq \aset(\fset)\geq \eta= 512(L_2^*)^{\delta+3}\cdot (4/3)^{\sqrt{\log N}}$, and so for each vertex $v$ of the tree with $S(\fset(v))\neq \emptyset$, we get that $|\opt_{\fset(v)}|\geq \frac{|\opt_{\fset}|}{L_2^*\cdot (4/3)^{\sqrt{\log N}}}\geq 512(L_2^*)^{\delta+2}$.

Clearly, the invariant holds at the beginning of the algorithm.
 In every iteration of the first stage, we consider some leaf vertex $v$ with $|\opt_{\fset(v)}|\geq\frac{|\opt_{\fset}|}{(4/3)^{\sqrt{\log N}}}$.
Let $\rho'=\frac{\rho}{(4/3)^{\sqrt{\log N}}}=(4/3)^{\sqrt{\log N}}$. From Observation~\ref{obs: rho and rho' accurate grids}, grid $G$ remains a $\rho'$-accurate grid for $\fset(v)$.

Since we are guaranteed that $\left(\frac 4 3\right )^{\sqrt{ \log N}}>32(L_2^*)^{\delta+3}$, we get that $\rho'\geq 32(L^*_2 )^2$, and we can apply Corollary~\ref{corollary: main partition with grid} to obtain a valid decomposition triple $(\fset_1,\fset_2,\fset_3)$ for $\fset(v)$, where for each $1\leq i\leq 3$, $|\fset_i|\leq L^*_2$ and $\fset_i$ is aligned with $G$, and so $\fset_i\in \cset_2(\fset)$. 

Assume without loss of generality that $|\opt_{\fset_1}|\leq |\opt_{\fset_2}|\leq |\opt_{\fset_3}|$. Notice that, since $|\opt_{\fset_3}|\leq 3|\opt_{\fset(v)}|/4$, and $\sum_{i=1}^3|\opt_{\fset_i}|\geq |\opt_{\fset(v)}|\left(1-\frac{\tilde c}{L^*_2}\right)$, we are guaranteed that 
$|\opt_{\fset_1}|+|\opt_{\fset_2}|\geq |\opt_{\fset(v)}|/8$, and so $|\opt_{\fset_2}|\geq |\opt_{\fset(v)}|/16$. We add three new vertices $v_1,v_2,v_3$ to the tree as the children of $v$, and we set $\fset(v_2)=\fset_2$ and $\fset(v_3)=\fset_3$. Notice that both $\fset_2,\fset_3\in \cset_2(\fset)$, and the invariant holds for them.

If $|\opt_{\fset_1}|\geq \frac{|\opt_{\fset(v)}|}{L_2^*}$, then we set $\fset(v_1)=\fset_1$, and otherwise we set $\fset(v_1)=\set{B}$. It is easy to see that our invariant continues to hold, and $\lambda(v)\leq \frac{2\tc |\opt_{\fset(v)}|}{L_2^*}$. This completes the description of the first stage. Let $\lset'$ be the set of all leaf vertices at the end of the first stage, and let $\iset'$ be the set of all inner vertices. We now bound the total loss $\sum_{v\in \iset'}\lambda(v)$, as follows. Notice that the longest root-to-leaf path in $\tset$ has length at most $\sqrt{\log N}$. We partition the vertices of $\iset'$ into $\sqrt{\log N}$ classes, where class $U_i$, for $1\leq i\leq \sqrt{\log N}$ contains all vertices $v$, such that the unique path from $v$ to the root of the tree contains exactly $i$ vertices. As before, if two vertices $v,v'\in U_i$, then neither of them is a descendant of the other, and so $\sum_{v\in U_i}|\opt_{\fset(v)}|\leq |\opt_{\fset}|$. We can now bound the total loss of all vertices in class $i$ by:

\[\sum_{v\in U_i}\lambda(v)\leq \sum_{v\in U_i}\frac{2\tc |\opt_{\fset(v)}|}{L_2^*}\leq \frac{2\tc |\opt_{\fset}|}{L_2^*}.\]

Overall, $\sum_{v\in \iset'}\lambda(v)\leq \sum_{i=1}^{\sqrt{\log N}}\sum_{v\in U_i}\lambda(v)\leq \frac{2\tc |\opt_{\fset}|\sqrt{\log N}}{L_2^*}$.

We now proceed to describe the second stage of the algorithm. Consider some leaf vertex $v\in \lset'$, such that $S(\fset(v))\neq \emptyset$. Our invariant guarantees that $|\opt_{\fset(v)}|\geq \frac{|\opt_{\fset}|}{L_2^*\cdot(4/3)^{ \sqrt{\log N}}}\geq \frac{\eta}{L_2^*\cdot(4/3)^{ \sqrt{\log N}}}\geq 512(L_2^*)^{\delta+2}$. Let $\rho''=\frac{\rho}{L_2^*\cdot(4/3)^{ \sqrt{\log N}}}=\frac{(4/3)^{\sqrt{\log N}}}{L_2^*}$.
 Then $G$ remains a $\rho''$-accurate grid for $\fset(v)$, from Observation~\ref{obs: rho and rho' accurate grids}. Since we have assumed that $\left(\frac 4 3\right )^{\sqrt{ \log N}}>32(L_2^*)^{\delta+3}$, we get that $\rho''\geq 32(L^*_2)^{\delta+2}$. Therefore, we can construct an $L_2^*$--$L_1^*$-cleanup tree $\tset'(v)$ for $\fset(v)$ and $G$. From the definition of the cleanup tree, it is easy to verify that for every vertex $v'\in V(\tset'(v))$, $\fset(v')\in \cset_2(\fset)\subseteq \cset$. Moreover, if $v'$ is a leaf of $\tset'(v)$ with $S(\fset(v'))\neq \emptyset$, then $|\fset(v')|\leq L_1^*$, and so $\fset(v')\in \cset_1$, as required. %Moreover, for every leaf vertex $v'$ of the tree $\tset'(v)$, $|\opt_{\fset(v')}|\geq\frac{|\opt_{\fset(v)}|}{(L_2^*)^{\delta+2}}\geq \frac{|opt_{\fset}|}{(L_2^*)^{\delta+3}\cdot (4/3)^{\sqrt{\log N}}}$.
 
  Once we add a clean-up tree to each vertex $v\in \lset'$ with $S(\fset(v))\neq \emptyset$, we obtain the final tree $\tset(\fset)$. 
  Recall that the loss of the cleanup tree $\tset'(v)$ is at most $\frac{12\tc |\opt_{\fset(v)}|}{L_1^*}$. Since for every pair $v_1,v_2\in \lset'$ of vertices, neither vertex is a descendant of the other in tree $\tset(\fset)$, we get that the total loss of all cleanup trees is:
  
  \[\sum_{v\in \lset'}\Lambda(\tset'(v))\leq \sum_{v\in \lset'}\frac{12\tc |\opt_{\fset(v)}|}{L_1^*}\leq \frac{12\tc |\opt_{\fset}|}{L_1^*}.\]
  
 The total loss of the tree $\tset(\fset)$ can now be bounded as follows:
 
 \[\Lambda(\tset(\fset))=\sum_{v\in \iset(\tset(\fset))}\lambda(v)=\sum_{v\in \iset'}\lambda(v)+\sum_{v\in \lset'}\Lambda(\tset'(v))\leq 
  \frac{2\tc |\opt_{\fset}|\sqrt{\log N}}{L_2^*}+ \frac{12\tc |\opt_{\fset}|}{L_1^*}\leq  \frac{24\tc |\opt_{\fset}|}{L_1^*},\]

  since $L_2^*=L_1^*\cdot \sqrt{\log N}$.
\end{proof}

We are now ready to complete the construction of the final partitioning tree $\tset$. The construction consists of $\sqrt{\log N}$ phases. We start with tree $\tset$ containing a single vertex $v_0$, associated with the important set of fake rectangles $\fset(v_0)=\set{\emptyset}$. Throughout the execution of the algorithm, we ensure that if $v$ is a leaf vertex of the current tree $\tset$, then either $S(\fset(v))=\emptyset$, or $\fset(v)\in \cset_1$. The invariant is clearly true at the beginning of the algorithm.

 In order to execute the $i$th phase, we let $U_i$ contain all leaf vertices $v$ of the current tree $\tset$ with $\fset(v)\not\in\cset'$. For every vertex $v\in U_i$, we construct the tree $\tset(v)$ given by Theorem~\ref{thm: inner tree}, and add it to $\tset$, by identifying its root vertex with $v$. Let $\lset(v)$ be the set of all leaf vertices of the tree $\tset(v)$. Recall that $\Lambda(\tset(v))=|\opt(\fset(v))|-\sum_{v'\in \lset(v)}|\opt(\fset(v'))|\leq \frac{24\tc |\opt_{\fset(v)}|}{L_1^*}$. Notice that for every leaf vertex $v'\in \lset(\tset(v))$, we are guaranteed that either $S(\fset(v'))=\emptyset$, or $\fset(v')\in \cset_1$, and in the latter case $|\opt_{\fset(v')}|\leq \frac{|\opt_{\fset(v)}|}{(4/3)^{\sqrt {\log N}}}$.
Since for every pair $v_1,v_2\in U_i$ of vertices, neither vertex is a descendant of the other in $\tset$, $\sum_{v\in U_i}|\opt_{\fset(v)}|\leq |\opt|$, and so:

\[\sum_{v\in U_i}\Lambda(\tset(v))\leq \sum_{v\in U_i} \frac{24\tc |\opt_{\fset(v)}|}{L_1^*}\leq \frac{24\tc |\opt|}{L_1^*}.\]

Clearly, after at most $\sqrt{\log N}$ phases, we obtain a valid partitioning tree $\tset$, such that for every leaf vertex $v$ of $\tset$, $\fset(v)\in \cset'$. It is easy to verify that the total loss of the tree $\tset$ is bounded by:

\[\Lambda(\tset)\leq \sum_{i=1}^{\sqrt{\log N}}\sum_{v\in U_i}\Lambda(\tset(v))\leq \frac{24\tc |\opt|\sqrt{\log N}}{L_1^*}\leq \frac{|\opt|\eps}{2},\]

since $L_1^*=100\tc \sqrt{\log N}/\eps$. From Observation~\ref{obs: final analysis},  we conclude that our algorithm computes a $(1-\eps)$-approximate solution, in time $n^{O(\sqrt{\log N}/\eps^3)}$.

%-------------------------------------------------------
%-------------------------------------------------------
%-------------------------------------------------------
%-------------------------------------------------------
%-------------------------------------------------------
%-------------------------------------------------------
%-------------------------------------------------------
%-------------------------------------------------------
%-------------------------------------------------------
%-------------------------------------------------------

\label{------------------------------------------sec: n to polyloglog opt running time----------------------------}
\section{A QPTAS with Running Time $n^{O((\log\log |\opt|)^4/\eps^4)}$}
\label{sec: n to polyloglog opt running time}

We start with an intuitive high-level overview of the algorithm. This overview is over-simplified and imprecise, and it is only intended to provide intuition.
A natural way to further improve the running time of the QPTAS from Section~\ref{sec: running time exp sqrt log n} is to use more levels of recursion, namely: instead of just two sets $\cset_1,\cset_2\subseteq \cset$, we will have $h=\Theta(\log \log N)$ such sets, where we refer to the sets $\fset\in \cset_i$ as \emph{level-$i$ sets}, and to corresponding instances $\rset(\fset)$ as \emph{level-$i$ instances}. We will also use parameters $L_1,\ldots,L_h$ associated with the instances of different levels. As before, family $\cset_1$ will contain all valid sets $\fset$ of fake rectangles, whose corners have integral coordinates, and $|\fset|\leq L_1$. For each $1<i\leq h$, for every set $\fset\in \cset_{i-1}$ of fake rectangles, we will define a family $\cset_i(\fset)$ of sets of fake rectangles, as follows. We compute a $\rho_{i-1}$-accurate grid $G_{i-1}$ for $\fset$, for an appropriately chosen parameter $\rho_{i-1}$, and we let $\cset_i(\fset)$ contain all valid sets $\fset'$ of fake rectangles that are aligned with $G_{i-1}$, such that $|\fset'|\leq L_i$ and $S(\fset')\subseteq S(\fset)$. We then set $\cset_i=\bigcup_{\fset\in \cset_{i-1}}\cset_i(\fset)$. Notice that the same set $\fset'$ of fake rectangles may belong to several families $\cset_i(\fset)$. It will be convenient in our analysis to view each such set as a separate set (though the algorithm does not distinguish between them),  and to keep track of the sets of fake rectangles from $\cset_1,\ldots,\cset_{i-1}$, and their corresponding grids $G_1,\ldots,G_{i-1}$, that were used to create the set $\fset'$. In order to do so, we will denote each level-$i$ instance by $\ffi=(\fset_1,G_1,\ldots,\fset_{i-1},G_{i-1},\fset_i)$, where for $1\leq i'<i$, $G_{i'}$ is a $\rho_{i'}$-accurate grid for $\fset_{i'}$, though only the set $\fset_i$ is added to $\cset_i$. Before we proceed to a formal definition of the sets $\cset_i$, we need the following definition.

\begin{definition} Let $G=(\vset,\hset)$, $G'=(\vset',\hset')$ be two grids. We say that $G'$ is \emph{aligned} with $G$ iff $\vset'\subseteq \vset$ and $\hset'\subseteq \hset$.
\end{definition}

\begin{claim}\label{claim: grid-aligned rho-accurate grid}
Let $\fset$ be any valid set of fake rectangles, and let $G$ be a $\rho$-accurate grid for $\fset$, for some parameter $\rho\geq 1$. Then for any $1\leq \rho'\leq \rho$, we can efficiently construct a $\rho'$-accurate grid $G'$ for $\fset$ of size $(z\times z)$, where $z\leq 2(4\rho'c_A\log\log(|\opt_{\fset}|)+2|\fset|)$, such that $G'$ is aligned with $G$.
\end{claim}

\begin{proof}
We start by constructing a $\rho'$-accurate grid $G''=(\vset'',\hset'')$ for $\fset$, of size $(z'\times z')$, where $z'\leq4\rho'c_A\log\log(|\opt(\fset)|)+2|\fset|$, using Claim~\ref{claim: find a rho-accurate grid}. In order to construct our final grid $G'=(\vset',\hset')$, start with $\vset'=\emptyset$. For each vertical line $V\in \vset''$, if $V\in \vset$, then we add $V$ to $\vset'$. Otherwise, we add to $\vset'$ two vertical lines of $\vset$: one that lies immediately to the left of $V$, and one that lies immediately to the right of $V$. This finishes the definition of the set $\vset'$ of vertical lines of $G'$. Notice that every vertical strip of $G'$ is either contained in some vertical strip of $G''$, or it is contained in some vertical strip of $G$. Since $\rho\geq \rho'$, the maximum number of mutually disjoint rectangles contained in each vertical strip of $G'$ is at most $\ceil{|\opt(\fset)|/\rho'}$, as required. 
The set $\hset'$ of the horizontal lines of $G'$ is constructed similarly.
\end{proof}

All logarithms in this section are to the base of $2$, unless stated otherwise. For convenience of notation, we denote $\exp(i)=2^i$.
We denote by $N$ the smallest integral power of $2$, such that $N\geq |\opt|$. We assume that $\eps>1/\log N$, since otherwise the $(1-\eps)$-approximation algorithm with running time $n^{O(\log N/\eps^3)}$ from Section~\ref{sec: running time exp log n} has running time $n^{O(1/\eps^4)}$. We assume that $N$ is large enough, so, for example, $\log N>\tc\cdot c_A(\log\log N)^5$, as otherwise $N$ is bounded by some constant and the problem can be solved efficiently via exhaustive search.
Recall that $\aset$ is the $(c_A\log\log |\opt|)$-approximation algorithm for \MISR from Corollary~\ref{cor: an log log opt approx}, and for any valid set $\fset$ of fake rectangles, we denote by $\aset(\fset)$ the value of the solution produced by algorithm $\aset$ on input $\rset(\fset)$.

\subsection{Parameter Setting}
We start with $h^*=\log\log N$.

For each $1\leq i\leq h^*$, we define a parameter $L_i=\frac{\tc\cdot (\log\log N)^3\cdot 2^{i}}{\eps}$, that will serve as the bound on the number of fake rectangles in each set $\fset\in \cset_i$. Notice that $L_1<L_2<\cdots<L_{h^*}=\frac{\tc\log N(\log\log N)^3}{\eps}$. 
Notice also that for $1\leq i<h^*$, $\ceil{\log_{4/3}(L_{i+1}/L_i)}=\ceil{\log_{4/3}2}= 3$, and we denote this value by $\delta$. 

We let $\eta=32L_{h^*}^{2\delta+4}$. Since we have assumed that $\eps>1/\log N$ and $N$ is large enough, it is easy to verify that $\log N<\eta\leq \log^{O(1)}N$.

For $1\leq i\leq h^*$, we define $\rho_i=N^{1/2^i}$. Clearly, for all $1< i\leq h^*$, $\rho_{i}=\sqrt{\rho_{i-1}}$.
We let $h$ be the largest integer, so that $\rho_{h}>\eta^{320}$. It is easy to verify that $\eta^{320}<\rho_{h}\leq\eta^{640}$, and $h<h^*$, as $\rho_{h^*}=1$. The number of the recursive levels in our construction will be $h$.
Finally, we define the value $\tau^*=\rho_{h-1}^3=(\log N)^{\Theta(1)}$.
From our definitions, we immediately obtain the following inequalities. First, for all $1\leq i <h$ and $1\leq j\leq h$,

\begin{equation}
\rho_i\geq (32 L_j^{2\delta+4})^{320}, \label{eq: rho bound}
\end{equation}

since $\rho_i\geq \rho_{h-1}\geq \eta^{320}\geq (32 L_{h^*}^{2\delta+4})^{320}\geq (32 L_j^{2\delta+4})^{320}$. Moreover, if $\fset$ is any valid set of fake rectangles with $|\opt_{\fset}|\geq \rho_{h-1}$, then for all $1\leq j\leq h$:

\begin{equation}
|\opt_{\fset}|\geq 512 L_j^{2\delta+4}, \label{eq: opt bound}
\end{equation}

using a similar reasoning as above.

 The family $\cset$ of important sets of fake rectangles will eventually be a union of $h$ subsets $\cset_1,\ldots,\cset_{h}$. 
Recall that the execution of the algorithm from Section~\ref{sec: running time exp sqrt log n} was partitioned into $O(\sqrt{\log N})$ phases, where the value of the optimal solution went down by a factor of roughly $2^{\Theta(\sqrt{\log N})}$ in every phase. At the end of every phase, we reduced the boundary complexities of all resulting instances from $L_2^*$ to $L_1^*$. This corresponded to adding clean-up trees to the partitioning trees $\tset(\fset)$ for $\fset\in \cset_1$. 

The algorithm in this section consists of $h$ recursive levels. The execution of the algorithm at every level is partitioned into a number of phases. The optimal solution value in each phase of level $i$ goes down by a factor of at least $(\rho_i)^{1/160}$. We then reduce the boundary complexities of the resulting level-$(i+1)$ instances from $L_{i+1}$ to $L_{i}$.

\subsection{The Algorithm}
Our algorithm uses the framework defined in Section~\ref{sec: alg outline}. Therefore, it is sufficient to define the family $\cset$ of important sets of fake rectangles, the family $\cset'\subseteq \cset$ of basic sets of fake rectangles, and the algorithms $\aset'$ and $\aset''$ for recognizing and approximately solving the instances corresponding to the basic sets of fake rectangles, respectively. %For convenience of 

For all $1\leq i\leq h$, it will be convenient to denote the level-$i$ sets of fake rectangles by $\ff i=(\fset_1,G_1,\ldots,G_{i-1},\fset_i)$, where for $1\leq i'<i$, $G_{i'}$ is a $\rho_{i'}$-accurate grid for $\fset_{i'}$, and all rectangles in $\fset_{i'+1}$ are aligned with $G_{i'}$. We also require that for all $1<i'<i$, grid $G_{i'}$ is aligned with $G_{i'-1}$, and that for all $1<i'\leq i$, $S(\fset_{i'})\subseteq S(\fset_{i'-1})$.

\paragraph{Level-$1$ Instances}
We let $\cset_1$ denote all valid sets $\fset$ of fake rectangles with $|\fset|\leq L_1$, such that all rectangles in $\fset$ have integral coordinates. For each $\fset\in \cset_1$, we define the level-$1$ set $\ff1=(\fset)$ of fake rectangles to be consistent with our notation for higher-level sets. We denote $\tcset_1=\set{\ff1=(\fset)\mid \fset\in \cset_1}$. Notice that $|\cset_1|\leq n^{O(L_1)}=n^{O((\log\log N)^3/\eps)}$. Notice also that sets $\set{B}$, $\set{\emptyset}$ of fake rectangles belong to $\cset_1$.

\paragraph{Level-$i$ instances}
Fix some $1<i\leq h$. For every level-$(i-1)$ instance $\ff{i-1}\in \tcset_{i-1}$, we define a set $\tcset_i(\ff{i-1})$ of level-$i$ instances, and we let
$\tcset_i=\bigcup_{\ff{i-1}\in \tcset_{i-1}}\tcset_i(\ff{i-1})$. We now describe the construction of the set $\tcset_i(\ff{i-1})$.

We assume that $\ff{i-1}=(\fset_1,G_1,\fset_2,G_2,\ldots,G_{i-2},\fset_{i-1})\in \tcset_{i-1}$ is a level-$(i-1)$ set of fake rectangles, where for each $1\leq i'< i-1$, $G_{i'}$ is a $\rho_{i'}$-accurate grid for $\fset_{i'}$, and if $i'>1$, then grid $G_{i'}$ is aligned with grid $G_{i'-1}$.
Moreover, for all  $1<i'\leq i-1$, set $\fset_{i'}$ contains at most $L_{i'}$ fake rectangles, that are aligned with the grid $G_{i'-1}$.
%If $\aset(\ff{i-1})<\tau_i$, then $\tcset(\ff{i-1})=\emptyset$. Assume now that $\aset(\ff{i-1})\geq \tau_i$

If $i>2$ and $\aset(\fset_{i-1})<\aset(\fset_{i-2})/\rho_{i-2}^{1/10}$, then we set $\tcset(\ff i)=\emptyset$. Assume now that $i> 2$ and $\aset(\fset_{i-1})\geq \aset(\fset_{i-2})/\rho_{i-2}^{1/10}$. Then:

\[|\opt_{\fset_{i-1}}|\geq \aset(\fset_{i-1})\geq \frac{\aset(\fset_{i-2})}{\rho_{i-2}^{1/10}}\geq \frac{|\opt_{\fset_{i-2}}|}{c_A\log\log N\cdot \rho_{i-2}^{1/10}}\geq \frac{|\opt_{\fset_{i-2}}|}{\rho_{i-2}^{1/5}},\]

since $\rho_{i-2}\geq\rho_{h}\geq \eta^{320}\geq \log N$, and if $N$ is large enough, we can assume that $\rho_{i-2}^{1/10}\geq c_A\log\log N$.
From Observation~\ref{obs: rho accurate for big subinstances}, grid $G_{i-2}$ then remains $\rho_{i-2}^{4/5}$-accurate for instance $\rset(\fset_{i-1})$, and, since $\rho_{i-1}=\sqrt{\rho_{i-2}}<\rho_{i-2}^{4/5}$, we can use Claim~\ref{claim: grid-aligned rho-accurate grid} to compute a $\rho_{i-1}$-accurate grid $G_{i-1}$ for $\fset_{i-1}$, so that $G_{i-1}$ is aligned with $G_{i-2}$. If $i=2$, then we simply compute any $\rho_1$-accurate grid $G_1$ for $\fset_1$. In either case, the size of the grid is $(z\times z)$, where:

\[
z=O(\rho_{i-1}\log\log N+L_{i-1})=O\left(\frac{\rho_{i-1}\log N(\log\log N)^3}{\eps}\right)=O\left(\rho_{i-1}\log^3 N\right )\leq O(\rho_{i-1}^2),
\]

since we have assumed that $\eps>1/\log N$, and $\rho_{i-1}\geq \rho_h\geq \eta^{320}\geq \log^{320}N$.

 We construct the set $\tcset_i(\ff{i-1})$ as follows. For every valid set $\fset'$ of fake rectangles, with $S(\fset')\subseteq S(\fset_{i-1})$, and $|\fset'|\leq L_i$, such that the rectangles in $\fset'$ are aligned with the grid $G_{i-1}$, we add a level-$i$ set $\ff i=(\fset_1,G_1,\fset_2,G_2,\ldots,G_{i-2},\fset_{i-1},G_{i-1},\fset')$ to $\tcset_i(\ff{i-1})$. Notice that $(\fset_1,G_1,\ldots,\fset_{i-1},G_{i-1},\fset_{i-1})\in \tcset_i(\ff{i-1})$, and:
 
 \[\begin{split}
 |\tcset_i(\ff{i-1})|&=z^{O(L_i)}\leq \rho_i^{O(\exp(i)(\log\log N)^3/\eps)}\\
 &=(N^{1/\exp(i)})^{O(\exp(i)(\log\log N)^3/\eps)}\\
 &=N^{O((\log\log N)^3/\eps)}.
 \end{split}\]

We set $\tcset_i=\bigcup_{\ff{i-1}\in \tcset_{i-1}}\tcset_i(\ff{i-1})$, and we let $\cset_i$ contain all sets $\fset$ of fake rectangles, such that for some $\ff i=(\fset_1,G_1,\ldots,G_{i-1},\fset_i)\in \tcset_i$, $\fset=\fset_i$.
Finally, we set $\cset=\bigcup_{i=1}^{h}\cset_i$.

We say that $\fset\in \cset$ is a basic set of fake rectangles, and add it to $\cset'$, iff $\aset(\fset)\leq \tau^*$. We can use the algorithm $\aset$ to determine, for each set $\fset\in \cset$, whether $\fset$ is a basic set. If $\fset$ is a basic set, then $|\opt_{\fset}|\leq c_A\cdot \aset(\fset)\cdot \log\log N\leq c_A\tau^*\log\log N=(\log N)^{O(1)}$, and we can use the algorithm from Section~\ref{sec: running time exp sqrt log n} to compute a $(1-\eps/2)$-approximate solution to instance $\rset(\fset)$ in time $n^{O(\sqrt{\log |\opt_{\fset}|}/\eps^3)}=n^{O(\log\log N/\eps^3)}$. We use this algorithm as algorithm $\aset''$ for the initialization step of the dynamic program.
This completes the definition of the family $\cset$ of important sets of fake rectangles, the family $\cset'\subseteq \cset$ of basic sets of fake rectangles, and the algorithms $\aset'$ and $\aset''$. We then use the dynamic programming-based algorithm from Section~\ref{sec: alg outline} to solve the problem. In order to analyze the running time of the algorithm, we first need to bound $|\cset|$. As we showed above, 

\[|\tcset_1|\leq O\left(n^{L_1}\right )=n^{O((\log\log N)^3/\eps)},\]

and for all $1<i\leq h$,

\[|\tcset_i|=\sum_{\ff{i-1}\in \tcset_{i-1}}|\tcset(\ff{i-1})|\leq |\tcset_{i-1}|\cdot N^{O((\log\log N)^3/\eps)}.\]

Since $h<\log\log N$, it is immediate to verify that $|\cset|=O(|\tcset_{h}|)\leq n^{O((\log\log N)^4/\eps)}$.
The initialization step then takes time $|\cset|\cdot n^{O(\log\log N/\eps^3)}=n^{O((\log\log N)^4/\eps^3)}$, and the remainder of the algorithm runs in time $|\cset|^{O(1)}$. Therefore, the total running time of the algorithm is bounded by $n^{O((\log\log N)^4/\eps^3)}$.
The following simple observation will be useful for us later.

\begin{observation}\label{obs: math}
Suppose we are given two valid sets  $\fset,\fset'$ of fake rectangles, such that for some $1\leq i\leq h$, $|\opt_{\fset}|\geq |\opt_{\fset'}|/\rho_{i}^{1/20}$. Then $\aset(\fset)\geq \aset(\fset')/\rho_{i}^{1/10}$.
\end{observation}

\begin{proof}

\[\aset(\fset)\geq \frac{|\opt_{\fset}|}{c_A\log\log N}\geq \frac{|\opt_{\fset'}|}{\rho_{i}^{1/20}c_A\log\log N}\geq \frac{\aset(\fset')}{\rho_{i}^{1/20}c_A\log\log N}\geq \frac{\aset(\fset')}{\rho_{i}^{1/10}},\]

since $\rho_{i}\geq\rho_{h}\geq \eta^{320}\geq  \log N$, and if $N$ is large enough, we can assume that $\rho_{i}^{1/20}\geq c_A\log\log N$.

\end{proof}
%------------------------------------------------------------------------------------------------------------------------
%------------------------------------------------------------------------------------------------------------------------
%------------------------------------------------------------------------------------------------------------------------
%------------------------------------------------------------------------------------------------------------------------
\label{--------------------------------------------subsec: analysis-----------------------------------------------}

%------------------------------------------------------------------------------------------------------------------------
%------------------------------------------------------------------------------------------------------------------------
%------------------------------------------------------------------------------------------------------------------------
%----------------------------------------------------------------------------------------------------------------------
\subsection{Analysis}\label{sec: analysis}
Since the algorithm is guaranteed to produce a feasible solution to the MISR instance, it now remains to show that the value of the solution is within a factor of $(1-\eps)$ of the optimal one. As before, we do so by constructing a partitioning tree. The construction of the partitioning tree is recursive. We first construct partitioning trees for level-$(h-1)$ instances. For each $1\leq i<h-1$, we then show how to construct level-$i$ partitioning tree for each level-$i$ instance $\ff i\in \tcset_i$, by combining a number of level-$(i+1)$ partitioning trees. We will then use a number of level-$1$ partitioning trees in order to construct our final partitioning tree. 
We now define level-$i$ partitioning trees.

Fix any $1\leq i<h$, and let $\ff i\in \tcset_i$ be any level-$i$ set of fake rectangles, where $\ff i=(\fset_1,G_1,\fset_2,\ldots,G_{i-1},\fset_i)$, such that, if $i>1$, then $\aset(\fset_i)\geq \aset(\fset_{i-1})/\rho_{i-1}^{1/10}$. Let $G_i$ be the $\rho_i$-accurate grid that we have computed for $\fset_i$ when constructing $\tcset_{i+1}(\ff i)$.
A level-$i$ partitioning tree $\tset(\ff i)$ for $\ffi$ is a valid partitioning tree for $\fset_i$ (that is, the root of the tree is labeled by $\fset_i$), such that for every leaf vertex $v$ of $\tset(\ff i)$, either (i) $\fset(v)\in \cset'$, or (ii) $\fset(v)$ is aligned with $G_i$, $|\fset(v)|\leq L_i$, and $|\opt_{\fset_i}|/\rho_i^{1/40}\leq |\opt_{\fset(v)}|\leq |\opt_{\fset_i}|/\rho_i^{1/160}$.

We define $\Lambda_{h-1}=\frac{22\tc\log \eta}{L_{h-1}}$, and for $1\leq i<h-1$, we let $\Lambda_i=\left(2\Lambda_{i+1}+\frac{12\tc}{L_i}\right )$.
Following is the main theorem in our analysis.

\begin{theorem}\label{thm: analysis of main algorithm}
For every $1\leq i<h$, for every level-$i$ set $\ff i=(\fset_1,G_1,\fset_2,\ldots,G_{i-1},\fset_i)\in \tcset_i$ of fake rectangles, such that, if $i>1$, then $\aset(\fset_i)\geq \aset(\fset_{i-1})/\rho_{i-1}^{1/10}$, there is a level-$i$ partitioning tree $\tset(\ff i)$ for $\ff i$, whose loss is bounded by $\Lambda_i\cdot |\opt_{\fset_i}|$.
\end{theorem}

The majority of the remainder of this section is dedicated to the proof of Theorem~\ref{thm: analysis of main algorithm}. The proof is by induction on $i$, starting from $i=h-1$.

\subsubsection{Induction Basis: $i=h-1$.}

We assume that we are given a level-$(h-1)$ set $\ff {h-1}=(\fset_1,G_1,\fset_2,\ldots,G_{h-2},\fset_{h-1})$ of fake rectangles with $\aset(\fset_{h-1})\geq \aset(\fset_{h-2})/\rho_{h-2}^{1/10}$. Assume first that $|\opt_{\fset_{h-1}}|\leq \rho_{h-1}^3=\tau^*$. Then $\aset(\fset_{h-1})\leq |\opt_{\fset_{h-1}}|\leq \tau^*$, and $\fset_{h-1}\in \cset'$. We then let tree $\tset(\ff {h-1})$ contain a single vertex $v$ with $\fset(v)=\fset_{h-1}$. This is a valid level-$(h-1)$ partitioning tree for $\ff {h-1}$, and its loss is $0$.

We now assume that $|\opt_{\fset_{h-1}}|>\rho_{h-1}^3$. The construction of the tree $\tset(\ff {h-1})$ is very similar to the construction of the tree $\tset(\fset)$ in the proof of Theorem~\ref{thm: inner tree}, except that we use different parameters. 
Let $G_{h-1}$ be the $\rho_{h-1}$-accurate grid that we constructed for $\fset_{h-1}$ when computing $\tcset_h(\ff {h-1})$.
For convenience, we denote grid $G_{h-1}$ by $G$, and parameter $\rho_{h-1}$ by $\rho$.
We ensure that for every vertex $v$ of the tree, $\fset(v)$ is aligned with $G_{h-1}$ and $|\fset(v)|\leq L_h$, thus ensuring that $(\fset_1,G_1,\ldots,\fset_{h-1},G_{h-1},\fset(v))\in \tcset_h(\ff {h-1})$, and in particular $\fset(v)\in \cset$.
Our initial tree $\tset(\ff {h-1})$ consists of a single vertex $v$, with $\fset(v)=\fset_{h-1}$. From the definition of $G$, $\fset_{h-1}$ is aligned with $G$, and $|\fset_{h-1}|\leq L_{h-1}\leq L_h$.

The construction of the tree $\tset(\ff {h-1})$ consists of two stages. The first stage is executed as long as there is any leaf vertex $v$ in $\tset(\ff {h-1})$ with $|\opt_{\fset(v)}|\geq\frac{|\opt_{\fset_{h-1}}|}{\rho^{1/160}}$. We ensure that throughout the execution of the first stage, the following invariant holds: for every vertex $v$ of the tree, if $S(\fset(v))\neq \emptyset$, then $|\opt_{\fset(v)}|\geq \frac{|\opt_{\fset_{h-1}}|}{L_h\cdot \rho^{1/160}}$. %Since we have assumed that $|\opt_{\fset_{h-1}}|\geq \rho^2$, it follows that $\rho\leq \frac{|\opt_{\fset_{h-1}}|}{\rho}\leq \frac{|\opt_{\fset(v)}|}{64}$ (since $\rho\geq \eta$).
We use the following easy observation.

\begin{observation}\label{obs: bounds stage 1}
Let $\fset$ be any valid set of fake rectangles with $S(\fset)\subseteq S(\fset_{h-1})$, such that $\fset$ is aligned with $G$, and $|\opt_{\fset}|\geq\frac{|\opt_{\fset_{h-1}}|}{L_h\cdot \rho^{1/160}}$. Then $|\opt_{\fset}|\geq 512L_h^{\delta+2}$, and $G$ is a $\rho'$-accurate grid for $\fset$, for some $\rho'>32L_h^{\delta+2}$.
\end{observation}

\begin{proof}
Since we have assumed that $|\opt_{\fset_{h-1}}|\geq \tau^*=\rho^3$, from Equation~(\ref{eq: rho bound}), $|\opt_{\fset}|\geq \rho$, and from Equation~(\ref{eq: opt bound}), $|\opt_{\fset}|\geq 512L_h^{\delta+2}$.

For the second assertion, let $\rho'=\frac{\rho}{L_h\cdot \rho^{1/160}}$. From Observation~\ref{obs: rho accurate for big subinstances}, grid $G$ remains $\rho'$-accurate for $\fset$. Since from Equation~(\ref{eq: rho bound}), $\rho\geq (32L_h^{2\delta+4})^{320}$, we get that $\rho'> 32L_h^{\delta+2}$. 
\end{proof}

In every iteration of the first stage, we consider some leaf vertex $v$ with $|\opt_{\fset(v)}|\geq \frac{|\opt_{\fset_{h-1}}|}{\rho^{1/160}}$. From Observation~\ref{obs: bounds stage 1}, $|\opt_{\fset(v)}|\geq 512L_h^{\delta+2}$, and $G$ is a $\rho'$-accurate grid for $\fset(v)$, for some $\rho'> 32L_h^{\delta+2}$.
 Therefore, we can apply Corollary~\ref{corollary: main partition with grid} to obtain a valid decomposition triple $(\fset^1,\fset^2,\fset^3)$ for $\fset(v)$, where for each $1\leq j\leq 3$, $|\fset^j|\leq L_h$ and $\fset^j$ is aligned with $G$. Then for all $1\leq j\leq 3$, $(\fset_1,G_1,\ldots,\fset_{h-1},G_{h-1},\fset^j)\in \tcset_h(\ff {h-1})$, and in particular $\fset^j\in \cset$.

Assume without loss of generality that $|\opt_{\fset^1}|\leq |\opt_{\fset^2}|\leq |\opt_{\fset^3}|$. Notice that, since $|\opt_{\fset^3}|\leq 3|\opt_{\fset(v)}|/4$, and $\sum_{j=1}^3|\opt_{\fset^j}|\geq |\opt_{\fset(v)}|\left(1-\frac{\tilde c}{L_h}\right)$, we are guaranteed that 
$|\opt_{\fset^1}|+|\opt_{\fset^2}|\geq |\opt_{\fset(v)}|/8$, and so $|\opt_{\fset^2}|\geq |\opt_{\fset(v)}|/16\geq |\opt_{\fset_{h-1}}|/(L_h\cdot \rho^{1/160})$. We add three new vertices $v_1,v_2,v_3$ to the tree as the children of $v$, and we set $\fset(v_2)=\fset^2$ and $\fset(v_3)=\fset^3$. %Notice that the invariants hold for $\fset(v_2)$ and $\fset(v_3)$.

If $|\opt_{\fset^1}|\geq \frac{|\opt_{\fset(v)}|}{L_h}$, then we set $\fset(v_1)=\fset_1$, and otherwise we set $\fset(v_1)=\set{B}$. It is easy to see that our invariant continues to hold. The loss at the vertex $v$ is bounded by: $\lambda(v)=|\opt_{\fset(v)}|-\sum_{j=1}^3|\opt_{\fset(v_j)}|\leq \frac{2\tc |\opt_{\fset(v)}|}{L_h}$. This completes the description of the first stage. Let $\lset'$ be the set of all leaf vertices of $\tset(\ff{h-1})$ at the end of the first stage, and let $\iset'$ be the set of all its inner vertices. We now bound the total loss $\sum_{v\in \iset'}\lambda(v)$, as follows. Notice that the longest root-to-leaf path in $\tset(\ff {h-1})$ has length at most $\ell=\log_{4/3}\rho_{h-1}^{1/160}$, as the values $|\opt_{\fset(v)}|$ decrease by a factor of at least $3/4$ along the path. 
Recall that $h$ is the largest integer with $\rho_{h}\geq \eta^{320}$, so $\rho_{h}\leq \eta^{640}$, and  $\rho_{h-1}=\rho_h^2\leq \eta^{1280}$. Therefore, $\ell=\log_{4/3}\rho^{1/160}\leq \log_{4/3}\eta^{1280/160}\leq 20\log \eta$.

 We partition the vertices of $\iset'$ into $\ell$ classes, where class $U_j$, for $1\leq j\leq \ell$ contains all vertices $v$, such that the unique path from $v$ to the root of the tree contains exactly $j$ vertices. As before, if two vertices $v,v'\in U_j$, then neither of them is a descendant of the other, and so $\sum_{v\in U_j}|\opt_{\fset(v)}|\leq |\opt_{\fset_{h-1}}|$. We can now bound the total loss of all vertices in class $j$ by:

\[\sum_{v\in U_j}\lambda(v)\leq \sum_{v\in U_j}\frac{2\tc |\opt_{\fset(v)}|}{L_h}\leq \frac{2\tc |\opt_{\fset_{h-1}}|}{L_h}.\]

Overall, $\sum_{v\in \iset'}\lambda(v)\leq \sum_{j=1}^{\ell}\sum_{v\in U_j}\lambda(v)\leq \frac{2\tc\ell |\opt_{\fset_{h-1}}|}{L_h}\leq \frac{40\tc\log \eta}{L_h}\cdot |\opt_{\fset_{h-1}}|$.

We now proceed to describe the second stage of the algorithm. Consider some leaf vertex $v\in \lset'$, such that $S(\fset(v))\neq \emptyset$. Our invariant guarantees that $|\opt_{\fset(v)}|\geq \frac{|\opt_{\fset_{h-1}}|}{L_h\cdot \rho^{1/160}}$, and so from Observation~\ref{obs: bounds stage 1},  $|\opt_{\fset(v)}|\geq 512L_h^{\delta+2}$, and $G$ is a $\rho'$-accurate grid for $\fset(v)$, for some $\rho'> 32L_h^{\delta+2}$.
 Therefore, we can construct an $L_h$--$L_{h-1}$-cleanup tree $\tset'(v)$ for $\fset(v)$ and $G$, using Theorem~\ref{thm: cleanup tree}. From the definition of the cleanup tree, it is easy to verify that for every vertex $v'\in V(\tset'(v))$, $\fset(v')$ is aligned with $G$, $|\fset(v')|\leq L_h$, and $S(\fset(v'))\subseteq S(\fset_{h-1})$,  so $(\fset_1,G_1,\ldots,\fset_{h-1},G_{h-1},\fset(v'))\in \tcset_h(\ff {h-1})$ and $\fset(v')\in \cset$. We add tree $\tset'(v)$ to $\tset(\ff {h-1})$, by identifying its root with vertex $v$. Once we add a clean-up tree $\tset'(v)$ to each vertex $v\in \lset'$ with $\fset(v)\neq \emptyset$, we obtain the final tree $\tset(\ff {h-1})$.

For each vertex $v\in \lset'$, let $\lset(v)$ and $\iset(v)$ denote the sets of all leaf and inner vertices of the tree $\tset'(v)$, respectively. From the definition of the cleanup trees, for every leaf $v'\in \lset(v)$ with $S(\fset(v'))\neq\emptyset$, $\fset(v)$ is aligned with $G$, and $|\fset(v)|\leq L_{h-1}$. 
We are  also guaranteed that:

\[|\opt_{\fset(v')}|\geq \frac{|\opt_{\fset(v)}|}{L_h^{\delta+2}}\geq \frac{|\opt_{\fset_{h-1}|}}{\rho_{h-1}^{1/160}L_h^{\delta+3}}\geq \frac{|\opt_{\fset_{h-1}}|}{\rho_{h-1}^{1/80}},\]

from Equation~(\ref{eq: rho bound}). Therefore, we obtain a valid level-$(h-1)$ tree $\tset(\ff {h-1})$  overall.
 Let $\lset$ be the set of all leaf vertices of $\tset(\ff {h-1})$. Then the total loss of the tree $\tset(\ff {h-1})$ is bounded by:

\[\begin{split} 
\Lambda(\tset(\ff {h-1}))=|\opt_{\fset_{h-1}}|-\sum_{v\in \lset}|\opt_{\fset(v)}|&=\sum_{v\in \iset'}\lambda(v)+\sum_{v\in \lset'}\left(|\opt_{\fset(v)}|-\sum_{v'\in \lset(v)}|\opt_{\fset(v')}|\right )\\
&\leq \frac{40\tc\log \eta}{L_h}\cdot |\opt_{\fset_{h-1}}|+\sum_{v\in \lset'}\frac{12\tilde c}{L_{h-1}}\cdot |\opt_{\fset(v)}|\\
&\leq  \frac{40\tc\log \eta}{L_h}\cdot |\opt_{\fset_{h-1}}|+\frac{12\tilde c}{L_{h-1}}\cdot |\opt_{\fset_{h-1}}|\\
&\leq \frac{22\tc\log \eta}{L_{h-1}}\cdot |\opt_{\fset_{h-1}}|=\Lambda_{h-1} |\opt_{\fset_{h-1}}|.\end{split}\]

(We have used the fact that for all $v,v'\in \lset'$, neither vertex is a descendant of the other, so $\sum_{v\in \lset'}|\opt_{\fset(v)}|\leq |\opt_{\fset_{h-1}}|$ from Observation~\ref{obs: partitioning tree - disjointness of non-descendants}. We also used the fact that $L_h=2L_{h-1}$.)
%------------------------------------------------------------------------------------------------------------------------
%------------------------------------------------------------------------------------------------------------------------
%------------------------------------------------------------------------------------------------------------------------
%------------------------------------------------------------------------------------------------------------------------
%------------------------------------------------------------------------------------------------------------------------
%------------------------------------------------------------------------------------------------------------------------
%------------------------------------------------------------------------------------------------------------------------
\subsubsection{Induction Step.} 
We now fix some $1\leq i<h-1$, and we assume that the theorem holds for all $i'>i$. Consider some level-$i$ set $\ff i=(\fset_1,G_1,\ldots, G_{i-1},\fset_{i})$, so that, if $i>1$, then $\aset(\fset_i)\geq \aset(\fset_{i-1})/\rho_{i-1}^{1/10}$. Let $G_i$ be the $\rho_i$-accurate grid that we have computed for $\ff i$, when defining $\tcset_{i+1}(\ff i)$. From our assumption, $\tcset_{i+1}(\ff i)\neq \emptyset$. %We consider two cases. %The first case happens when $|\opt_{\fset_i}|\geq \tau_{i+1}\cdot \rho_i^{1/10}L_{i+1}^{\delta+1}$, and otherwise the second case happens.
%
%\paragraph{Case 1: $|\opt_{\fset_i}|\geq \tau_i\cdot \rho_i^{1/10}L_{i+1}^{\delta+1}$.}
%
For convenience, we will denote the tree $\tset(\ff i)$ by $\tset$. 
%Assume first that $|\opt_{\fset_i}|\leq \tau^*$. Then we can let $\tset$ contain a single vertex $v$ with $\fset(v)=\fset_i$. Since $\fset_i\in \cset'$ holds, we obtain a valid level-$i$ tree, with .
%We distinguish between two cases. The first case happens when $|\opt_{\fset_i}|\geq \rho_i^{323/160}\cdot 64L_{i+1}^{\delta}$.
%We start with the first case.

The algorithm for constructing the partitioning tree again consists of two stages. We start with the tree $\tset$ containing a single vertex $v$ with $\fset(v)=\fset_i$.  The first stage continues as long as 
there is some leaf $v$ in the tree $\tset$ with $\fset(v)\not\in \cset'$, and $|\opt_{\fset(v)}|\geq \frac{|\opt_{\fset_i}|}{\rho_i^{1/160}}$. Throughout this stage, we ensure the following invariants. Consider any leaf vertex $v$ of $\tset$, and denote $\fset(v)$ by $\fset$. Then either $\fset\in \cset'$, or:

\begin{properties}{I}
\item $\fset$ is aligned with $G_i$,  $S(\fset)\subseteq S(\fset_i)$, and $|\fset|\leq L_{i+1}$; and\label{prop 1}

\item $|\opt_{\fset}|\geq \frac{|\opt_{\fset_i}|}{\rho_i^{3/160}}$. \label{prop2}
\end{properties}

%For convenience, we also define the following invariant, that is a weaker version of Invariant~(\ref{prop last strong}):

%\begin{properties}[4]{I}
%\item $|\opt_{\fset}|\geq \frac{|\opt_{\fset_i}|}{\rho_i^{1/20}}$. \label{prop last weak}
%\end{properties}

%Notice that whenever Property~(\ref{prop last strong}) holds for set $\fset$, property~(\ref{prop last weak}) holds as well.
%Notice that since $\rho_i\geq \rho_h\geq \eta^{160}$, from the definition of $\eta$, if the above invariants hold, then $|\opt_{\fset(v)}|\geq \frac{|\opt_{\fset_i}|}{\rho_i^{1/40}}$. 
We need the following simple observation.

\begin{observation}\label{obs: from invariants}
Suppose we are given any set $\fset$ of fake rectangles, for which Properties~(\ref{prop 1}) and~(\ref{prop2}) hold.  Then:

\begin{enumerate}
\item If we denote $\ff{i+1}=(\fset_1,G_1,\ldots,G_{i-1},\fset_i,G_i,\fset)$, then $\ff{i+1}\in \tcset_{i+1}(\ff i)$, and in particular $\fset\in \cset$ - this is immediate from the invariants and the definition of the set $\tcset_{i+1}(\ff i)$;

\item $\aset(\fset)\geq \aset(\fset_i)/\rho_i^{1/10}$ - this follows from Observation~\ref{obs: math} and Invariant~(\ref{prop2}); 

%\item if $\fset\not\in \cset'$, then $|\opt_{\fset(v)}|\geq \rho_{i+1}^2$ - this follows since $\rho_i=\rho_{i+1}^2$, and $|\opt_{\fset(v)}|\geq \frac{|\opt_{\fset_i}|}{\rho_i^{1/20}}\geq \frac{\rho_i^2}{\rho_i^{1/20}}$;

\item If we denote $\rho'=\rho_i^{157/160}$, then $G_i$ is $\rho'$-accurate for $\fset$, and $\rho'\geq 32L_{i+1}^{\delta+2}$ - this follows from Observation~\ref{obs: rho accurate for big subinstances} together with Invariant~(\ref{prop2}), and Equation~(\ref{eq: rho bound}); and

\item if $\fset\not\in \cset'$, then $|\opt_{\fset(v)}|\geq 512L_{i+1}^{\delta+2}$ - this follows from Equation~(\ref{eq: opt bound}) and the definition of $\cset'$.
\end{enumerate}
\end{observation}

 We maintain a set $U\subseteq V(\tset)$ of vertices, that will be used for the analysis. These are all vertices that serve as the leaves of the tree $\tset$ at any time during its construction. At the beginning, we let $\tset$ contain a single root vertex $v_r$ with $\fset(v_r)=\fset_i$, and we let $U=\set{v}$. Notice that all invariants hold for $v_r$.

The first stage is executed as follows. While there is some leaf vertex $v$ in the tree $\tset$, with  $\fset(v)\not\in \cset'$, and $|\opt_{\fset(v)}|\geq \frac{|\opt_{\fset_i}|}{\rho_i^{1/160}}$, let $v$ be any such vertex.
From the first two statements of Observation~\ref{obs: from invariants}, $\ff {i+1}=(\fset_1,G_1,\ldots,G_{i-1},\fset_i,G_i,\fset(v))$ is a valid input to Theorem~\ref{thm: analysis of main algorithm}, and so we can compute a level-$(i+1)$ partitioning tree $\tset(\ff {i+1})$, that we denote by  $\tset(v)$. Let $\lset(v)$ be the set of the leaves of this tree. We add  $\tset(v)$ to $\tset$, by identifying its root with the vertex $v$.  We also add all vertices of $\lset(v)$ to set $U$. We now verify that all invariants hold for every vertex $v'\in \lset(v)$.  Let $G_{i+1}$ be the $\rho_{i+1}$-accurate grid that we have computed for $\ff {i+1}$, when constructing $\tcset_{i+2}(\ff {i+1})$. 
From the definition of the level-$(i+1)$ tree, either (i) $\fset(v')\in \cset'$, or (ii) $\fset(v')$ is aligned with $G_{i+1}$ (and hence with $G_i$, as $G_{i+1}$ is aligned with $G_i$), $|\fset(v')|\leq L_{i+1}$, and $|\opt_{\fset(v)}|/\rho_{i+1}^{1/40}\leq |\opt_{\fset(v')}|\leq |\opt_{\fset(v)}|/\rho_{i+1}^{1/160}$. Since $\rho_i=\rho_{i+1}^2$, we get that $ |\opt_{\fset(v')}|\geq |\opt_{\fset(v)}|/\rho_i^{1/80}\geq |\opt_{\fset_i}|/\rho_i^{3/160}$.
We are also guaranteed that $S(\fset(v'))\subseteq S(\fset(v))\subseteq S(\fset_i)$ from the definition of the partitioning tree. Therefore, Invariants~(\ref{prop 1}) and (\ref{prop2}) hold for $\fset(v')$.

%for every leaf vertex $v'\in \lset(v)$, either (i) $\fset(v')\in \cset'$, or (ii) $\fset(v')$ is aligned with $G_{i+1}$ (and hence with $G_i$, as $G_{i+1}$ is aligned with $G_i$, $|\fset(v')|\leq L_i$), and $\tau_i\leq |\opt_{\fset(v')}|\leq |\opt_{\fset(v)}|/\rho_i^{1/10}$. Therefore, all our invariants continue to hold. We add to $U$ all leaf vertices of $\tset(\ff {i+1})$.

The first stage terminates when for every leaf $v$ of $\tset$, either $\fset(v)\in \cset'$, or $|\opt_{\fset(v)}|< \frac{|\opt_{\fset_i}|}{\rho_i^{1/160}}$. 
We partition the vertices of $U$ into classes, where class $U_j$ contains all vertices $v\in U$, such that the unique path in $\tset$ connecting $v$ to the root $v_r$ of $\tset$ contains exactly $j$ vertices of $U$. Since for every non-leaf vertex $v\in U$, for every vertex $v'\in \lset(v)$, $|\opt_{\fset(v')}|\leq |\opt_{\fset(v)}|/\rho_{i+1}^{1/160}$, and $\rho_i=\rho_{i+1}^2$, it is easy to see that the total number of non-empty sets $U_j$ is at most $3$, and only $U_1,U_2,U_3\neq \emptyset$, while $U_3$ only contains the leaves of the current tree, and $U_1$ contains a single vertex - the root of the tree.

For every pair $v,v'\in U_2$ of vertices, neither vertex is a descendant of the other, and so $\sum_{v\in U_2}|\opt_{\fset(v)}|\leq |\opt_{\fset_i}|$ from Observation~\ref{obs: partitioning tree - disjointness of non-descendants}. 
For every vertex $v\in U_1\cup U_2$, we define the modified loss of $v$ to be: $\tilde{\lambda}(v)=|\opt_{\fset(v)}|-\sum_{v'\in \lset(v)}|\opt_{\fset(v')}|$. Then for every vertex $v\in U_1\cup U_2$, $\tilde{\lambda}(v)\leq \Lambda_{i+1}|\opt_{\fset(v)}|$ from the induction hypothesis, and overall:

\[\sum_{j=1}^2\sum_{v\in U_j}\tilde{\lambda}(v)\leq \sum_{j=1}^2\sum_{v\in U_j}\Lambda_{i+1}|\opt_{\fset(v)}|\leq \sum_{j=1}^2  \Lambda_{i+1}|\opt_{\fset_i}|\leq 2\Lambda_{i+1}|\opt_{\fset_i}|.\]

For the second stage, consider any leaf vertex $v$ of $\tset$, with $\fset(v)\not\in \cset'$. Then $|\opt_{\fset(v)}|\geq |\opt_{\fset_i}|/\rho_i^{3/160}$. From Observation~\ref{obs: from invariants},  if we denote $\rho'=\rho_i^{157/160}$, then $G_i$ is $\rho'$-accurate for $\fset(v)$, $\rho'\geq 32L_{i+1}^{\delta+2}$, and $|\opt_{\fset(v)}|\geq 512 L_{i+1}^{\delta+2}$. Therefore, we can use Theorem~\ref{thm: cleanup tree}, to construct an $L_{i+1}$--$L_{i}$ cleanup tree $\tset(v)$, such that for every leaf $v'$ of tree $\tset(v)$, if $S(\fset(v'))\neq\emptyset$, then $\fset(v')$ is aligned with $G_i$, $|\fset(v')|\leq L_i$, and $|\opt_{\fset(v')}|\geq \frac{|\opt_{\fset(v)}|}{L_{i+1}^{\delta+2}}\geq \frac{|\opt_{\fset_i}|}{\rho_i^{3/160}L_{i+1}^{\delta+2}}\geq \frac{|\opt_{\fset_i}|}{\rho_i^{1/40}}$ from Equation~(\ref{eq: rho bound}). Therefore, we obtain a valid level-$i$ partitioning tree for $\ff i$. 
Let $\lset'$ be the set of all vertices $v$ that served as the leaves of the tree $\tset$ at the end of the first stage, with $\fset(v)\not\in \cset'$. For each such vertex $v$, let $\lset(v)$ be the set of the leaves in the cleanup tree $\tset(v)$. As before, we define the modified loss of vertex $v$ to be $\tilde{\lambda}(v)=|\opt_{\fset(v)}|-\sum_{v'\in \lset(v)}|\opt_{\fset(v')}|$. From Theorem~\ref{thm: cleanup tree}, $\tilde{\lambda}(v)\leq |\opt_{\fset(v)}|\cdot\frac{12\tc}{L_i}$. For every pair $v,v'$ of vertices in $\lset'$, neither vertex is a descendant of the other, and so $\sum_{v\in \lset'}|\opt_{\fset(v)}|\leq |\opt_{\fset_i}|$ from Observation~\ref{obs: partitioning tree - disjointness of non-descendants}. Therefore,

\[\sum_{v\in \lset'}\tilde{\lambda}(v)\leq \sum_{v\in \lset'}|\opt_{\fset(v)}|\cdot\frac{12\tc}{L_i}\leq \frac{12\tc}{L_i}|\opt_{\fset_i}|.\]

Overall, the total loss of tree $\tset$ is:

\[\Lambda(\tset)=|\opt_{\fset_i}|-\sum_{v\in \lset(\tset)}|\opt_{\fset(v)}|\leq \sum_{j=1}^2\sum_{v\in U_j}\tilde{\lambda}(v)+\sum_{v\in \lset'}\tilde{\lambda}(v)\leq \left (2\Lambda_{i+1}+\frac{12\tc}{L_i}\right )|\opt_{\fset_i}|=\Lambda_i|\opt_{\fset_i}|.\]

This completes the proof of Theorem~\ref{thm: analysis of main algorithm}.

We are now ready to construct our final partitioning tree $\tset$. We start with $\tset$ containing a single vertex $v_r$, with $\fset(v_r)=\emptyset$. Throughout the algorithm execution, we maintain the invariant that for every leaf vertex $v$ of the tree, $\fset(v)\in  \cset_1$.  The algorithm is executed as long as there is any leaf vertex $v\in \tset$ with $\fset(v)\not\in \cset'$. Given any such vertex $v$, we let $\ff 1=(\fset(v))$, and we let $\tset(v)$ be the level-1 tree $\tset(\ff 1)$ given by Theorem~\ref{thm: analysis of main algorithm}. We add the tree $\tset(v)$ to $\tset$, by identifying its root with the vertex $v$, and we denote by $\lset(v)$ the set of leaves of $\tset(v)$. We then continue to the next iteration. It is immediate to verify that the invariant continues to hold. The algorithm terminates, when for every leaf $v$ of $\tset$, $\fset(v)\in \cset'$. As before, we let $U$ contain all vertices of $\tset$, that served as the leaves of $\tset$ at any point of the algorithm execution. We partition the set $U$ into subsets $U_1,U_2,\ldots$, where set $U_j$ contains all vertices $v$, such that the unique path from $v$ to the root $v_r$ of $\tset$ in $\tset$ contains exactly $j$ vertices of $U$. Recall that for every vertex $v\in U$ that is not a leaf of $\tset$, for every vertex $v'\in \lset(v)$, if $\fset(v')\not\in \cset'$, then $|\opt_{\fset(v')}|\leq |\opt_{\fset(v)}|/\rho_1^{1/160}$. Since $\rho_1=\sqrt N$, the number of non-empty sets $U_j$ is bounded by $320$. 
For every vertex $v\in U$, we again define the modified loss at $v$ to be $\tilde{\lambda}(v)=|\opt_{\fset(v)}|-\sum_{v'\in \lset(v)}|\opt_{\fset(v')}|$. From Theorem~\ref{thm: analysis of main algorithm}, $\tilde{\lambda}(v)\leq \Lambda_1|\opt_{\fset(v)}|$ for all $v\in U$. As before, for all $1\leq j\leq 320$, no vertex of $U_j$ is a descendant of another, and so $\sum_{v\in U_j}|\opt_{\fset(v)}|\leq |\opt|$ from Observation~\ref{obs: partitioning tree - disjointness of non-descendants}. We can now bound the total loss of the tree as:

\[\Lambda(\tset)=|\opt|-\sum_{v\in \lset(\tset)}|\opt_{\fset(v)}|=\sum_{j=1}^{320}\sum_{v\in U_j}\tilde{\lambda}(v)\leq \sum_{j=1}^{320}\sum_{v\in U_j}\Lambda_1|\opt_{\fset(v)}|\leq 320\Lambda_1|\opt|.\]

 Using the recursive definition  $\Lambda_{h-1}=\frac{22\tc\log \eta}{L_{h-1}}$, and  $\Lambda_i=\left(2\Lambda_{i+1}+\frac{12\tc}{L_i}\right )$ for $1\leq i<h-1$, it is easy to verify that:

\[\begin{split}
\Lambda_1&\leq \sum_{i=1}^{h-2}\frac{2^i\cdot 12 \tc}{L_i}+2^h\cdot \frac{22\tc\log \eta}{L_h}\\
&\leq \sum_{i=1}^{h-2}\frac{2^i\cdot 12\tc}{\tc(\log\log N)^3\cdot 2^i/\eps}+2^h\cdot\frac{O(\log\log N)}{\tc 2^h(\log\log N)^3/\eps}\\
&\leq \frac{\eps}{640},
\end{split}\]

assuming that $N$ is large enough.
Therefore, the total loss of the tree $\tset$ is bounded by $\eps\cdot |\opt|/2$. From Observation~\ref{obs: final analysis}, our algorithm computes a $(1-\eps)$-approximate solution overall. As discussed above, the running time of the algorithm is bounded by $n^{O((\log\log N/\eps)^4)}$.

\label{---------------------------------------------------The end------------------------------------------------}

\bibliographystyle{alpha}
\bibliography{misr.v4}

\newpage
\appendix
\bigskip \bigskip \noindent \Large \textbf{APPENDIX}

\normalsize

%-------------------------------------------------------
%-------------------------------------------------------
%-----------------------------------------------------------------------------------------------------
\section{Proofs Omitted from Section~\ref{sec: prelims}}
%-------------------------------------------------------
%-------------------------------------------------------
%-------------------------------------------------------
%-------------------------------------------------------

%-------------------------------------------------------
%-------------------------------------------------------
%-----------------------------------------------------------------------------------------------------
\subsection{Proof of Claim~\ref{claim: canonical instances}}
%-------------------------------------------------------
%-------------------------------------------------------
%-------------------------------------------------------
%-------------------------------------------------------

%\begin{proof}
Since we assume that the rectangles are open, we can obtain an equivalent non-degenerate instance $\rset''$, as follows. Intuitively, for each rectangle $R\in \rset$, we move its right boundary towards left by a very small random amount. Similarly, we move its left boundary towards right, top boundary down, and bottom boundary up, by very small random amounts. 

More formally, let $X$ be the set of all {\bf distinct} $x$-coordinates of the corners of the rectangles in $\rset$, and let $\Delta$ be the minimum value of $|x-x'|$ for any pair $x,x'\in X$ with $x\neq x'$. Each rectangle $R\in \rset$ chooses a random value $\Delta_R\in (0,\Delta/4)$. We then obtain a new rectangle $R'$ by increasing the $x$-coordinates of the two left corners of $R$ by $\Delta_R$ and reducing the $x$-coordinates of the two right corners of $R$ by $\Delta_R$. We define the set $Y$ of all distinct $y$-coordinates of the corners of the rectangles in $\rset$, and perform a similar transformation with the $y$-coordinates of the corners of the rectangles. Let $\rset''$ be the final set of the rectangles. Then with probability $1$, $\rset''$ is non-degenerate. Moreover, since the rectangles are open, it is easy to see that the transformation preserves rectangle intersections: that is, $R_i,R_j\in \rset$ intersect if and only if their corresponding new rectangles $R_i',R_j'\in \rset''$ intersect. Therefore, from now on we assume that our input instance is non-degenerate.

Given a non-degenerate instance $\rset''$, we can transform it into a combinatorially equivalent instance, where the coordinates of the rectangles' corners are integers between $1$ and $2n$. Indeed, let $X'$ be the set of all $x$-coordinates of the corners of the rectangles in $\rset''$, so $|X'|= 2n$. Assume that $X'=\set{a_1,a_2,\ldots,a_{2n}}$, where $a_1<a_2<\cdots<a_{2n}$. We define a mapping $f: X'\rightarrow \set{1,2,\ldots,2n}$, where $f(a_i)=i$. Let $Y'$ be the set of all distinct $y$-coordinates of the corners of the rectangles in $\rset''$. We define a mapping $g: Y'\rightarrow \set{1,2,\ldots,2n}$ similarly. The final set $\rset'$ of rectangles is defined as follows: $\rset'=\set{R_1',\ldots,R_n'}$, where for each $1\leq i\leq n$, the lower left corner of $R'_i$ is $(f(x_i^{(1)}), g(y_i^{(1)}))$, and its upper right corner is $(f(x_i^{(2)}), g(y_i^{(2)}))$. It is immediate to verify that for all $1\leq i\neq j\leq n$, $R'_i$ and $R'_j$ intersect if and only if $R_i$ and $R_j$ intersect. 

Therefore, for any set $\tilde R\subseteq \rset$ of disjoint rectangles, the corresponding set $\set{R'\mid R\in \rset}$ of rectangles in the new instance $\rset'$ is also disjoint and vice versa.
%\end{proof}

%-------------------------------------------------------
%-------------------------------------------------------
%-----------------------------------------------------------------------------------------------------
\subsection{Proof of Theorem~\ref{thm: from n to opt}}
%-------------------------------------------------------
%-------------------------------------------------------
%-------------------------------------------------------
%-------------------------------------------------------
We assume without loss of generality that instance $\rset$ is non-degenerate. We will construct an $\left(O(w^*)\times O(w^*)\right )$-grid $G$, and then round the boundaries of all rectangles in $\rset$ to the grid $G$. (We note that the value $w^*$ is not known to the algorithm).

We start by constructing a set $\vset$ of vertical lines of the grid, that have the following property: for every rectangle $R\in \rset$, at least one vertical line $V\in \vset$ intersects $R$. In order to construct $\vset$, let $\iset$ be the set of intervals, obtained by projecting all rectangles $R\in \rset$ onto the $X$-axis. Notice that the intervals in $\iset$ are open. Let $\iset^*\subseteq\iset$ be a maximum independent set of the intervals\footnote{Set $\iset^*\subseteq \iset$ of intervals is independent if and only if no pair of intervals in $\iset^*$ intersect.} in $\iset$ with the following additional property: if $I\in \iset\setminus\iset^*$, then no interval $I'\in \iset^*$ strictly contains $I$. In order to construct $\iset^*$, we start with any maximum independent set of $\iset$ (that can be computed efficiently via standard dynamic programming techniques), and then iterate. While there are intervals $I\in \iset\setminus\iset^*$, $I'\in \iset^*$ with $I\subsetneq I'$, we replace $I'$ with $I$ and continue. It is easy to see that after $O(n)$ iterations we obtain the desired set $\iset^*$. Notice that $|\iset^*|\leq w^*$, as the set of all rectangles whose intervals belong to $\iset^*$ must form an independent set. Let $X$ be the set of points, constructed as follows. For every interval $I\in \iset^*$, we add to $X$ the coordinates of the left endpoint of $I$, the right endpoint of $I$, and one arbitrary additional inner point on $I$. Observe that for every interval $I'\in \iset$, there is some point $x\in X$ with $x\in I'$.

The final set $\vset$ of vertical lines contains one vertical line $V_x$ for each coordinate $x\in X$, and also the left and right boundary of the bounding box (that is, the lines $x = 0$ and $x = 2n + 1$). Then $|\vset|\leq 3|\iset^*| + 2\leq 5w^*$, and for every rectangle $R\in \rset$, at least one line $V\in \vset$ intersects $R$. Similarly, we build a set $\hset$ of at most $5w^*$ horizontal lines, such that for each rectangle $R\in \rset$, at least one line in $\hset$ intersects $R$. Finally, we construct a new instance $\rset'=\set{R'\mid R\in \rset}$, as follows. Consider some rectangle $R\in \rset$. If its right boundary does not lie on any line $V\in \vset$, then we move the right boundary of $R$ to the right, until it lies on some such line. Similarly, we move its left boundary to the left, top boundary up and bottom boundary down, until all four edges lie on the lines of the grid $G$. This defines the rectangle $R'$, that is added to $\rset'$. Clearly, $R\subseteq R'$, and so any solution $\sset'\subseteq\rset'$ to the new instance immediately defines a solution $\sset\subseteq \rset$ of the same value to the original instance. It is easy to see that the number of distinct rectangles in $\rset'$ is at most $O\left ((w^*)^4\right)$, since there are at most $\binom{O(w^*)}{2}$ possible choices for the $x$-coordinates of the left and the right boundaries of each rectangle, and at most $\binom{O(w^*)}{2}$  possible choices for the $y$-coordinates of its top and bottom boundaries. Let $\opt'$ denote the  optimal solution to the resulting instance $\rset'$ of \MISR. The following lemma will then finish the proof of the theorem.

\begin{lemma}
 $|\opt'|\geq \Omega(w^*)$.
 \end{lemma}
We define the following five instances of the problem. Let $\rset_0=\rset$. Let $\rset_1$ be the instance obtained from $\rset$ after we round, for each rectangle $R\in \rset_0$, its right boundary only, by moving it to the right until it lies on some vertical line of the grid. Similarly, $\rset_2$ is obtained from $\rset_1$ by rounding the left boundaries of rectangles in $\rset_1$, $\rset_3$ is obtained from $\rset_2$ by rounding the top boundaries of the rectangles in $\rset_2$, and $\rset_4$ is obtained from $\rset_3$ by rounding the bottom boundaries of the rectangles in $\rset_3$. Notice that for each $0\leq i\leq 4$, the grid $G$ still has the property that for each rectangle $R\in \rset_i$, at least one vertical line of the grid intersects $R$, and at least one horizontal line of the grid intersects $R$. Let $\sset_0$ be any optimal solution for instance $\rset_0$, so $|\sset_0|=w^*$. We construct, for $i=1,2,3,4$, a feasible solution $\sset_i$ for $\rset_i$, such that $|\sset_i|\geq \Omega(|\sset_{i-1}|)$. It will then follow that $|\opt'|\geq |\sset_4|\geq \Omega(w^*)$.

For simplicity, we show how to obtain $\sset_1$ from $\sset_0$; the other three cases are analyzed similarly. Let $\sset\subseteq\rset_1$ be the set of rectangles, corresponding to the rectangles in $\sset_0$, that is, $\sset=\set{R'\mid R\in \sset_0}$, where rectangle $R'$ is obtained from $R$ by moving its right boundary to closest grid line to its right. We will find an independent set $\sset_1\subseteq\sset$ of size $\Omega(|\sset_0|)$.

For each rectangle $R_1\in \sset_0$, each of the two left corners of $R_1$ shoots a straight line to the left, until it hits some other rectangle $R_2\in \sset_0$ or its boundary. We say that $R_1$ \emph{tags} $R_2$ in this case. Similarly, each of the two right corners of $R_1$ shoots a straight line to the right, until it hits some other rectangle $R_3\in \sset_0$ or its boundary. We again say that $R_1$ tags $R_3$. Note that $R_1$ may tag at most four rectangles.
Following is the central claim in our analysis.

\begin{claim}\label{claim: aux}
Let $R_1',R_2'\in \sset$ be any two rectangles, and assume that they intersect. Then either $R_1$ tagged $R_2$, or $R_2$ tagged $R_1$.
\end{claim}

\begin{proof}
Since the original rectangles $R_1,R_2$ do not intersect, but the new rectangles $R_1',R_2'$ intersect, the projections of $R_1,R_2$ onto the $y$-axis must intersect, and their projections onto the $x$-axis cannot intersect. Therefore, we can assume without loss of generality that the $x$-coordinate of the right boundary of $R_1$ is smaller than or equal to the $x$-coordinate of the left boundary of $R_2$ (that is, $R_1$ lies to the left of $R_2$). Let $x_1$ be the $x$-coordinate of the right boundary of $R_1$, and let $x_2$ be the $x$-coordinate of the right boundary of $R_1'$. Let $x'$ be the $x$-coordinate of the left boundary of $R_2$ (see Figure~\ref{fig: claim-illustration}).  Then $x_1\leq x'<x_2$, and $x_2$ is the smallest $x$-coordinate to the right of $x_1$ through which a vertical line of the grid passes.

\begin{figure}[h]
\scalebox{0.5}{\includegraphics{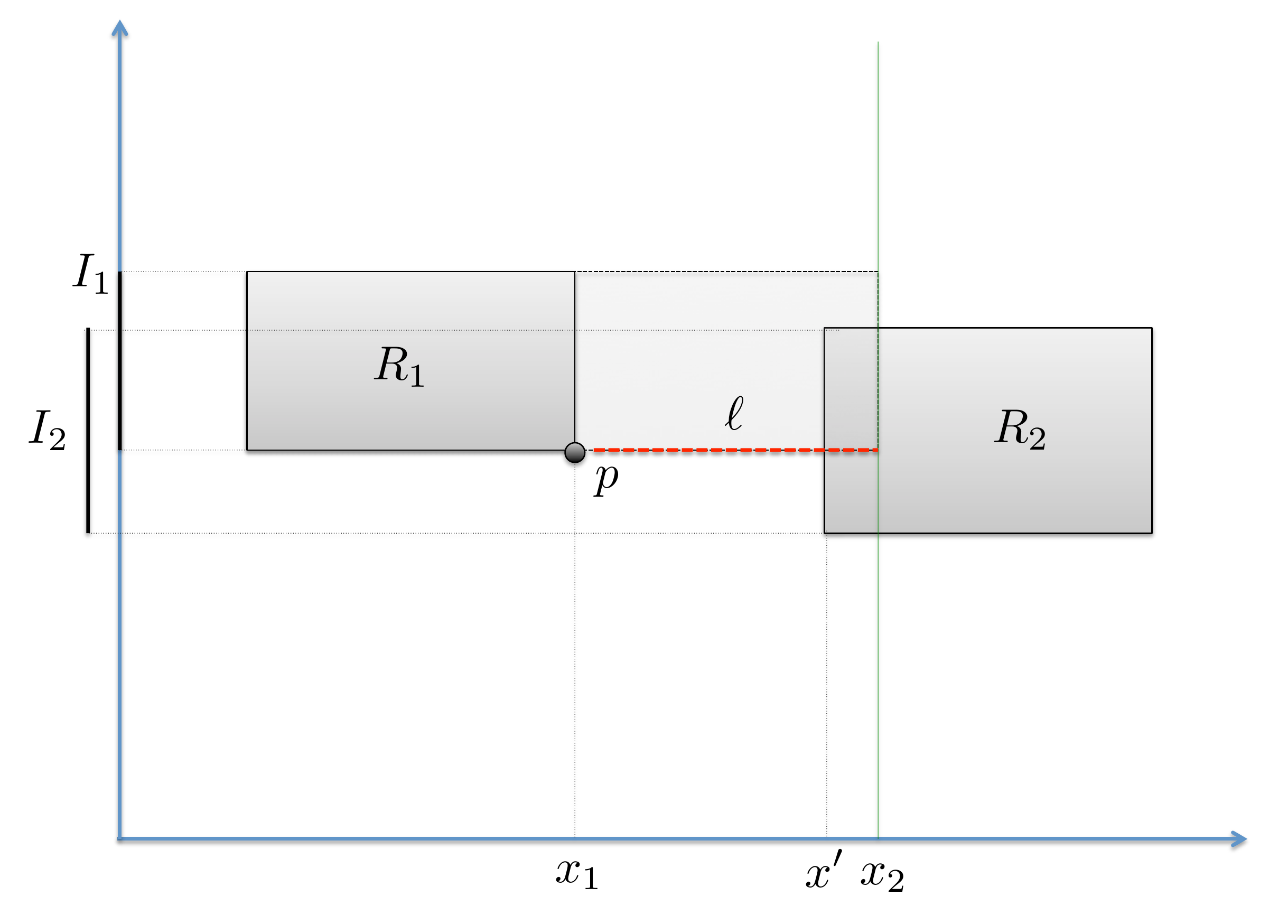}}\caption{Illustration to the proof of Lemma~\ref{claim: aux}\label{fig: claim-illustration}}
\end{figure}

 Let $I_1,I_2$ be the projections of $R_1',R_2'$ onto the $y$-axis, respectively, that we consider to be closed intervals. Then,  either one of the endpoints of $I_1$ is contained in $I_2$, or one of the endpoints of $I_2$ is contained in $I_1$. Assume without loss of generality that it is the former, and let $p$ be the right corner of $R_1$, whose corresponding endpoint of $I_1$ is contained in $I_2$. We claim that the line $\ell$ that $p$ shot to the right must have tagged $R_2$. Assume otherwise. Let $\ell'$ be the straight horizontal line connecting $p$ to some point $p'$ on the boundary of $R_2$. If $R_1$ did not tag $R_2$, then there is some other rectangle $R\in \sset_0$ that intersects $\ell'$, between $p$ and $p'$. But then there is some vertical line $V$ of the grid intersecting $R$. Then the right boundary of $R_1'$ should have been rounded to $V$, and so $R_1'$ cannot intersect $R_2'$. The case where one of the endpoints of $I_2$ is contained in $I_1$ is analyzed similarly.
 \end{proof}
 
  We now build a graph $H$, whose vertex set is $\set{v_R\mid R\in \sset_0}$, and there is an edge $(v_{R_1},v_{R_2})$ if and only if one of $R_1,R_2$ tags the other. Observe that if we find an independent set $\iset$ in $H$, then the rectangles corresponding to $\iset$ define an independent set in $\rset_1$. Therefore, it is enough to prove that there is an independent set in $H$ of size $\Omega(|\sset_0|)$. We do so using standard techniques. We show that for every subset $U\subseteq V(H)$ of vertices of $H$, at least one vertex of $U$ has a constant degree in $H[U]$. Indeed, every rectangle may tag at most $4$ other rectangles, and so the number of edges in $H[U]$ is at most $4|U|$. Therefore, at least one vertex of $U$ has degree at most $8$. In order to build the independent set $\iset$ of $H$, we start with any vertex $v\in V(H)$, whose degree is at most $8$. We add $v$ to $\iset$, and delete $v$ and all its neighbors from $H$. We then continue to the next iteration. From the above discussion, in every iteration, we can find a vertex of degree at most $8$ in the remaining graph, and it is easy to see that throughout the algorithm $\iset$ is an independent set. In each iteration, we add one vertex to $\iset$ and delete at most $9$ vertices from $H$. Therefore, in the end, $|\iset|\geq |\sset_0|/9$, and we obtain an independent set $\sset_1\subseteq \rset_1$, whose corresponding vertices belong to $\iset$, of size at least $|\sset_0|/9$.

%-------------------------------------------------------
%-------------------------------------------------------
%-----------------------------------------------------------------------------------------------------
\subsection{Proof of Lemma~\ref{lemma: simple tiling}}
%-------------------------------------------------------
%-------------------------------------------------------
%-------------------------------------------------------
%-------------------------------------------------------

%\begin{proof}
The proof is by induction on the number of corners $L$ on the boundary of $P$. The base case is when $L=4$, and $P$ is a rectangle. In this case $\fset$ contains a single rectangle $P$. We now assume the correctness of the claim for polygons with up to $L-1$ corners on their boundary, for $L\geq 5$, and prove it for $L$.

Let $P$ be any polygon with $L$ corners on its boundary. Then $P$ is not a rectangle, and so there is at least one corner $p$ on the boundary of $P$, such that the two edges $e,e'$ of the boundary of $P$ adjacent to $p$ form a $270$-degree internal angle. We assume without loss of generality that $e$ is a vertical edge, $e'$ is a horizontal edge, and that $p$ is the bottom endpoint of $e$ (see Figure~\ref{fig: padding}). We draw a line $\ell$ from $p$ down, until it reaches any point $p'$ on the boundary of $P$. Line $\ell$ splits $P$ into two simple closed polygons, that we denote by $P_1$ and $P_2$. Let $L_1$ and $L_2$ denote the number of the corners on the boundaries of $P_1$ and $P_2$, respectively. We claim that $L_1+L_2\leq L+2$. Indeed, point $p$ served as a corner of $P$, and it continues to serve as a corner of exactly one of the two polygons $P_1,P_2$. Point $p'$ may now serve as a corner of both polygons. No other point serves as a corner in both polygons, and no other point, that did not serve as a corner of $P$, may become a corner of $P_1$ or $P_2$. Therefore, $L_1+L_2\leq L+2$. Since each of $P_1$ and $P_2$ must have at least four corners, $L_1,L_2<L$ must hold. From the induction hypothesis, there is a set $\fset_1$ of at most $L_1-3$ closed internally-disjoint axis-parallel rectangles whose union is $P_1$, and similarly there is such a set $\fset_2$ of cardinality at most $L_2-3$ for $P_2$. Setting $\fset=\fset_1\cup \fset_2$, we obtain a set of at most $L_1-3+L_2-3\leq L-3$ rectangles, whose union is $P$.

\begin{figure}[h]
\scalebox{0.4}{\includegraphics{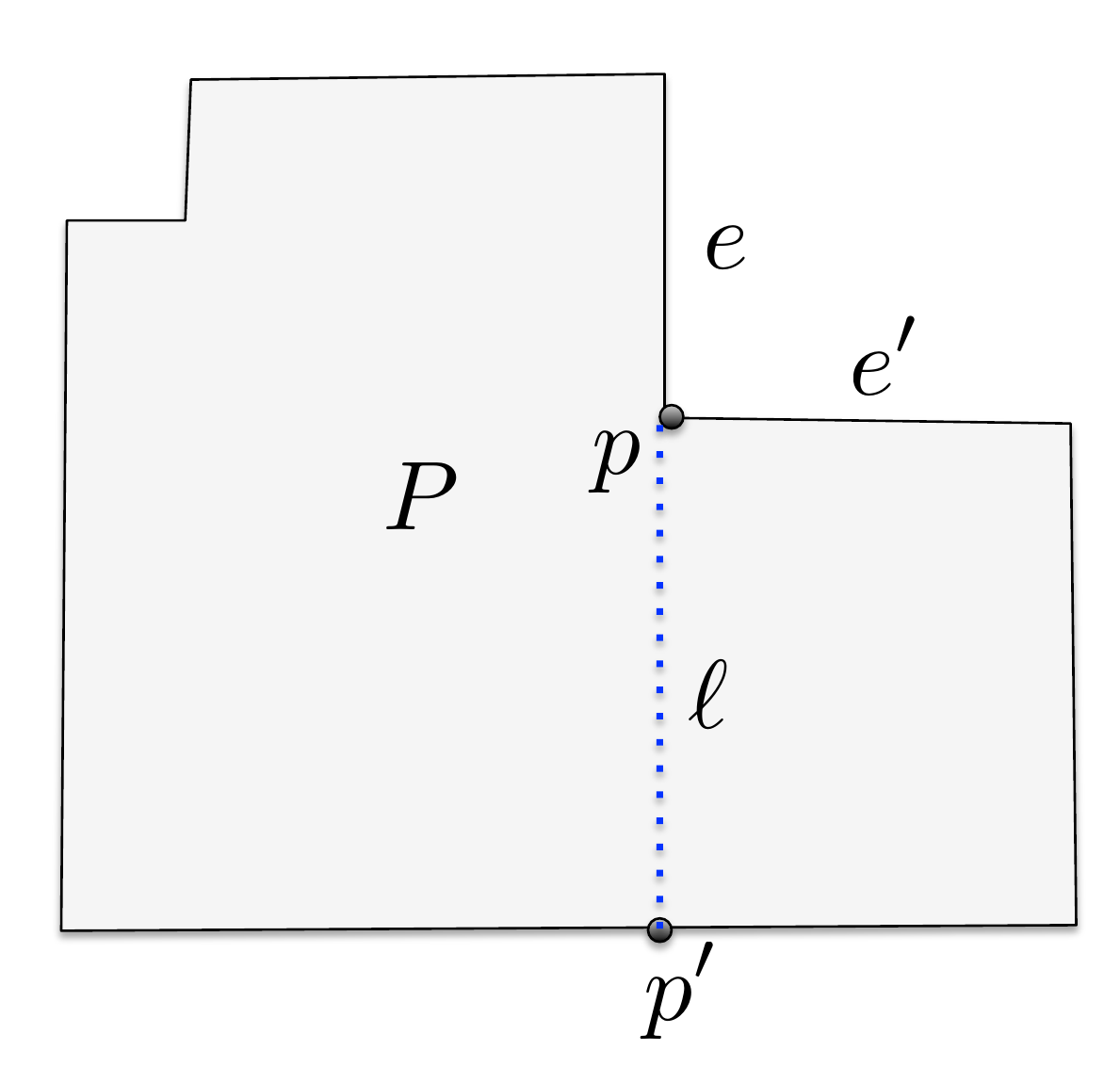}}\caption{Partitioning $P$ into smaller polygons \label{fig: padding}}
\end{figure}

It is easy to see that both $P_1$ and $P_2$ are aligned with $Z$, since $p\in Z$, and, if we denote $p'=(x,y)$, then $x$ is the $x$-coordinate of $p$, and $y$ is the $y$-coordinate of one of the corners of $P$, as either $p'$ is itself a corner of $P$, or it belongs to a horizontal edge on the boundary of $P$. From the induction hypothesis, every rectangle in $\fset_1$ and $\fset_2$ is aligned with $Z$.
%\end{proof}

%-------------------------------------------------------
%-------------------------------------------------------
%-----------------------------------------------------------------------------------------------------
\subsection{Proof of Lemma~\ref{lemma: non-simple-tiling}}
%-------------------------------------------------------
%-------------------------------------------------------
%-------------------------------------------------------
%-------------------------------------------------------

%\begin{proof}
If the boundaries of $B$ and $P$ are not disjoint, then it is easy to see that $B\setminus P$ is a collection of disjoint simple polygons. Let $\pset$ denote this collection of polygons. Then the total number of corners on the boundaries of the polygons in $\pset$ is at most $L+4$: the $L$ corners of $P$ and $4$ additional corners of $B$. We then use Lemma~\ref{lemma: simple tiling} to tile each of the polygons in $\pset$ separately, thus obtaining a collection $\fset$ of at most $L+1$   closed internally-disjoint axis-parallel rectangles, such that $\bigcup_{F\in \fset}F=B\setminus P$. From Lemma~\ref{lemma: simple tiling}, every rectangle in $\fset$ is aligned with $Z$.

Assume now that the boundaries of $B$ and $P$ are disjoint. Among all vertical edges on the boundary of $P$, let $e$ be the left-most one (breaking ties arbitrarily). Let $R$ be the rectangle whose right boundary is the edge $e$, and the left boundary lies on the left boundary of $B$. Notice that except for the right boundary of $R$, rectangle $R$ is completely disjoint from $P$. Moreover, $R$ is aligned with $Z\cup P$, and $B\setminus (P\cup R)$ is a simple polygon with at most $L+4$ corners. From Lemma~\ref{lemma: simple tiling}, there is a set $\fset'$ of at most $L+1$ closed internally-disjoint axis-parallel rectangles, such that $\bigcup_{F\in \fset'}F=B\setminus (P\cup R)$, and each rectangle of $\fset'$ is aligned with $Z$. Setting $\fset=\fset'\cup \set{R}$ gives the desired set of rectangles.
%\end{proof}

%-------------------------------------------------------
%-------------------------------------------------------
%-----------------------------------------------------------------------------------------------------
\section{Proofs Omitted from Section~\ref{sec: partitions}}
%-------------------------------------------------------
%-------------------------------------------------------
%-------------------------------------------------------
%-------------------------------------------------------

%-------------------------------------------------------
%-------------------------------------------------------
%-----------------------------------------------------------------------------------------------------
\subsection{Proof of Theorem~\ref{thm: r-good partition}}
%-------------------------------------------------------
%-------------------------------------------------------
%-------------------------------------------------------
%-------------------------------------------------------

%-------------------------------------------------------------------------
%-------------------------------------------------------------------------
%We now turn to proving Theorem~\ref{thm: r-good partition}. 
%For convenience, we denote $w(\opt')$ by $W$. Recall that $\lset,\hset$ are the sets of the $r$-light and $r$-heavy rectangles of $\opt'$, respectively. %We prove below that there exists a required partition of $B$ into rectangular cells, containing at most $22r$ cells, where every cell intersects at most $O(N'/r)$ rectangles of $\opt'$, assuming that $m\leq r$. Using the same argument with parameter $r/22$ instead of $r$ gives the required partition.
%Instead of finding an $r$-good partition $\pset$ of $S$, we will find a partition $\pset^*$ of  the region contained in the bounding box $B$ into $O(r)$ rectangular cells, such that each cell intersects at most $O(N'/r)$ rectangles of $\opt'$, and each fake rectangle $F\in \fset$ serves as a cell of $\pset^*$. The final partition $\pset$ is then defined to be $\pset^*\setminus \fset$.

We find the partition $\pset$ in two steps. In the first step, we construct an initial partition $\pset'$ of $B$ into $O(r)$ rectangular cells. In the second step, we further subdivide some cells $P\in \pset'$ into smaller cells, thus obtaining the final partition $\pset$.
Our proof follows the arguments of~\cite{Har-Peled14}  very closely.

\paragraph{Step 1.} We start by constructing a set $\wset=\sset_1\cup \sset_2$ of rectangles as follows. Initially, $\sset_1=\fset$. Additionally, each rectangle $R\in \opt'$ is added to $\sset_1$ independently at random with probability $r/|\opt'|$. In order to construct the set $\sset_2$ of rectangles, we start with $\sset_2=\emptyset$, and add each rectangle $R\in \opt'$ to $\sset_2$ independently at random with probability $r/|\opt'|$. We then set $\wset=\sset_1\cup \sset_2$. (We note that we could have, equivalently, directly added each rectangle $R\in \opt'$ to $\wset$ with probability $2r/|\opt'|$, in addition to adding all rectangles of $\fset$ to $\wset$. However, as we will see later, this two-stage randomized procedure significantly simplifies the analysis). We say that the bad event $\event_1$ happens if and only if $|\wset|>18r$. Notice that the expected number of rectangles in $\wset$ is bounded by $2r+m\leq 3r$.
We use the following standard Chernoff bound:

\begin{theorem}\label{thm: Chernoff}(Theorem 1.1 in~\cite{measure-conc}.)
Let $X_1,\ldots,X_n$ be random variables independently distributed in $[0,1]$, and let $X=\sum_iX_i$. Then for any $t>2e\cdot \expect{X}$, $\prob{X>t}\leq 2^{-t}$.
\end{theorem}

From Theorem~\ref{thm: Chernoff}, the probability that $\event_1$ happens is bounded by $1/2^{18}$.

Suppose we are given any set $\xset\subseteq \opt'\cup \fset$ of rectangles, where $\fset\subseteq \xset$. We associate a partition $\pset(\xset)$ of $B$ into rectangular cells with $\xset$. The partition $\pset(\xset)$ is constructed as follows. For each rectangle $R \in \xset$, each of the two top corners of $R$ shoots a ray up, until it reaches the boundary of some other rectangle in $\xset$, or the bounding box $B$. Similarly, each of the two bottom corners of $R$ shoots a ray down, until it reaches the boundary of some other rectangle in $\xset$, or the bounding box $B$. Consider the partition $\pset(\xset)$ of $B$, defined by the boundaries of the rectangles in $\xset$, the bounding box $B$,  and the vertical lines that we have just constructed. 

\begin{observation}
Every cell of $\pset(\xset)$ is a rectangle.
\end{observation}
\begin{proof}
Consider some cell $C\in \pset(\xset)$, and assume for contradiction that it is not rectangular. It is easy to see that the boundary of $C$ is a simple cycle. Consider a tour of the boundary of $C$, traversing it in a clock-wise fashion. Since all lines in our partition are parallel to the axes, every turn of the tour has either $90$ or $270$ degrees. Moreover, at least one turn must be a $270$-degree turn if $C$ is not a rectangle. Consider some corner $p$ of the boundary of $C$, where the tour makes a $270$-degree turn. Let $e,e'$ be the edges of the boundary of $C$ incident on $p$. One of these edges must be a horizontal line, and one a vertical line. Assume w.l.o.g. that $e$ is the horizontal edge. Then $e$ must be contained in the top or the bottom boundary of some rectangle $R\in \wset$, and $p$ must be a corner of that rectangle (see Figure~\ref{fig: wrong turn}). But then there must be two vertical lines adjacent to $p$: one going up and one going down, making it impossible that both $e$ and $e'$ lie on the boundary of the same cell $C$. 
\begin{figure}[h!]
    \begin{center}
    \scalebox{0.95}{\includegraphics{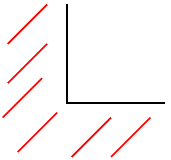}}
    \caption{A $270$-degree corner in the tour of the boundary of $C$.}
    \label{fig: wrong turn}
    \end{center}
\end{figure}
\end{proof}

 Let $B^t,B^b$ be the top and the bottom boundaries of the bounding box $B$, and let $\bset=\set{B^t, B^b}$. 
A rectangle $C\subseteq B$ is called a \emph{potential cell} if and only if there is some subset $\xset\subseteq \opt'\cup \fset$ of rectangles with $\fset\subseteq \xset$, such that $C$ is a cell in the partition $\pset(\xset)$. We next define a set $\dset(C)=\set{R^b,R^t,R^{\ell},R^r}$ of four rectangles of $\opt'\cup \fset\cup \bset$, that we view as defining the cell $C$. If $C$ itself is a rectangle of $\opt'\cup \fset$, then we set $R^b=R^t=R^{\ell}=R^r=C$. 

Assume now that $C\not\in \fset\cup \opt'$.
Notice that all horizontal lines in the partition $\pset(\xset)$ of $B$ are contained in the top or the bottom boundaries of the rectangles in $\xset\cup \bset$. 
Let $R^t\in \xset\cup \bset$ be the rectangle whose bottom boundary contains the top boundary of $C$. We think of the rectangle $R^t$ as defining the top boundary of $C$. Since all rectangles in $\opt'\cup \fset$ are internally disjoint, $R^t$ is uniquely defined. Similarly, let $R^b\in \xset\cup \bset$ be the rectangle whose top boundary contains the bottom boundary of $C$. We view $R^b$ as defining the bottom boundary of $C$. We next define rectangles $R^{\ell}$ and $R^r$, that we view as defining the left and the right boundaries of $C$, respectively. If the top left corner of $C$ is a corner of $R^t$, then we set $R^{\ell}=R^t$. Otherwise, if the bottom left corner of $C$ is a corner of $R^b$, then we set $R^{\ell}=R^b$. Otherwise, there must be at least one rectangle $R\in \xset$, whose right boundary is contained in the left boundary of $C$ (notice that the left boundary of $C$ cannot be contained in the left boundary of $B$, since in that case, the bottom left corner of $R^t$ is the top left corner of $C$). If at least one of the rectangles whose right boundary is contained in the left boundary of $C$ belongs to $\fset$, then we let $R^{\ell}$ be the topmost among all such rectangles $R$. Otherwise, there is exactly one rectangle $R\in \opt'$ (due to non-degeneracy), such that the right boundary of $R$ is contained in the left boundary of $C$. We set $R^{\ell}=R$. We define the rectangle $R^r$ similarly. Let $\dset(C)=\set{R^t,R^b,R^{\ell},R^r}$, so $|\dset(C)|\leq 4$. %We will call $\dset(C)$ \emph{the quadruple defining $C$}, even though $\dset(C)$ may contain fewer than four distinct rectangles. 
Notice that $\dset(C)$ uniquely defines the potential cell $C$: that is, if $C\neq C'$, then $\dset(C)\neq \dset(C')$. Moreover, if $C\in \pset(\xset)$ for some set $\xset\subseteq \opt'\cup \fset$ with $\fset\subseteq \xset$, then all rectangles of $\dset(C)$ must belong to $\xset\cup \bset$. We use this fact later.

We let the initial partition $\pset'$  of the bounding box be $\pset(\wset)$. Our next step is to bound the number of cells in $\pset'$. For each cell $C$ of $\pset'$, if $C\not\in \wset$, then we charge $C$ to the rectangle $R^{\ell}\in \dset(C)$, which must belong to $\wset\cup \bset$. It is easy to see that each rectangle $R\in \wset\cup \bset$ may be charged at most $3$ times: twice for the cells whose left corners coincide with the top left and bottom left corner of $R$, and once for a cell whose left boundary contains the right boundary of $R$, while each of the two rectangles in $\bset$ can be charged once. Therefore, the total number of cells in the partition $\pset'$ is at most $4|\wset|+2$. The expected number of cells in $\pset'$ is therefore at most $4(3r)+2\leq 12r+2$. Moreover, if event $\event_1$ does not happen, $\pset'$ contains at most $72r+2\leq 73r$ cells. %If $r\geq |\opt'|/2$, then we terminate the algorithm and return the partition $\pset'$. It is easy to verify that it is a valid partition. From now on we assume that $r<|\opt'|/2$.

\paragraph{Step 2.}
 Let $C$ be any potential cell, and assume that $C\not\in \fset$. Recall that $N_C$ is the total number of all rectangles in $\opt'$ intersecting $C$. For an integer $t\geq 0$, we say that $C$ has {\it excess $t$} iff $\floor{rN_C/|\opt'|}=t$. 
 We need the following lemma.
 
 \begin{lemma}\label{lemma: r_squared_partition}
 Let $C$ be a cell of $\pset'$, such that $C\not\in\fset$, and assume that $C$ has excess $t\geq 10$. Then there is a partition of $C$ into at most $t^2$ rectangular cells, where for each resulting cell $C'$, $N_{C'}\leq 10|\opt'|/r$.
 \end{lemma}
 
 \begin{proof}
We partition $C$ into at most $t^2$ cells by first building a grid inside $B$,  with at most $t+1$ vertical lines and at most $t+1$ horizontal lines, and then using the partition of $C$ defined by the grid. Let $\trset\subseteq\opt'$ be the set of all rectangles intersecting $C$, so $|\trset|=N_C\geq t|\opt'|/r$. 

We build the set $\tilde{\vset}$ of the vertical lines of the grid as follows. For each $1\leq i< t$, let $V_i$ be the leftmost vertical line, such that the total number of the rectangles of $\trset$ lying completely to the left of $V_i$ is at least $i N_C/t$ (notice that this includes the rectangle whose right boundary lies on $V_i$). Observe that, since $\trset$ is non-degenerate,  the number of the rectangles of $\trset$ lying completely to the left of $V_i$ is at most $\ceil{\frac{iN_C}{t}}\leq \frac{(i+1)N_C}{t}$, as $t\leq \frac{rN_C}{|\opt'|}$, and so $\frac{N_C}{t}\geq\frac{|\opt'|}{r}\geq 2$. Let $V_0$ and $V_t$ be the left and the right boundaries of $C$, respectively. We set $\tilde{\vset}=\set{V_0,V_1,\ldots,V_{t-1},V_t}$. Observe that all lines $V_0,\ldots,V_{t}$ have integral $x$-coordinates. For each consecutive pair $V_{i-1},V_i$ of the vertical lines, consider the rectangle $S^V_i$, whose left and right boundaries are $V_{i-1}$ and $V_i$ respectively, and top and bottom boundaries coincide with the top and the bottom boundaries of $C$. We call $S^V_i$ the \emph{vertical strip of $C$ defined by $V_{i-1}$ and $V_i$}. Then the number of the rectangles $R\in \trset$, that are contained in $S^V_i$ is at most $2N_C/t$. We define the set $\tilde{\hset}=\set{H_0,\ldots,H_r}$ of the horizontal lines of the grid, and the corresponding horizontal strips $S^H_j$ of $C$ similarly.

Consider the partition $\pset^*$ of $C$ defined by the cells of the resulting grid. Then $\pset^*$ contains at most $t^2$ cells. We claim that for each cell $C'\in \pset^*$, $N_{C'}\leq 6|\opt'|/r$. Indeed, consider some cell $C'\in \pset^*$, and let $\trset'$ be the set of all rectangles of $\trset$ intersecting $C'$. 
At most four rectangles of $\trset'$ may contain the corners of the cell $C'$. Each one of the remaining rectangles must be contained in either the vertical strip of the grid to which $C'$ belongs, or the horizontal strip  of the grid to which $C'$ belongs. The total number of the rectangles of $\trset$ contained in these strips is at most $4N_C/t\leq 8|\opt'|/r$, since $t\geq \frac{rN_C}{2|\opt'|}$. Therefore, $N_{C'}\leq 4+8|\opt'|/r\leq 10|\opt'|/r$.
 \end{proof}

 We are now ready to describe the second step of the algorithm. For each cell $C\in \pset'\setminus\fset$, whose excess $t_C\geq 10$, we apply Lemma~\ref{lemma: r_squared_partition} to partition $C$ into $t_C^2$ rectangular cells, where the value $N_{C'}$ of each resulting cell $C'$ is at most $10|\opt'|/r$. We let $\pset$ be the final partition of $B$, obtained after we process all cells $C\in \pset'\setminus\fset$ whose excess is at least $10$. Clearly, the value $N_{C'}$ of each resulting cell $C'$ is at most $20|\opt'|/r$. It now only remains to prove that $\pset$ contains $O(r)$ cells with constant probability. Assuming that event $\event_1$ does not happen, $\pset'$ contains at most $73r$ cells. Therefore, it is enough to show that with probability at least $\half$, the number of cells added in the second step is $O(r)$. The following claim is central in the analysis of the algorithm.

\begin{claim}\label{claim: exponential decay}
For each $t\geq 10$, the expected number of cells in $\pset'$ that have excess $t$ is at most $O(re^{-t})$.
\end{claim}

We prove the claim below, after we complete the proof of Theorem~\ref{thm: r-good partition} here. Let $n_t$ be the number of cells with excess $t$ in the partition $\pset'$. The expected number of the new cells added in the second step is then at most:

\[\sum_{t\geq 10}\expect{n_t}\cdot t^2\leq \sum_{t\geq 10}t^2\cdot O(re^{-t})\leq r\sum_{t\geq 10}O(e^{-t/2})=O(r),\]

since for $t\geq 10$, $t^2<e^{t/2}$. Using Markov's inequality, with probability at least $1/2$, the number of cells added in the second iteration is at most $O(r)$, and so overall, with constant probability, $\pset$ contains at most $O(r)$ cells.
It now only remains to prove Claim~\ref{claim: exponential decay}.

\begin{proofof}{Claim~\ref{claim: exponential decay}}
 The proof of the claim is practically identical to the proof of Lemma 3.3 in~\cite{Har-Peled14}.
 Notice that for each rectangle $R\in \opt'\cup \fset$, $\prob{R\in \sset_1\mid R\in \wset}\geq \half$: if $R\in \fset$, then this probability is $1$; otherwise, it is easy to see that this probability is at least $\half$. For convenience, we denote $\pset(\sset_1)$ by $\pset''$. 

 Consider now some potential cell $C$. We claim that $\prob{C\in \pset''\mid C\in \pset'}\geq \frac 1 {16}$.
 
 Indeed, given a cell $C$, let $\rset_C$ be the set of all rectangles of $\opt'$ intersecting $C$, and let $\event(C)$ be the event that none of the rectangles in $\rset_C$ belong to $\wset$. Then:
 
 \[\begin{split}
 \prob{C\in \pset''\mid C\in \pset'}&=\frac{\prob{(C\in \pset'')\band(C\in \pset')}}{\prob{C\in \pset'}}\\
 &=\frac{\prob{(\dset(C)\subseteq S_1)\band\event(C)}}{\prob{(\dset(C)\subseteq \wset)\band \event(C)}}\\
 &=\frac{\prob{\dset(C)\subseteq S_1}}{\prob{\dset(C)\subseteq \wset}}\\
 &=\frac{\prod_{R\in \dset(C)}\prob{R\in S_1}}{\prod_{R\in \dset(C)}\prob{R\in \wset}}\\
 &=\prod_{R\in \dset(C)}\frac{\prob{R\in S_1}}{\prob{R\in \wset}}\\
 &=\prod_{R\in \dset(C)}\prob{R\in S_1\mid R\in \wset}\geq \frac{1}{16},\end{split}\]
 
 since there are at most four rectangles in $\dset(C)$.
 
 Therefore, $ \prob{C\in \pset'}=\frac{\prob{(C\in \pset'')\band (C\in \pset')}}{\prob{C\in \pset''\mid C\in \pset'}}\leq 16\cdot \prob{(C\in \pset'')\band (C\in \pset')}$.
 Let $\cset_t$ be the set of all potential cells with excess at least $t$. Then:
 
 \[
 \begin{split}
 \expect{|\cset_t\cap \pset'|}&=\sum_{C\in \cset_t}\prob{C\in \pset'}\\
 &\leq \sum_{C\in \cset_t}16\cdot \prob{(C\in \pset'')\band (C\in \pset')}\\
 &=16\sum_{C\in \cset_t}\prob{C\in \pset'\mid C\in \pset''}\cdot \prob{C\in \pset''}.
 \end{split}
 \]
 
 Let $\event'(C)$ be the event that no rectangle $R\in \rset_C$ belongs to $\sset_2$.
 Note that if $C\in \pset''$, then $C$ can only belong to $\pset'$ if event $\event'(C)$ happens. Therefore,
 
 \[\prob{C\in \pset'\mid C\in \pset''}\leq\prod_{R\in \rset_C}\prob{R\not\in \sset_2}\leq \prod_{R\in \rset_C}\left(1-\frac{r}{|\opt'|}\right)\leq e^{-\sum_{R\in\rset_C}r/|\opt'|}=e^{-N_C\cdot r/|\opt'|}\leq  e^{-t},\]
 
  and
 
 \[\expect{|\cset_t\cap \pset'|}\leq 16e^{-t} \sum_{C\in \cset_t}\prob{C\in \pset''}\leq 16e^{-t}\expect{|\pset''|}.\]
 
 From previous discussion, the number of cells in $\pset''$ is bounded by $4|\sset_1|+2$, and $\expect{|S_1|}=2r+m\leq 3r$, so $\expect{|\pset''|}=O(r)$. We conclude that $\expect{|\cset_t\cap \pset'|}\leq O(re^{-t})$.
\end{proofof}

%-------------------------------------------------------
%-------------------------------------------------------

%---------------------------------------------------------------------------
%---------------------------------------------------------------------------
\subsection{Proof of Theorem~\ref{thm: balanced partition reducing boundary complexity}}
%---------------------------------------------------------------------------
%---------------------------------------------------------------------------
%\begin{proof}
%Let $r=\frac{L^2}{2^{14}}$, and let $\opt'$ be any optimal solution of $\rset(\fset)$. We use Theorem~\ref{thm: r-good partition} to find a partition $\pset$ of $B$ into $r$ rectangular cells, where $\fset\subseteq \pset$, each cell of $\pset$ intersects at most $\frac{c^*\cdot |\opt'|}{r}$ rectangles of $\opt'$, and all cells of $\pset$ are aligned with $Z$.
%
%Intuitively, we would like to build a graph $G_1$ whose vertices are the corners of the cells in $\pset$, and edges are defined by the boundaries of the rectangles in $\pset$. We can then compute the dual $G_2$ of this graph, and apply Theorem~\ref{thm: balanced separators in planar graphs} to  $G_2$, to partition $S(\fset)$ into two regions, that will in turn define the two sets $\fset_1,\fset_2$ of fake rectangles. However, in order to do so, we need to ensure that $G_2$ is a 2-connected graph. Notice that this is always the case, since even in the case where we have cells spanning vertically or horizontally through the whole bounding box, the vertex corresponding to the outer face of the primal graph guarantees that each such vertex cannot disconnect the dual graph.
Let $G_1$ be the graph whose vertices are the corners of the rectangles in $\pset$, and edges are the boundary edges of the rectangles in $\pset$. Notice that $G_1$ has at most $|\pset|+1\leq c^*r+1$ faces, including the outer face, and each vertex of $G_1$ is adjacent to at most 4 faces. Let $G_2$ be the graph dual to $G_1$. Then $G_2$ contains at most $c^*r + 1$ vertices, and every face of $G_2$ has at most $4$ vertices on its boundary. Moreover, it is easy to see that $G_2$ is $2$-connected. We assign weights to the vertices of $G_2$ as follows: for each fake rectangle $F\in \fset$, we assign the weight of $1$ to the vertex of $G_2$ that corresponds to the face $F$ of $G_1$. All other vertex weights are set to $0$. The total weight of the vertices of $G_2$ is then exactly $L$. We now use Theorem~\ref{thm: balanced separators in planar graphs} to obtain a simple cycle $C$ of length at most $16\sqrt{c^*r}$ in $G_2$, such that $C$ is a weighted balanced separator. 

Let $v^*$ denote the vertex of $G_2$ corresponding to the outer face of $G_1$. Cycle $C$ partitions the plane into two regions. If $v^*\not\in V(C)$, then we define the exterior of $C$ as the region containing $v^*$, and the other region is called the interior of $C$. Otherwise, we designate one of the two regions as the exterior, and the other as the interior of $C$ arbitrarily.
Let $V_1$ be the set of vertices of $G_2$ with weight $1$ lying in the exterior of $C$, and $V_2$ the set of vertices with weight $1$ lying in the interior of $C$. Then $|V_1|,|V_2|\leq 2L/3$.

 We distinguish between two cases.
The first case happens when $v^*\in V(C)$. Since $C$ is a simple cycle, we visit $v^*$ only once while traversing $C$.  Let $\aset$ be the set of all cells of the partition $\pset$ that correspond to vertices of $V(C)\setminus\set{v^*}$, so $|\aset| \leq O( \sqrt{r})$. Let $v_1, v_2$ be the neighbors of $v^*$ in the cycle $C$. Then the cells corresponding to $v_1$ and $v_2$ must be cells whose boundaries have non-empty intersections with the boundary of the bounding box $B$. We can ``connect" these two cells through a segment $\sigma$ of the bounding box, contained in the interior of $C$.  The outer boundary of $\bigcup_{P\in \aset}P$ combined with the line $\sigma$ must now be a simple polygon, that we denote by $J$.

The second case happens when $v^*\not\in C$. In this case, we let $\aset$ be the set of all the cells of $\pset$ whose corresponding vertex belongs to $C$. Then $|\aset|\leq O(\sqrt r)$, and the outer boundary of $\bigcup_{P\in \aset}P$ is a simple polygon, that we denote by $J$.

In either case, we obtain a simple polygon $J$, whose edges are parallel to the axes, and $J$ is aligned with $Z$. Moreover, since the length of $C$ is at most $O(\sqrt r)$, the boundary of $J$ contains at most $O(\sqrt r)$ corners. We use the boundary of $J$ to define the two sets $\fset_1,\fset_2$ of the fake rectangles, so that $S(\fset_1)=(B\setminus  J)\setminus (\bigcup_{F\in \fset}F)$ and $S(\fset_2)=J\setminus (\bigcup_{F\in \fset}F)$.

From Lemma~\ref{lemma: simple tiling}, there is a set $\fset_1'$ of $O(\sqrt r)$ internally disjoint closed rectangles, whose union is $J$, such that every rectangle of $\fset_1'$ is aligned with $Z$. Let $\fset_1''$ be the set of all rectangles $F\in \fset$ that are contained in $B\setminus J$. Notice that $|\fset_1''|\leq 2L/3$, since we used a balanced cut. We let $\fset_1=\fset_1'\cup \fset_1''$. Then $|\fset_1|\leq \frac{2L}{3}+O(\sqrt r)$. From the above discussion, every rectangle in $\fset_1$ is aligned with $Z$.

We define $\fset_2$ similarly. From  Lemma~\ref{lemma: non-simple-tiling}, there is a set $\fset_2'$ of $O(\sqrt r)$ internally disjoint closed rectangles, whose union is $B\setminus J$, such that every rectangle of $\fset_2'$ is aligned with $Z$. Let $\fset_2''$ be the set of all rectangles $F\in \fset$ that have contributed to the weight of $V_2$, so $|\fset_2''|\leq 2L/3$, and let $\fset_2'''$ be the set of all fake rectangles $F\in \fset$ whose corresponding cell belongs to $\aset$, so $|\fset_2'''|=O(\sqrt r)$. We let $\fset_2=\fset_2'\cup \fset_2''\cup \fset_2'''$. Then $|\fset_2|\leq \frac{2L}{3}+O(\sqrt r)$. As before, every rectangle in $\fset_2$ is aligned with $Z$.

It is easy to see that $S(\fset_1)=(B\setminus J)\setminus(\bigcup_{F\in \fset}F)$, and $S(\fset_2)=J\setminus  (\bigcup_{F\in \fset}F)$. Therefore, $S(\fset_1)\cap S(\fset_2)=\emptyset$, and $S(\fset_1)\cup S(\fset_2)\subseteq S(\fset)$, so $(\fset_1,\fset_2)$ is a valid decomposition pair as required. The only rectangles of $\opt'$ that are not contained in $\rset(\fset_1)\cup \rset(\fset_2)$ are the rectangles that intersect the boundaries of the cells of $\aset$ (observe that line $\sigma$ does not intersect any rectangles).  The number of the rectangles of $\opt'$ that the boundary of each such cell $P$ intersects is at most $N_P\leq O(|\opt'|/r)$, and so the total number of the rectangles of $\opt'$ that do not belong to $\rset(\fset_1)\cup \rset(\fset_2)$ is bounded by $O(|\opt'|/\sqrt r)$.
We conclude that there are constants $c,c'$ with $|\fset_1|,|\fset_2|\leq \frac{2L}{3}+c\sqrt r$, and $|\opt_{\fset_1}|+|\opt_{\fset_2}|\geq |\opt_{\fset}|\left(1-\frac{c'}{\sqrt r}\right )$. Setting $c_1=\max\set{c,c',1}$ concludes the proof of the theorem.

%---------------------------------------------------------------------------
%---------------------------------------------------------------------------
\subsection{Proof of Theorem~\ref{thm: balanced partition general}}
The proof closely follows the proof of Theorem~\ref{thm: balanced partition reducing boundary complexity}, except that we define the weights of the vertices of $G_2$ differently.
%A cell $P\in \pset$ is a \emph{vertical cell} if the top boundary of $P$ lies on the top boundary of the bounding box $B$, and the bottom boundary of $P$ lies on the bottom boundary of $B$. A cell $P\in \pset$ is a \emph{horizontal cell} if the left boundary of $P$ lies on the left boundary of the bounding box $B$, and the right boundary of $P$ lies on the right boundary of $B$. Notice that $\pset$ may contain vertical cells, or horizontal cells, but not both types of cells.

Let $G_1$ be the graph whose vertices are the corners of the rectangles in $\pset$, and edges are the boundary edges of the rectangles in $\pset$. As before, $G_1$ has at most $c^*r+1$ faces, including the outer face, and each vertex of $G_1$ is adjacent to at most four faces. Let $G_2$ be the graph dual to $G_1$. Then $G_2$ contains at most $c^*r + 1$ vertices, and every face of $G_2$ has at most $4$ vertices on its boundary. As before, graph  $G_2$ is $2$-connected. We assign weights to the vertices of $G_2$ as follows. Consider some vertex $v\in V(G_2)$, and let $P$ be the face of $G_1$ corresponding to $v$. If $P$ is the outer face, or $P=F$ for some fake rectangle $F\in \fset$, then we set the weight of $v$ to $0$. Otherwise, the weight of $v$ is the total number of all rectangles $R\in \opt'$, whose upper left corner is either an internal point of $P$, or lies on the left or the top boundaries of $P$, excluding the bottom-left corner and the top-right corner of $P$. The total weight of the vertices of $G_2$ is then exactly $|\opt'|$. We now use Theorem~\ref{thm: balanced separators in planar graphs} to compute a simple cycle $C$ of length at most $16\sqrt{c^*r}$ in $G_2$, such that $C$ is a weighted balanced separator in $G_2$ with respect to the vertex weights. We define the interior and the exterior of $C$, and construct the polygon $J$ exactly as in the proof of Theorem~\ref{thm: balanced partition reducing boundary complexity}. 

We now claim that the total number of the rectangles of $\opt'$ contained in $J$ is at most $3|\opt'|/4$, and the same holds for $B\setminus J$. Since we have computed a balanced partition, the total weight of all vertices lying in the interior of $C$ is at most $2|\opt'|/3$, and the same is true for the total weight of all vertices in the exterior of $C$. Therefore, $B\setminus J$ contains at most $2|\opt'|/3$ rectangles of $\opt'$. Polygon $J$ contains at most $2|\opt'|/3$ rectangles, in addition to some rectangles that may have been added by including the cells of $\aset$ in $J$. However, each cell $P\in \aset$ intersects at most $20|\opt'|/r$ rectangles of $\opt'$, and so the total number of the rectangles of $\opt'$ contained in $J$ is at most $\frac{2|\opt'|}{3}+\frac{20|\opt'|}{r}\cdot 16\sqrt{c^*r}\leq \frac{3|\opt'|}{4}$, as $r\geq 2^{24}c^*$.

We next define the sets $\fset_1$ and $\fset_2$ of fake rectangles exactly like in the proof of Theorem~\ref{thm: balanced partition reducing boundary complexity}, so $S(\fset_1)=(B\setminus J)\setminus(\bigcup_{F\in \fset}F)$, and $S(\fset_2)=J\setminus  (\bigcup_{F\in \fset}F)$. As before,  $S(\fset_1)\cap S(\fset_2)=\emptyset$, and $S(\fset_1)\cup S(\fset_2)\subseteq S(\fset)$, so $(\fset_1,\fset_2)$ is a valid decomposition pair for $\fset$, as required. The only rectangles of $\opt'$ that are not contained in $\rset(\fset_1)$ or $\rset(\fset_2)$ are the rectangles that were intersected by the boundaries of the cells of $\aset$. The total number of the rectangles of $\opt'$ that the boundary of each such cell intersects is $O(|\opt'|/r)$, while $|\aset|=O(\sqrt r)$, and so the total number of the rectangles in $\opt'$ that do not belong to $\rset(\fset_1)\cup \rset(\fset_2)$ is bounded by $O(|\opt'|/\sqrt r)$.

From the above discussion, it is clear that $|\opt_{\fset_1}|, |\opt_{\fset_2}|\leq 3|\opt_{\fset}|/4$. Finally, it remains to bound $|\fset_1|$ and $|\fset_2|$. Both are bounded by $|\fset|+O(z)$, where $z=O(\sqrt r)$ is the number of corners of $J$. 
We conclude that there are constants $c,c'$ with $|\fset_1|,|\fset_2|\leq L+c\sqrt r$, and $|\opt_{\fset_1}|+|\opt_{\fset_2}|\geq |\opt_{\fset}|\left(1-\frac{c'}{\sqrt r}\right )$. Setting $c_2=\max\set{c,c',1}$ concludes the proof of the theorem.

%\end{proof}

%---------------------------------------------------------------------------
%---------------------------------------------------------------------------
\subsection{Proof of Corollary~\ref{corollary: main partition}}
%---------------------------------------------------------------------------
%---------------------------------------------------------------------------
%\begin{proof}
Throughout the proof, we use the constants $c_1,c_2$ from Theorems~\ref{thm: balanced partition reducing boundary complexity} and~\ref{thm: balanced partition general}, and constant $c^*$ from the definition of $r$-good partitions.
We use a parameter $r=\floor{\left(\frac{L^*}{3(c_1+c_2)}\right )^2}$. By appropriately setting the constant $c_3$, we can ensure that $r\geq \max\set{L^*,2^{24}c^*,16c_1^2, 16c_2^2 }$, and that $c_3>6(c_1+c_2)^2$.

Assume first that $L>3$. Let $\opt'$ be any optimal solution to $\rset(\fset)$. We compute an $r$-good partition $\pset$ of $B$ with respect to $\fset$ and $\opt'$, using Theorem~\ref{thm: r-good partition} (notice that we are guaranteed that $r\leq (L^*)^2\leq |\opt_{\fset}|/64$). Let $Z$ be the set of the corners of all rectangles in $\pset$. Recall that if the corners of all rectangles in $\fset$ have integral coordinates, then all points in $Z$ have integral coordinates. We then  apply Theorem~\ref{thm: balanced partition reducing boundary complexity} to $\fset$ and $\pset$, to find a valid decomposition pair $(\fset_1',\fset_2')$ for $\fset$, such that the rectangles in $\fset_1'\cup \fset_2'$ are aligned with $Z$, so all their corners have integral coordinates. 
Moreover,  $|\fset_1'|,|\fset_2'|\leq \frac{2L}{3}+c_1\sqrt r<L^*$, from the definition of $r$. We also have that $|\opt_{\fset_1'}|+|\opt_{\fset_2'}|\geq |\opt_{\fset}|\left (1-\frac{c_1}{\sqrt r}\right )$. Assume without loss of generality that $|\opt_{\fset_1'}|\geq |\opt_{\fset_2'}|$, so $|\opt_{\fset_2'}|\leq |\opt'|/2$.
Then:

\[|\opt_{\fset_1'}|\geq \frac{|\opt_{\fset}|} 2 -\frac{|\opt_{\fset}|\cdot c_1}{2\sqrt r}\geq \frac{|\opt_{\fset}|}{4}\geq 2r,\]

since $r\geq 16c_1^2$, and $r\leq |\opt_{\fset}|/64$.
We now let $\opt''$ be the optimal solution to $\rset(\fset'_1)$. 
 Since $r\geq \max\set{L^*,2^{24}c^*}$, we can compute an $r$-good partition $\pset'$ of $B$ with respect to $\fset_1'$ and $\opt''$, using Theorem~\ref{thm: r-good partition}. Let $Z'$ be the set of the corners of all rectangles in $\pset'$. Then all vertices in $Z'$ have integral coordinates. We then apply Theorem~\ref{thm: balanced partition general} to $\fset_1'$ and $\pset'$, to find a valid decomposition pair $(\fset_1,\fset_2)$ for $\fset'_1$, such that the rectangles in $\fset_1\cup \fset_2$ are aligned with $Z'$, so all their corners have integral coordinates. 

Notice that Theorem~\ref{thm: balanced partition general} guarantees that for each $i\in\set{1,2}$, $|\opt_{\fset_i}|\leq 3|\opt_{\fset}|/4$. Moreover, $|\fset_1|,|\fset_2|\leq |\fset_1'|+c_2\sqrt r\leq \frac{2L^*}{3}+(c_1+c_2)\sqrt r\leq L^*$, from the definition of $r$. %Moreover, $|\opt(\fset_1)|,|\opt(\fset_2)|\leq 3|\opt''|/4\leq 3|\opt'|/4$, while $|\opt(\fset_1)|+|\opt(\fset_2)|\geq |\opt(\fset_1')|\cdot \left (1-\frac{8\cdot c^*}{\sqrt {r}}\right )\geq |\opt(\fset_1')|\cdot \left (1-\frac{2^{12}\cdot c^*}{L^*}\right )$. 

The final three sets of fake rectangles are $\fset_1,\fset_2$ and $\fset_2'$. To finish the proof, we observe that:

\[\begin{split}
|\opt_{\fset_1}|+|\opt_{\fset_2}|+|\opt_{\fset_2'}|&\geq  |\opt_{\fset_1'}|\cdot \left (1-\frac{c_2}{\sqrt r}\right )+|\opt_{\fset_2'}|\\
&\geq \left (|\opt_{\fset_1'}|+|\opt_{\fset_2'}|\right )\cdot  \left (1-\frac{c_2}{\sqrt r}\right )\\
&\geq |\opt_{\fset}|\cdot  \left (1-\frac{c_1}{\sqrt r}\right )\left (1-\frac{c_2}{\sqrt r}\right ) \\
&\geq |\opt_{\fset}|\cdot  \left (1-\frac{c_1+c_2}{\sqrt r}\right )\\
&\geq |\opt_{\fset}|\cdot \left (1-\frac{c_3}{L^*}\right ),\end{split}\]

from the definition of $r$, and since $c_3>6(c_1+c_2)^2$.

Assume now that $L<3$. Then instead of running the first part of the above algorithm, we simply set $\fset_1'=\fset$ and $\fset_2'=\set{B}$. We then apply the second part of the algorithm to $\fset_1'$ exactly as before (in other words, we only need to apply Theorem~\ref{thm: balanced partition general}).
%\end{proof}

%-------------------------------------------------------
%-------------------------------------------------------
%-------------------------------------------------------
%-------------------------------------------------------
%-------------------------------------------------------
%-------------------------------------------------------
\section{Proofs Omitted from Section~\ref{sec: running time exp sqrt log n}}
%---------------------------------------------
%---------------------------------------------
%---------------------------------------------
\subsection{Proof of Claim~\ref{claim: find a rho-accurate grid}}
%---------------------------------------------
%---------------------------------------------
%---------------------------------------------

Let $\tau$ be a parameter that we will determine later. Let $\opt'$ be an optimal solution to instance $\rset(\fset)$ (note that the algorithm does not know this solution or its value).
Given $\tau$, we define a grid $G_{\tau}=(\vset^{\tau},\hset^{\tau})$ as follows. The set $\vset^{\tau}$ of vertical lines of the grid is the union of two subsets, $\vset_1^{\tau}$ and $\vset_2^{\tau}$. Set $\vset_2^{\tau}$ contains, for every corner $(x,y)$ of every rectangle $F\in \fset$, a vertical line $V$ with coordinate $x$. Therefore, $|\vset_2^{\tau}|\leq 2|\fset|$. We now proceed to define $\vset_1^{\tau}$.

Initially, $\vset_1^{\tau}$ only contains $V_0$ - the left boundary of $B$. We set $t=0$, and then iterate. For each vertical line $V$ to the right of $V_t$, consider the vertical strip $S(V)$ of $B$, contained between $V_t$ and $V$. Let $\rset'\subseteq \rset(\fset)$ denote the set of all rectangles that are contained in this strip. We run the $(c_A\log \log |\opt'|)$-approximation algorithm from Corollary~\ref{cor: an log log opt approx} on $\rset'$. If the value of the solution returned by the algorithm is at least $\tau/2$, then $V$ is called a candidate line. Among all candidate lines, let $V_{t+1}$ be the leftmost one. We can assume without loss of generality, that there is a rectangle $R\in \rset$ whose right boundary lies on $V_{t+1}$ (alternatively, we can shift the line $V_{t+1}$ to the left until this happens; since both vertical strips defined by the old and by the new locations of $V_{t+1}$ contain exactly the same set of rectangles, the values of the solutions returned by the $O(\log\log |\opt'|)$-approximation algorithm are the same for both instances).  We add $V_{t+1}$ to $\vset_1^{\tau}$, set $t:=t+1$, and continue. We use the following observation.

\begin{observation}\label{obs: bound in strips}
Let $\rset'\subseteq \rset(\fset)$ contain all rectangles $R$ with $R\subseteq S(V_{t+1})$. Then the value of the optimal solution to instance $\rset'$ is at most $\ceil{\frac{\tau \cdot c_A\log\log |\opt'|} 2}$.
\end{observation}
\begin{proof}
Let $V'$ be a vertical line lying immediately to the left of $V_{t+1}$, and let $\rset''\subseteq \rset(\fset)$ contain all rectangles $R$ with $R\subseteq S(V')$. Since $V'$ is not a candidate line, the algorithm from Corollary~\ref{cor: an log log opt approx} returned a solution to instance $\rset''$ whose value is less than  $\tau/2$, and so the optimal solution value for instance $\rset''$ is less than $\frac{\tau \cdot c_A\log\log |\opt'|} 2$. Since the input set of rectangles are non-degenerate, $|\rset'\setminus\rset''|\leq 1$, and so the optimal solution value for instance $\rset'$ is at most $\ceil{\frac{\tau \cdot c_A\log\log |\opt'|} 2}$.
\end{proof}

The algorithm terminates in iteration $t$ when no candidate lines exist anymore. The last strip $S$ of the grid must also have the property that if $\rset'\subseteq \rset(\fset)$ contains all rectangles $R$ with $R\subseteq S$, then the value of the optimal solution to instance $\rset'$ is at most $\ceil{\frac{\tau \cdot c_A\log\log |\opt'|} 2}$.

We set $\vset^{\tau}=\vset_1^{\tau}\cup \vset_2^{\tau}$, and we define the set $\hset^{\tau}$ of horizontal lines of grid $G_{\tau}$ similarly. 
This concludes the definition of the grid $G_{\tau}$.
Since we have ensured that for each vertical line $V\in \vset^{\tau}$, there is some rectangle in $\rset\cup\fset$ whose left or right boundary is contained in $V$, if all corners of all rectangles in $\rset\cup\fset$ have integral coordinates, so do all lines in $\vset^{\tau}$. A similar argument holds for the lines in $\hset^{\tau}$. Therefore, if the corners of all rectangles in $\fset\cup \rset$ have integral coordinates, then so do all vertices of the grid.

We are now ready to complete the algorithm. Let $w$ be a large enough constant, so that $\log w>2\log(c_A\log\log w)$. Notice that we can efficiently check whether $|\opt'|\geq w$ via exhaustive search. Assume first that $|\opt'|\geq w$ holds. Our algorithm needs an estimate on the value $|\opt'|$. In order to obtain this estimate,  we run the $(c_A\log \log |\opt'|)$-approximation algorithm from Corollary~\ref{cor: an log log opt approx} on instance $\rset(\fset)$, and denote by $W'$ the value of the solution returned by the algorithm. Then  $\frac{|\opt'|}{c_A\log\log |\opt'|}\leq W'\leq |\opt'|$. 
Since we have assumed that $|\opt'|\geq w$, we get that $\log |\opt'|>2\log(c_A\log\log |\opt'|)$, and so $\frac{\log\log |\opt'|}2\leq \log\log W'\leq \log\log |\opt'|$. Therefore, $W'\leq |\opt'|\leq c_AW'\log\log |\opt'|\leq 2c_AW'\log\log W'$.

We will use a new integral parameter $W$ as our guess on the value of $|\opt'|$, trying all integral values of $W$ between $W'$ and $2c_A W'\log\log W'$, so one of the values we try is guaranteed to be $|\opt'|$. For each guessed value of $W$, we let $\tau_W=\frac{W}{\rho\cdot c_A\log\log W}$, and we construct the grid $G_{\tau_{W}}$ as above. From Observation~\ref{obs: bound in strips}, we are guaranteed that for every vertical and every horizontal strip of the resulting grid, the value of the optimal solution to the sub-instance defined by the strip is at most:

 \[\ceil{\frac{\tau_W\cdot c_A\log\log |\opt'|}{2}}=\ceil{\frac{W\log\log |\opt'|}{2\rho \log\log W}}\leq \ceil{\frac{W}{\rho}\cdot \frac{\log\log |\opt'|}{2\log\log W'}}\leq \ceil{\frac{W}{\rho}}.\]

 In particular, if $W\leq |\opt'|$, then this value is bounded by $\ceil{|\opt'|/\rho}$. On the other hand, if $W=|\opt'|$, then, since every vertical strip defined by the lines in $\vset_1^{\tau_W}$ contains at least $\frac{\tau_W}{2}\geq\frac{|\opt'|}{2\rho c_A\log\log |\opt'|}$ mutually disjoint rectangles, the  number of such strips is bounded by $2\rho c_A\log\log |\opt'|\leq 4\rho c_A\log \log W'$, and the same holds for the number of the horizontal strips defined by the lines in $\hset_1^{\tau_W}$.
Therefore, for $W=|\opt'|$, we get $|\vset^{\tset_W}_1|,  |\hset^{\tset_W}_1|\leq 4\rho c_A\log\log W'$.
We choose the smallest value $W$ for which the $|\vset^{\tset_W}_1|,  |\hset^{\tset_W}_1|\leq 4\rho c_A\log\log W'$ and return the corresponding grid $G_{\tau_W}$ as the output of the algorithm. We are then guaranteed that $W\leq |\opt'|$, and  $|\vset^{\tau_W}|\leq 4\rho c_A\log\log W+2|\fset|\leq 4\rho c_A\log\log |\opt'|+2|\fset|$, and the same holds for $|\hset^{\tau_W}|$. From the above discussion, for every vertical and every horizontal strip of the grid, the optimal solution value for the instance defined by the rectangles contained in the grid is at most $\ceil{|\opt'|/\rho}$.

Finally,  consider the case when $|\opt'|<w$. In this case we can compute the optimal solution $\opt'$ exactly via exhaustive search. We  then use the grid $G^{\tau}$ for $\tau=\frac{|\opt'|}{\rho \cdot c_A\log\log |\opt'|}$. From Observation~\ref{obs: bound in strips}, for every vertical and every horizontal strip of the grid, the optimal solution value for the instance defined by the rectangles contained in the strip is at most:

\[\ceil{\frac{\tau \cdot c_A\log\log |\opt'|} 2}=\ceil{\frac{|\opt'|}{2\rho}}.\]

Since every strip of the bounding box defined by the vertical lines in $\vset_1^{\tau}$ contains at least $\tau/2\geq |\opt'|/(2\rho c_A\log\log |\opt'|)$ rectangles that are disjoint from each other, $|\vset_1^{\tau}|\leq 2\rho c_A\log\log |\opt'|$, and the same holds for $|\hset_1^{\tau}|$.

%\end{proof}

%---------------------------------------------
%---------------------------------------------
%---------------------------------------------
\subsection{Proof of Theorem~\ref{thm: rho-aligned r-good partition}}
%---------------------------------------------
%---------------------------------------------
%---------------------------------------------

\begin{proof}
Let $r'=8r$, and notice that $\max\set{m,3}\leq r'\leq |\opt'|/2$. Let $\pset$ be the $r'$-good partition of $B$ with respect to $\fset$ and $\opt'$, given by Theorem~\ref{thm: r-good partition} (this partition is not necessarily aligned with $G$). We gradually transform $\pset$ into a $G$-aligned $r$-good partition. Recall that $|\pset|\leq c^*r'$, and every cell of $\pset$ intersects at most $20|\opt'|/r'$ rectangles of $\opt'$.

Let $P$ be any cell of $\pset$, such that $P$ is not contained in any vertical strip of the grid $G$. Let $V_L$ be the leftmost vertical line of $G$ intersecting $P$, and let $V_R$ be the rightmost vertical line of $G$ intersecting $P$ (it is possible that $V_L=V_R$). We partition $P$ into up to three cells along the lines $V_R$ and $V_L$ (if $V_R=V_L$, then we partition into two cells). At the end of this procedure, for each cell $P$, either $P$ is contained in a vertical strip, or the left and the right boundaries of $C$ are aligned with the grid $G$. Nottice that cells $P\in \fset$ are not changed at the end of this step, as each such cell is aligned with $G$.

We do the same with the horizontal lines of $G$. Let $\pset'$ be the resulting partition. Then $|\pset'|\leq 9r'$, and every cell of $\pset'$ intersects at most $20|\opt'|/r'$ rectangles of $\opt'$. Moreover, every rectangle of $\fset$ is a cell of $\pset'$. We say that a cell $P'$ of $\pset'$ is \emph{small} if $P'\not\in \fset$, and it is contained in a vertical or a horizontal strip of $G$; otherwise, it is large. Notice that for every large cell $P'$ of $\pset'$, the boundary of $P'$ is aligned with the grid $G$. 

We now construct our final partition $\pset^*$ of the bounding box into cells, by starting with $\pset^*=\emptyset$, and then gradually adding rectangular cells to $\pset^*$. All cells that we add will be internally disjoint and aligned with the grid $G$. At the same time, we will consider the grid $G$, and we will mark the cells of $G$ that are covered by the rectangles currently in $\pset^*$. First, we add to $\pset^*$ all large cells that belong to $\pset'$, and we mark all cells of $G$ that are covered by such large cells of $\pset'$. (At this point, all fake rectangles in $\fset$ belong to $\pset^*$).

Let $C$ be any unmarked cell of the grid $G$. If there is some cell $P'\in \pset'$, such that one of the four corners of $P'$ belongs to $C$ or its boundary, then we add $C$ to $\pset^*$, and we mark cell $C$ of $G$. We call such a cell $C$ a {\it neutral cell}. We note that for each rectangle in $\pset'$, we add at most 16 new rectangles to $\pset^*$ at this step, as each of the four corners of each rectangle $P'\in \pset'$ may intersect a boundary of up to $4$ cells of $G$. Therefore, so far, $|\pset^*|\leq 16\cdot 9\cdot c^* r'=144 c^*r'$.

Let $C$ be any cell of $G$ that remains unmarked after the previous step. Since $C$ is not contained in any large rectangle of $\pset'$, there must be a small rectangle $P'\in \pset'$ intersecting $C$. Moreover, no corner of $P'$ is contained in $C$ or its boundary. Let $S^V$ be the vertical strip of $G$ in which $C$ lies, and let $S^H$ be the horizontal strip of $G$ in which $C$ lies. Then $P'$ is contained in $S^V$ or in $S^H$, and exactly one of these two things must happen, as no corner of $P'$ is contained in $C$. In the former case, we say that $P'$ intersects $C$ vertically, and in the latter case we say that it intersects $C$ horizontally. Notice that  If $P'$ intersects $C$ vertically, then all other small cells of $\pset'$ intersecting $C$ must intersect it vertically. We say that $C$ is a vertical cell of $G$ in such a case. Otherwise, we say that $C$ is a horizontal cell.

While $G$ contains an unmarked vertical cell $C$, let $S^V$ be the vertical strip of $G$ containing $C$. Let $\sset$ be the set of all vertical cells contained in $S^V$ (notice that $\sset$ only contains cells that are currently unmarked). Let $\sset'\subseteq \sset$ be a subset of these vertical cells, that appear consecutively in the strip $S^V$, such that $C\in \sset'$, and $\sset'$ is maximal satisfying those conditions. Let $P^*$ be the rectangle consisting of the union of the cells in $\sset'$. We add $P^*$ to $\pset^*$, and we mark all cells in $\sset'$. We call $P^*$ a vertical cell of $\pset^*$. Notice that the cell $C^*$ of $G$, appearing right below $P^*$, must be a neutral cell (since the small cells of $\pset'$ intersecting the bottommost cell of $\sset'$ must terminate at $C^*$ or its boundary). We charge $P^*$ to $C^*$.

We continue to process all unmarked vertical cells of $G$, until all of them are marked. Each neutral cell is charged at most once, so the number of cells in $\pset^*$ at most doubles.

We process all unmarked horizontal cells of $G$ similarly. In the end, we obtain a partition $\pset^*$ of the bounding box into rectangles, containing at most $4\cdot 144 \cdot c^* r'=O(r)$ cells. As observed above, $\fset\subseteq \pset^*$. It now only remains to show that every cell of $\pset^*$ intersects at most $20|\opt'|/r$ rectangles of $\opt'$.

\begin{lemma}
Each cell of $\pset^*$ intersects at most $20|\opt'|/r$ rectangles of $\opt'$.
\end{lemma}

\begin{proof}
 Let $P\in \pset^*$ be any such cell. Recall that every cell in $\pset'$ intersects at most $20|\opt'|/r'$ rectangles of $\opt'$, and that $r'=8r$.

If $P$ is a large cell, then $P\in \pset'$, so it intersects at most $20|\opt'|/r'\leq 20|\opt'|/r$ rectangles of $\opt'$. Assume now that $P$ is a neutral cell. Let $S^V$ and $S^H$ be the vertical and the horizontal strips of $G$ containing $P$. At most four rectangles of $\opt'$ contain the corners of $P$, and the remaining rectangles intersecting $P$ are contained in either $S^V$ or $S^H$. Since grid $G$ is $\rho$-accurate, and $ r\leq \rho$, at most $\ceil{\frac{|\opt'|}{\rho}}\leq \ceil{\frac{|\opt'|} r}$ rectangles of $\opt'$ can be contained in $S^V$, and the same holds for $S^H$. So $P$ intersects at most $\ceil{\frac{2|\opt'|}{r}}+4\leq \frac{20|\opt'|}{r}$ rectangles of $\opt'$. It now only remains to analyze the case where $P$ is a vertical or a horizontal cell. We only analyze the former; the latter case is analyzed similarly.

In order to be consistent with prior notation, we let $P^*$ be any vertical cell of $\pset^*$, and let $C$ be the vertical cell of $G$, such that $P^*$ was added to $\pset^*$ when $C$ was processed. Let $S^V$ be the vertical strip of $G$ containing $P^*$, and let $A,B$ be the left and the right boundaries of $P^*$, respectively.

The rectangles of $\opt'$ intersecting $P^*$ can be partitioned into three subsets: set $\xset_1$ containing all rectangles $R\subseteq S^V$ - their number is bounded by $\ceil{\frac{|\opt'|}{\rho}}\leq\ceil{\frac{|\opt'|}{r}}$ since $G$ is $\rho$-accurate and $r\leq \rho$; set $\xset_2$ containing all rectangles $R$ that intersect $A$; and set $\xset_3$ containing all rectangles that intersect $B$. We now show that $|\xset_2|\leq \frac{20|\opt'|}{r'}$, and $|\xset_3|$ is bounded similarly.

In order to show that $|\xset_2|\leq \frac{20|\opt'|}{r'}$, we show that there is some cell $P'\in \pset'$, such that $P'$ contains $A$, in the following claim.

\begin{claim}\label{claim: bound intersection}
There is some cell $P'\in \pset'$, such that $P'$ contains $A$ (possibly as part of its boundary).
\end{claim}

Notice that every rectangle in $\xset_2$ intersects $A$, and so, from the above claim, it also intersects $P'$. Since each cell in $\pset'$ intersects at most $\frac{20|\opt'|}{r'}$ rectangles of $\opt'$, we get that $|\xset_2|\leq \frac{20|\opt'|}{r'}$. A similar analysis shows that $|\xset_3|\leq \frac{20|\opt'|}{r'}$. Since $r'=8r$, the total number of rectangles intersecting $P^*$ is bounded by:
$\frac{40|\opt'|}{r'}+\ceil{\frac{|\opt'|}{r}}<\frac{20|\opt'|}{r}.$
It now only remains to prove Claim~\ref{claim: bound intersection}.

\begin{proof}
The intuition for the proof is that $P^*$ cannot contain a corner of any cell of $\pset'$: otherwise, one of the cells of $\sset'$ should have been neutral. Since $\pset'$ partitions the bounding box, $A$ must be contained in some cell of $\pset'$. We now give a formal proof.

Let $v$ be any point on the left boundary of the cell $C$. Then some cell $P'$ of $\pset'$ must contain $v$ (where possibly $v$ belongs to the boundary of $P'$). If $v$ lies on the right boundary of $P'$, then there is some other cell $P''\in \pset'$, such that $v$ lies on the left boundary of $P''$. We then replace $P'$ with $P''$. Therefore, we assume that $v$ does not lie on the right boundary of $P'$.% We claim that $A$ is contained in $C'$: otherwise, either the top right or the bottom right corner of $C'$ is contained in $R$, so one of the cells of $\sset'$ should have been neutral, a contradiction.

Assume first that $P'$ is a large cell. Then its four sides are aligned with the lines of the grid. Since $v$ does not lie on the right boundary of $P'$, it is easy to see that $C\subseteq P'$ must hold, which is impossible, since $C$ was unmarked when it was processed.% So $v$ must lie on the boundary of $P'$. Moreover, it lies on the right boundary of $P'$, since $C$ is not contained in $P'$. We claim that $A$ must be contained in the boundary of $P'$: otherwise, either the top right or the bottom right corner of $P'$ is contained in $A$, and one of the cells of $\sset'$ should have been neutral. Therefore, if $P'$ is a large cell, then $A$ is contained in the boundary of $P'$.

Assume now that $P'$ is a small cell. Since cell $C$ is not neutral, and it is vertical, $P'$ must intersect $C$ vertically. We claim that $A$ is contained in the boundary of $P'$: otherwise, the one of the corners of $P'$ is contained in $P^*$, and so one of the cells of $\sset'$ should have been neutral.
\end{proof}
\end{proof}
\end{proof}

%---------------------------------------------
%---------------------------------------------
%---------------------------------------------
\subsection{Proof of Corollary~\ref{corollary: main partition with grid}}
%---------------------------------------------
%---------------------------------------------
%---------------------------------------------

Throughout the proof, we use the constants $c_1,c_2$ from Theorems~\ref{thm: balanced partition reducing boundary complexity} and~\ref{thm: balanced partition general}, and constant $c^{**}$ from the definition of $r$-good grid-aligned partitions.
We use a parameter $r=\floor{\left(\frac{L^*}{12(c_1+c_2)}\right )^2}$. By appropriately setting the constant $\tilde c$, we can ensure that $r\geq \max\set{L^*,2^{24}c^{**},16c_1^2, 16c_2^2 }$, and that $\tc>6(c_1+c_2)^2$.

%We use a parameter $r=(L^*)^2$. By appropriately setting the constant $\tc$, we can ensure that $r\geq \max\set{L^*,2^{24}c^{**},4(c_1')^2 }$. 

Assume first that $L>3$, and  let $\opt'$ be any optimal solution to $\rset(\fset)$.
Notice that we also are guaranteed that $\max\set{L^*,3}\leq r\leq \min\set{\rho,|\opt'|/16}$, and so we can apply Theorem~\ref{thm: rho-aligned r-good partition} to compute a $G$-aligned $r$-good partition $\pset$ of $B$ with respect to $\fset$ and $\opt'$. Let $Z$ be the set of the corners of all rectangles in $\pset$. Since grid $G$ is $\rho$-accurate for $\fset$, all rectangles in $\fset$ are aligned with the grid $G$, and so are the points of $Z$. We then  apply Theorem~\ref{thm: balanced partition reducing boundary complexity} to $\fset$ and $\pset$, to compute a valid decomposition pair $(\fset_1',\fset_2')$ for $\fset$, such that the rectangles in $\fset_1'\cup \fset_2'$ are aligned with $Z$, and hence with $G$. 
Moreover, we are guaranteed that  $|\fset_1'|,|\fset_2'|\leq \frac{2L}{3}+c_1\sqrt r\leq \frac{3L^*}{4}$ from the choice of $r$, and $|\opt_{\fset_1'}|+|\opt_{\fset_2'}|\geq |\opt_{\fset}|\left (1-\frac{c_1}{\sqrt r}\right )$. Assume w.l.o.g. that $|\opt_{\fset_1'}|\geq |\opt_{\fset_2'}|$, so $|\opt_{\fset_2'}|\leq |\opt'|/2$.
Then:

\[|\opt_{\fset_1'}|\geq \frac{|\opt_{\fset}|} 2 -\frac{|\opt_{\fset}|\cdot c_1}{2\sqrt r}\geq \frac{|\opt_{\fset}|}{4},\]

since $r\geq 16c_1^2$. Moreover, from Observation~\ref{obs: rho accurate for big subinstances}, grid $G$ remains $\rho/4$-accurate for $\fset_1'$. Let $\rho'=\rho/4$, and recall that $r\leq \rho'$. We now let $\opt''$ be the optimal solution to $\rset(\fset'_1)$, so $|\opt''|\geq|\opt_{\fset}|/4>16$. Since $\max\set{L^*,3}\leq r\leq \min\set{\rho',|\opt''|/16}$, we can again apply Theorem~\ref{thm: rho-aligned r-good partition} to compute a $G$-aligned $r$-good partition $\pset'$ of $B$ with respect to $\fset'_1$ and $\opt'$.
 Let $Z'$ be the set of the corners of all rectangles in $\pset'$. Then all vertices in $Z'$ are aligned with $G$. We then apply Theorem~\ref{thm: balanced partition general} to $\fset_1'$ and $\pset'$, to compute a valid decomposition pair $(\fset_1,\fset_2)$ for $\fset'_1$, such that the rectangles in $\fset_1\cup \fset_2$ are aligned with $Z'$, and thus with $G$. 

Notice that Theorem~\ref{thm: balanced partition general} guarantees that for each $i\in\set{1,2}$, $|\opt_{\fset_i}|\leq 3|\opt_{\fset}|/4$. Moreover, $|\fset_1|,|\fset_2|\leq |\fset_1'|+c_2\sqrt r\leq \frac{2L^*}{3}+(c_1+c_2)\sqrt r\leq 3L^*/4$, from the definition of $r$. %Moreover, $|\opt(\fset_1)|,|\opt(\fset_2)|\leq 3|\opt''|/4\leq 3|\opt'|/4$, while $|\opt(\fset_1)|+|\opt(\fset_2)|\geq |\opt(\fset_1')|\cdot \left (1-\frac{8\cdot c^*}{\sqrt {r}}\right )\geq |\opt(\fset_1')|\cdot \left (1-\frac{2^{12}\cdot c^*}{L^*}\right )$. 

The final three sets of fake rectangles are $\fset_1,\fset_2$ and $\fset_2'$. To finish the proof, we observe that:

\[\begin{split}
|\opt_{\fset_1}|+|\opt_{\fset_2}|+|\opt_{\fset_2'}|&\geq  |\opt_{\fset_1'}|\cdot \left (1-\frac{c_2}{\sqrt r}\right )+|\opt_{\fset_2'}|\\
&\geq \left (|\opt_{\fset_1'}|+|\opt_{\fset_2'}|\right )\cdot  \left (1-\frac{c_2}{\sqrt r}\right )\\
&\geq |\opt_{\fset}|\cdot  \left (1-\frac{c_1}{\sqrt r}\right )\left (1-\frac{c_2}{\sqrt r}\right ) \\
&\geq |\opt_{\fset}|\cdot  \left (1-\frac{c_1+c_2}{\sqrt r}\right )\\
&\geq |\opt_{\fset}|\cdot \left (1-\frac{\tc}{L^*}\right ),\end{split}\]

since $\tc>6(c_1+c_2)^2$.

Assume now that $L<3$. Then instead of running the first part of the above algorithm, we simply set $\fset_1'=\fset$ and $\fset_2'=\set{B}$. We then apply the second part of the algorithm to $\fset_1'$ exactly as before (in other words, we only need to apply Theorem~\ref{thm: balanced partition general}).

\end{document}